\documentclass[pdflatex,sn-mathphys-num]{sn-jnl}

\usepackage{graphicx}%
\usepackage{multirow}%
\usepackage{amsmath,amssymb,amsfonts}%
\usepackage{amsthm}%
\usepackage{mathrsfs}%
\usepackage[title]{appendix}%
\usepackage{textcomp}%
\usepackage{manyfoot}%
\usepackage{booktabs}%
\usepackage{algorithm}%
\usepackage{algorithmicx}%
\usepackage{algpseudocode}%
\usepackage{listings}%

\usepackage{subfigure}

\theoremstyle{thmstyleone}%
\newtheorem{theorem}{Theorem}
\newtheorem{proposition}[theorem]{Proposition}%
\newtheorem{corollary}[theorem]{Corollary}
\newtheorem{lemma}[theorem]{Lemma}
\theoremstyle{remark}
\newtheorem{claim}{Claim}
\theoremstyle{thmstyletwo}%

\theoremstyle{thmstylethree}%
\newtheorem{definition}{Definition}%

\theoremstyle{definition}
\newtheorem{assumption}{Assumption}

\usepackage[table]{xcolor}
\usepackage{booktabs}
\usepackage{multirow}
\usepackage{multicol}
\usepackage{adjustbox}

\usepackage{bm}
\usepackage{bbold}
\usepackage{listings}

\newcommand{\note}[1]{{\textcolor{black}{{#1}}}}
\newcommand{\newchange}[1]{{\textcolor{black}{{#1}}}}
\newcommand{\newrevision}[1]{{\textcolor{black}{{#1}}}}

\newcommand{\commentout}[1]{%
}

\begin{document}

\title[Inferring Individual Direct Causal Effects Under Heterogeneous Peer Influence]{Inferring Individual Direct Causal Effects Under Heterogeneous Peer Influence}


\author*[1]{\fnm{Shishir} \sur{Adhikari}}\email{sadhik9@uic.edu}

\author*[1]{\fnm{Elena} \sur{Zheleva}}\email{ezheleva@uic.edu}


\affil*[1]{\orgdiv{Department of Computer Science}, \orgname{University of Illinois Chicago}, \orgaddress{\city{Chicago}, \postcode{60608}, \state{Illinois}, \country{USA}}}




\abstract{
Causal inference is crucial for understanding the effectiveness of policies and designing personalized interventions. Causal inference involves estimating the causal effects of treatments on outcomes of interest after modeling appropriate assumptions.
Most causal inference approaches assume that a unit's outcome is independent of the treatments or outcomes of other units. However, this assumption is unrealistic when inferring causal effects in networks where a unit's outcome can be influenced by the treatments and outcomes of its neighboring nodes, a phenomenon known as interference. Causal inference in networks should explicitly account for interference. In interference settings, the direct causal effect measures the impact of the unit's own treatment while controlling for the treatments of peers. Existing solutions to estimating direct causal effects under interference consider either homogeneous influence from peers or specific heterogeneous influence mechanisms (e.g., based on local neighborhood structure). In this work, we define \textit{heterogeneous peer influence} (HPI) as the general interference that occurs when a unit's outcome may be influenced differently by different peers based on their attributes and relationships, or when each network node may have a different susceptibility to peer influence. This paper presents IDE-Net, a framework for estimating individual, i.e., unit-level, direct causal effects in the presence of HPI where the mechanism of influence is not known a priori. We first propose a structural causal model for networks that can capture different possible assumptions about network structure, interference conditions, and causal dependence and that enables reasoning about causal effect identifiability and discovery of potential heterogeneous contexts. We then propose a novel graph neural network-based estimator to estimate individual direct causal effects. We show empirically that state-of-the-art methods for individual direct effect estimation produce biased results in the presence of HPI, and that our proposed estimator is robust.}




\maketitle

\vspace{-1.5em}
\section{Introduction}

Causal inference is pivotal in informed decision-making across various domains, including vaccination distribution~\cite{barkley-aas20}, policy making~\cite{patacchini-jcbo17}, and online advertising~\cite{nabi-frontiers22}. {The goal of causal inference is to estimate the causal effect of treatments on outcomes of interest, and thus help with improving policy effectiveness and targeted interventions.}
When the outcome of a unit can be caused not only by its own treatment but also by the treatments of others, this outcome \emph{interference} has to be accounted for in causal effect estimation~\cite{leavitt-jasp51,halloran-ep95,ugander-kdd13,aral-ohen16}. 
Three main causal effects of interest in interference settings are \textit{direct effects} induced by a unit's own treatment, \textit{peer (or indirect) effects} induced by the treatments of other units, and \textit{total effects} induced by both unit's and others' treatments~\cite{hudgens-jasa08}. These causal effects are measured on population-level as average effects and on unit level as individual (or heterogeneous) effects.

In the presence of interference, direct effect estimation, which is the focus of this work, needs to take into consideration the extent to which peer treatments influence the unit's outcome and separate that influence from the influence of the unit's treatment itself. There is a subtle but important difference between peer effect estimation and direct effect estimation under peer influence. Peer effect estimation considers the difference in a unit's outcome under two different treatment regimes of other units (e.g., all other units being treated and all other units not being treated). Direct effect estimation under peer influence considers the difference in a unit's outcome between two different treatments of the unit itself while \emph{controlling for} the treatment and attributes of other units where the treatment regime of other units is fixed.

Peer influence is modeled through \emph{exposure mapping}~\cite{aronow-aas17}, which is a function that maps peer treatments and other contexts to a representation referred to as \textit{peer exposure} that summarizes exposure to peer treatments, reduces high dimensionality, and is invariant to irrelevant contexts (e.g., permutation).
Most methods assume homogeneous (or equal) influence from each peer and use exposure mapping like the fraction of treated peers. 
However, \textit{heterogeneous peer influence} (HPI) can occur when the influence of each peer on a unit varies based on the unit's and peer traits, relationship characteristics, and network properties. Recent methods focus on HPI due to specific contexts, such as local neighborhood structure~\cite{yuan-www21}, known node attributes (e.g., roles like parent-child or siblings)~\cite{qu-arxiv21} or node types (e.g., customer and product)~\cite{lin-ecmlkdd23}, known edge weights~\cite{forastiere-jasa21}, node attribute similarity~\cite{zhao-arxiv22}, and group interactions~\cite{ma-kdd22}. 


\begin{figure*}[!t]
    \centering
    \includegraphics[width=\linewidth]{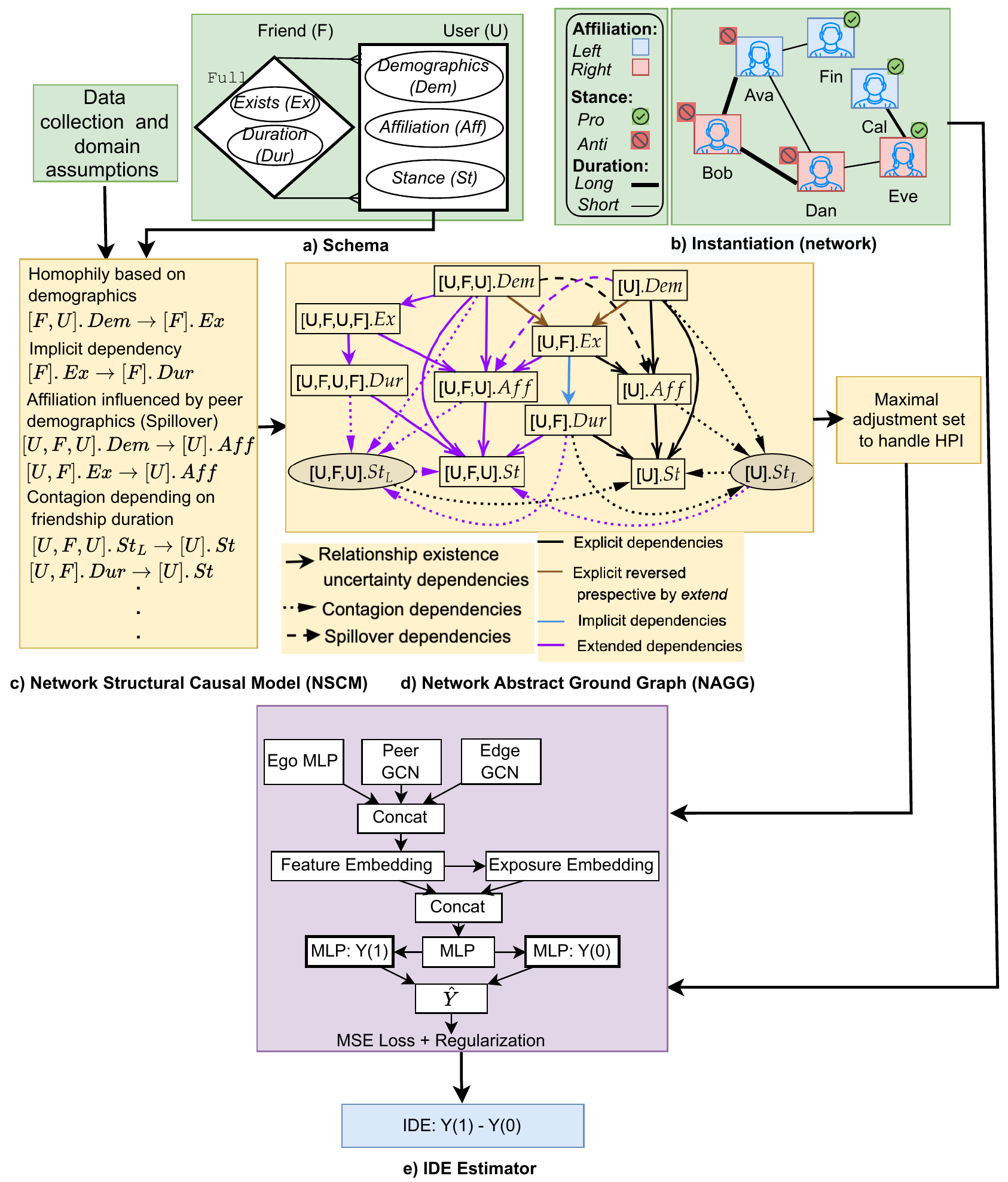}
    \vspace{-1.5em}
    \caption{{Overview of the IDE-Net framework through an example query of whether a user's political affiliation (X) affects the user's stance (Y). NSCM\newrevision{, like structural equation model,} encodes the domain and data assumptions \newrevision{(e.g., homophily and selection bias)} and NAGG\newrevision{, like a graphical model,} enables causal reasoning to find a maximal adjustment set of relational variables, if the IDE is identifiable. The IDE Estimator uses a GNN-based feature representation of the maximal adjustment set and network features and it learns to predict counterfactual outcomes.}}
    \label{fig:overview}
    \vspace{-1em}
\end{figure*}
\begin{figure}[!t]
    \begin{minipage}[b]{0.48\textwidth}
    \centering
    \includegraphics[width=\linewidth]{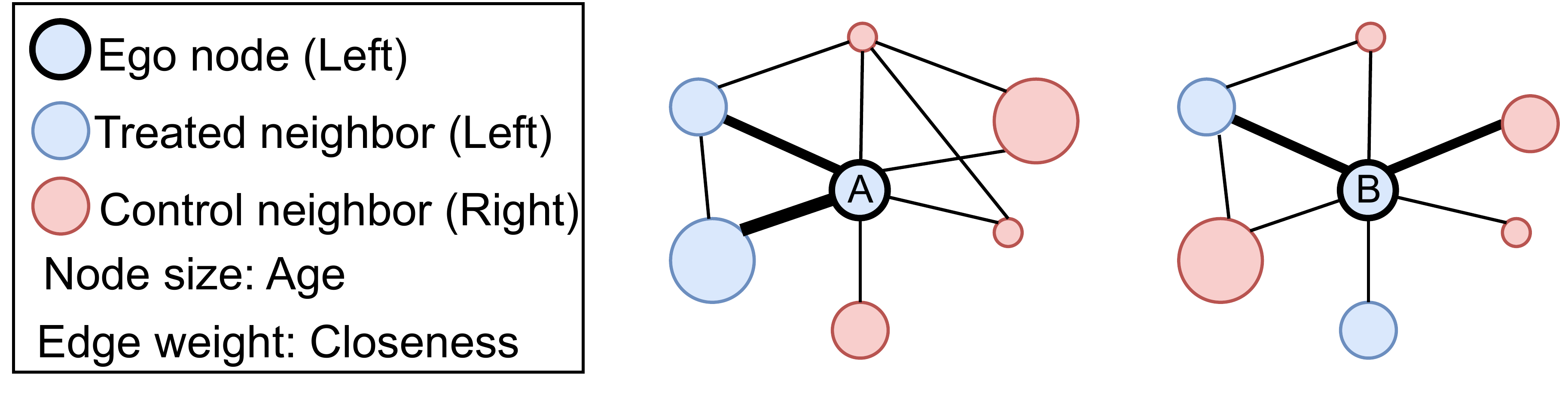}
    \caption{Two ego networks with the same number of treated and untreated friends but possibly different peer exposure conditions. }
    \label{fig:eg2-exp}
    \end{minipage}
    \begin{minipage}[b]{0.48\textwidth}
    \centering
    \includegraphics[width=0.8\linewidth]{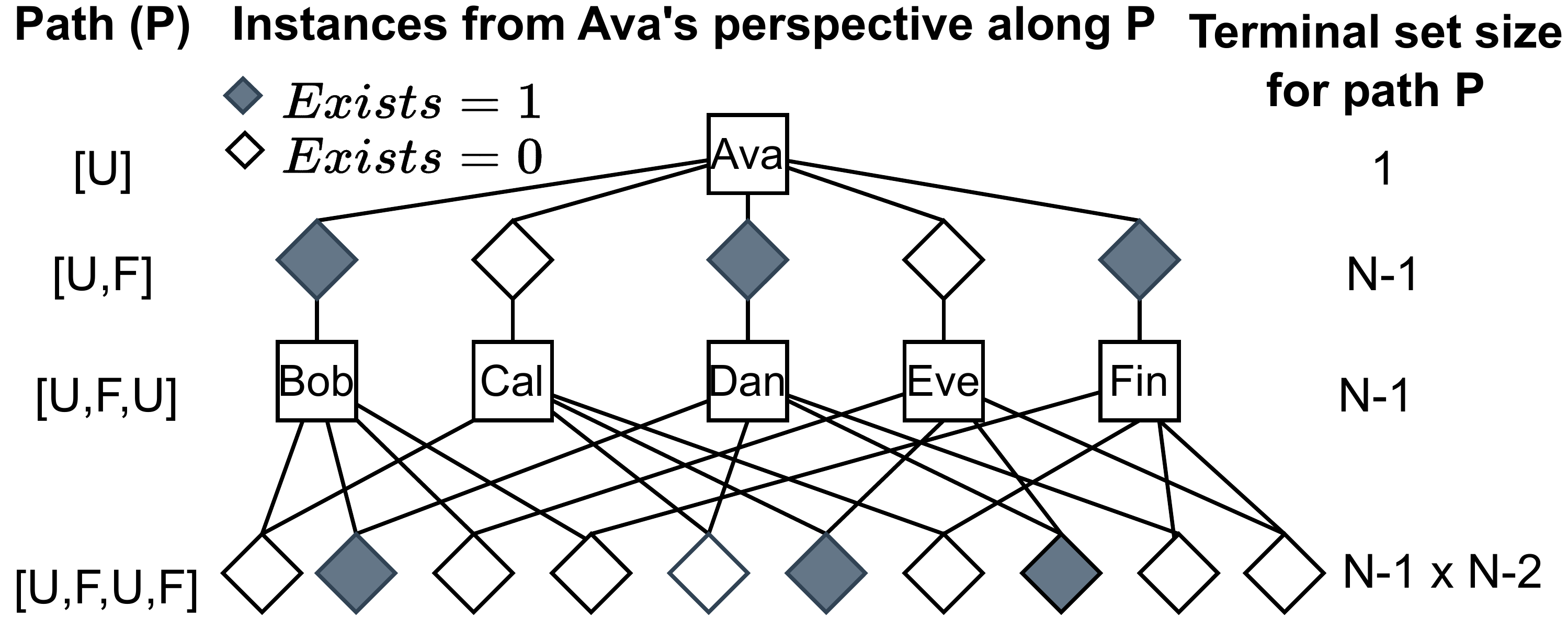}
    \caption{Relational paths (i.e., traversal in the schema) for a toy social network. The paths capture an ego, the ego's relationships, other nodes, and their relationships (except with the ego).}
    \label{fig:eg2-paths}
    \vspace{-1em}
    \end{minipage}
\end{figure}

Instead of assuming specific HPI contexts, our work focuses on estimating \textit{individual} direct effects (IDE) when the HPI contexts and heterogeneous susceptibilities to peer influence are not known a priori. As a motivating example, consider a toy social network, depicted in Fig. \ref{fig:overview}a), abstracted by a \textit{schema} (Fig. \ref{fig:overview}b)) with entity class User~\texttt{(U)}, relationship class Friend~\texttt{(F)}, and attributes associated with each class. The instantiation of the schema is an attributed network capturing the existence (\texttt{Ex}) and duration (\texttt{Dur}) of friendship for users with different demographics (\texttt{Dem}), political affiliations (\texttt{Aff}), and stances (\texttt{St}) on a policy issue (e.g., vaccination). The stances of individuals in the social network can be influenced by their political affiliation and the affiliations and/or stances of their peers. Understanding the causal impact of an individual's affiliation (the treatment) on the individual's own stance (the outcome) requires individual direct effect (IDE) estimation while controlling for indirect effects from peer treatments, outcomes, and other confounding factors, such as age. To illustrate heterogeneous peer influence, let us consider two ego networks (i.e., {the subnetwork of} a unit, its neighbors, and edges between them) for units A and B, depicted in Fig. \ref{fig:eg2-exp}, where the units have the same number of treated and untreated friends but different network structures and peer attributes. Methods assuming homogeneous peer influence would produce the same representation of peer exposure (e.g., 33\% treated) for the two units. However, peer exposure should be able to capture {more complex mechanisms of influence} due to node attributes (e.g., similarity of age), edge attributes (e.g., closeness), and network structure (e.g., peer degree). For example, {the treated peers of unit A are closest peers to the ego and are mutual friends}, and they may influence more collectively {than the treated peers of unit B}. A misspecified {exposure mapping} may not control the indirect effects {which would result} in a biased IDE estimate. 
{Similarly, peer exposure captures a unit's environment in the network and it may be an effect modifier, i.e., a context that leads to different IDEs for different individuals, and its misspecification could bias the IDE estimates.}
{Under unknown HPI mechanisms, exposure mapping should be expressive to automatically capture relevant contexts but be invariant to irrelevant contexts. Moreover, a feature mapping accounting for confounders and effect modifiers is crucial for IDE estimation. To ensure expressive feature extraction for IDE estimation, our framework, \emph{IDE-Net}, depicted in Figure \ref{fig:overview}, employs causal modeling and reasoning in network settings followed by feature extraction with graph neural networks (GNN) to capture contexts that can explain the underlying HPI mechanism, confounding, and effect modification. 
We propose an expressive \textit{Network Structural Causal Model (NSCM)} that encodes assumptions about the domain (e.g., interference conditions) and data collection (e.g., known selection bias). NSCM extends the relational causal model (RCM)~\cite{maier-thesis14,lee-thesis18} to incorporate relationship existence in networks (e.g., due to homophily, i.e., the tendency of similar individuals to be friends), latent variables, and selection bias. 
While existing methods prefer certain network models (e.g., distinct clusters~\cite{mcnealis-arxiv23}) or mechanisms (e.g., equilibrium~\cite{ogburn-jrssa20}), our model does not assume any specific network model, and by design, like RCM, is general for any network instantiation. 
The \textit{Network Abstract Ground Graph (NAGG)}, a graphical model derived from NSCM, enables reasoning about the identification of causal effects similar to the structural causal model (SCM)~\cite{pearl-book09}. We prove that NAGG is sound and complete for reasoning about relational d-separation. Using NSCM and NAGG, we can identify a maximal adjustment set of relational variables for IDE estimation. {The advantage of modeling causal relationships at the schema level with the NSCM and NAGG instead of at the instance level is that we can model generalizable mechanisms (e.g., homophily) in terms of abstract concepts.
}
Next, we propose a novel graph neural network (GNN)-based estimator that encourages expressive representation of the maximal adjustment set of relational variables to estimate IDEs by controlling potential confounders, heterogeneous peer influence contexts, and effect modifiers. {The use of GNN helps us to automatically capture features beyond unit's own covariates such as local network structure around the unit and neighbor covariates.} We use regularizations in the loss functions of the causal effect estimator to promote invariance by reducing sensitivity to irrelevant features. We empirically show the robustness of our approach for the estimation of average and individual direct effects under HPI.}
\newchange{To summarize, we adapt causal modeling in network and IDE estimation under realistic setting of heterogeneous peer influence with unknown mechanism of influence. Sections 2–5 of the paper are organized to present related work, problem setup using SCM framework, causal modeling and identifiability in network setting, and estimation of IDE, respectively. Then, in Section 6, we present the experimental setup and key results before discussing future work and concluding the paper.}
\section{Related Work} \label{sec:related}
Heterogeneous treatment effect estimation often assumes no interference and 
has focused on discovering heterogeneous contexts and estimating individual treatment effects (ITE) from i.i.d.~data~\cite{athey-pnas16,kunzel-pnas19,qidong-acml20} and from network data but without interference~\cite{guo-wsdm20,gilad-arxiv21}. A prominent approach for ITE estimation~\cite{johansson-icml16,shalit-icml17} involves representation learning with neural networks to balance covariates in the treatment and control groups and predict counterfactual outcomes. We focus on representation learning with graph neural networks (GNNs)~\cite{kipf-iclr16,xu-iclr18} for ITE estimation in networks with interference. 
{In contrast to other works in interference settings using GNN \newrevision{to summarize neighborhood covariates}, our GNN model is designed for increased expressiveness of node, edge, and network features from a unit's perspective to encourage controlling unknown HPI contexts for better IDE estimation.}

There are two main lines of work on modeling causal inference with interference. The first one relies on the potential outcomes framework~\cite{rubin-jep74} and typically assumes that the variables that satisfy the unconfoundedness assumption~\cite{rubin-jasa80} are known a priori~\cite{hudgens-jasa08,forastiere-jasa21,qu-arxiv21,bargagli-arxiv20}. The second one relies on graphical models~\cite{pearl-book09} which allow for the modeling of causal assumptions and reasoning about causal effect identification without requiring a priori knowledge about which variables would meet the unconfoundedness assumption. These works include causal diagrams for interference~\cite{ogburn-ss14}{, structural equation models (SEMs)~\cite{ogburn-jasa22}}, chain graphs~\cite{shpitser-arxiv17,ogburn-jrssa20}, structural causal models (SCMs) at the unit level with linear structural equations~\cite{zhang-neurips22}, declarative languages~\cite{salimi-sigmod20}, and relational causal models (RCMs)~\cite{arbour-kdd16,maier-thesis14,lee-thesis18}. Multiple works have considered identifiability and estimation under specific SCMs with interference~\cite{shalizi-sar11,cristali-neurips22,fatemi-arxiv23,hayes-arxiv22}. Our causal modeling builds on RCMs because of their flexibility to represent and reason about causal dependence in relational domains 
and because they lend themselves to learning causal models from data~\cite{maier-uai13,lee-uai20}. {\citet{ogburn-jasa22} study causal inference in social network using SEMs with summary variables but that work does not consider unknown HPI contexts like ours.}

Most methods for estimating individual or heterogeneous causal effects under interference~\cite{bargagli-arxiv20,jiang-cikm22,ogburn-jasa22,cai-cikm23,chen-icml24} assume homogeneous interference between units. \newrevision{These methods adapt techniques like adversarial training~\cite{jiang-cikm22}, propensity score reweighting~\cite{cai-cikm23}, and doubly robust estimation via targeted learning~\cite{chen-icml24} for causal effect estimation in homogeneous interference settings.} Recent works solving a diverse set of problems have {implicitly or explicitly} addressed heterogeneous peer influence (HPI) due to known contexts~\cite{qu-arxiv21,forastiere-jasa21} or specific contexts~\cite{yuan-www21,ma-kdd22,tran-aaai22,zhao-arxiv22,lin-ecmlkdd23}. Some of these works refer to HPI as heterogeneous interference~\cite{qu-arxiv21,zhao-arxiv22,lin-ecmlkdd23}. \citet{tran-aaai22} study peer contagion effects with homogeneous influence but different unit-level susceptibilities to the influence. \citet{yuan-www21} capture peer exposure with causal network motifs, i.e., recurrent subgraphs in a unit's ego network with treatment assignments. \citet{ma-kdd22} focus on addressing heterogeneous influence due to group interactions utilizing hypergraphs. 
\citet{zhao-arxiv22} deal with heterogeneity due to node attribute similarity using attention weights to estimate peer exposure and causal effects. \citet{lin-ecmlkdd23} consider heterogeneity due to multiple entities and relationships in networks.
In contrast, we focus on causal modeling in networks that helps to identify relational random variables that capture unknown HPI contexts, confounders, and effect modifiers. %
{Our work focuses on addressing mispecified exposure mappings due to unknown HPI, but there are other works studying causal inference with mispecified exposure mappings~\cite{savje-bio24}.}
\section{Causal Inference Problem Setup}
\subsection{Data model}
We represent the network as an undirected graph {\small $G=(\mathcal{V},\mathcal{E})$} with a set of {\small $N=|\mathcal{V}|$} vertices and a set of edges {\small $\mathcal{E}$}. We assume all vertices {\small$\mathcal{V}$} belong to a single entity class {\small $E$} (e.g., User) and the edges {\small $\mathcal{E}$} belong to a single relationship class {\small $R$} (e.g., Friend). We denote node attributes with {\small $\mathcal{A}(E)$} and edge attributes with {\small $\mathcal{A}(R)$}.
Let {\small $X=<X_1,...,X_i,...,X_N>$} be a random variable comprising the treatment variables {\small $X_i$} for each node $v_i$ in the network and {\small $Y_i$} be a random variable for $v_i$'s outcome.
Let {\small $\bm{\pi}=<\pi_1,...,\pi_i,...,\pi_N>$} be an assignment to {\small $X$} with {\small $\pi_i \in \{0,1\}$} assigned to {\small $X_i$}. 
{Let {\small $X_{-i}=X\setminus X_i$} and {\small $\bm{\pi}_{-i}=\bm{\pi} \setminus \pi_i$} denote random variable and its value for treatment assignment to other units except $v_i$.
\subsection{Causal estimand}
We focus on estimating individual direct effects when the contexts for \emph{heterogeneous peer influence} (HPI) and effect modification are not known a priori {and use the structural causal model (SCM) framework to define the causal estimand}. Let {\small $\mathcal{Z}_i$} denote unknown HPI contexts that will be identified as functions of node attributes, edge attributes, and network structure (in \S \ref{sec:repr}).}
\newchange{
\begin{definition}[Heterogeneous Peer Influence (HPI)]
    Heterogeneous peer influence is a general form of interference where an attribute of a unit, $A_i$, is influenced by attribute of other units, $B_{-i}$, with strength of influence depending on context $\mathcal{Z}_i$, i.e., $A_i = f_A(\phi(B_{-i}, \mathcal{Z}_i), ...)$, where $f_A$ and $\phi$ are functions, $A,B \in \mathcal{A}$ are attributes, and $...$ indicate other terms in $f_A$ impacting $A_i$.
\end{definition}
}

The \textit{individual direct effect} (IDE) for a unit $v_i \in \mathcal{V}$, denoted as {\small $\tau_i$} 
, with treatment {\small $X_i=1$} vs. {\small $X_i=0$} given observed or assigned treatments of other units {\small $X_{-i}=\bm{\pi}_{-i}$} and contexts {\small $\mathcal{Z}_i$} is defined as 
{\small
\begin{equation}
    \begin{split}
    \label{eq:dir_eff_intermediate}
    \tau_i = E[Y_i(X_i=1, X_{-i}=\bm{\pi}_{-i}) - Y_i(X_i=0, X_{-i}=\bm{\pi}_{-i})| X_{-i}=\bm{\pi}_{-i}, \mathcal{Z}_i], 
    \end{split}
\end{equation}
}
where the counterfactual outcomes of unit $v_i$ (e.g., {\small $Y_i(X_i=1, X_{-i}=\bm{\pi}_{-i})$}) express the idea that the outcome is influenced by the entire treatment assignment vector {\small $\bm{\pi}$} due to interference. Equation \ref{eq:dir_eff_intermediate} captures that peer exposure and other effect modifiers are defined by treatments of other units {\small $X_{-i}$} and contexts {\small $\mathcal{Z}_i$}.
Equation \ref{eq:dir_eff_intermediate} is {similar} to \citet{hudgens-jasa08}'s Individual Direct Causal Effect estimand with fixed neighborhood treatment assignments. 
{Our work focuses on challenges of IDE estimation under unknown HPI mechanisms and the scope of the IDE estimand (Eq. \ref{eq:dir_eff_intermediate}) is to deal with known/observed neighborhood assignments and not different possible or counterfactual neighborhood assignments. Estimation of peer and other network effects with counterfactual neighborhood assignments under HPI is outside the scope of this work. The estimand above is important  to evaluate the effectiveness of treatments (e.g., vaccination) for each unit, taking into account the network environment and HPI.}

Using the consistency {rule}~\footnote{Consistency appears as a rule or axiom in the definition of counterfactual in the SCM framework whereas it appears as an assumption for effect estimation in the potential outcome framework~\cite{pearl-epi10}.}~\cite{pearl-book09,pearl-epi10} in causal inference that renders outcome generation independent of treatment assignment mechanisms (e.g., natural or intervention or counterfactual), we rewrite $\tau_i$ {by excluding neighborhood counterfactual term, i.e.,}
{\small 
\begin{equation}
        \begin{split}
        \label{eq:dir_eff}
        &\tau_i = E[Y_i(X_i=1)| X_{-i}=\bm{\pi}_{-i}, \mathcal{Z}_i] - E[Y_i(X_i=0) | X_{-i}=\bm{\pi}_{-i}, \mathcal{Z}_i].
\end{split}
\end{equation}
}
Equation \ref{eq:dir_eff}, similar to \citet{ma-kdd22}, differs from the ``insulated" individual effects~\cite{arbour-kdd16,jiang-cikm22} that consider the effects with no peer exposure. In \S \ref{sec:iden},  we discuss the identifiability of individual effects $\tau_i$ with experimental and observational data.


\subsection{Preliminaries on the relational causal model (RCM)}
\note{To model the causal relationships between units}, prior works~\cite{arbour-kdd16} rely \note{on relational causal models (RCM) and abstract ground graphs (AGG)~\cite{maier-thesis14,lee-thesis18}. The relational causal model (RCM) provides a principled way to define random variables for causal modeling and reasoning in relational settings \newrevision{similar to the structural causal model (SCM).  A \textbf{\textit{relational causal model}} (RCM) {\small $\mathcal{M}(\mathcal{S}, \mathbf{D}, \mathcal{P})$} encapsulates a schema {\small $\mathcal{S}$}, a set of relational dependencies {\small $\mathbf{D}$}, and optionally parameters {\small $\mathcal{P}$}. An AGG is a graphical model with relational variables as nodes, and dependencies between the relational variables as arcs.} AGG can be used for causal reasoning, including using d-separation and do-calculus similar to SCM~\cite{pearl-book09} for expressing counterfactual terms in Equation \ref{eq:dir_eff} in terms of observational or experimental distribution and identifying heterogeneous contexts $\mathcal{Z}$ in terms of relational variables. We provide the main RCM concepts here and give a more complete overview of RCMs in the Appendix.
Social networks like the one depicted in Fig. \ref{fig:overview}a) can be abstracted by a \textbf{\textit{relational schema}} {\small $\mathcal{S} = (E,R, \mathcal{A})$} comprising an entity class {\small $E$} (e.g., User), a relationship class {\small $R=<E, E>$} (e.g., Friend), entity attributes {\small $\mathcal{A}(E)$}, and relationship attributes {\small $\mathcal{A}(R)$}. The relational schema's partial instantiation, known as a \textbf{\textit{relational skeleton}}, specifies the entity instances (e.g., user nodes) and relationship instances (e.g., friendship edges) (Fig. \ref{fig:overview}b)).}

\note{An important concept in RCMs is \textbf{\textit{relational random variable}} which allows us to capture sets of random variables, such as the political affiliations of friends. It is defined from the perspective of an \textit{item class}} {\small $I\in\{E,R\}$}, i.e., an entity or a relationship class \note{(e.g., User), and consists of a \textbf{\textit{relational path}} and} an attribute \note{reached following the path. A relational path {\small $P=[I_j,...,I_k]$} is an alternating sequence of entity and relationship classes and captures path traversal in the schema. For the social network example, the relational paths capture the self, i.e., \texttt{[U]}, and friendship paths from the perspective of a user such as immediate friends, i.e., \texttt{[U,F,U]}, and friends of friends, i.e., \texttt{[U,F,U,F,U]}. The relational variables capture the self's attributes like stance \texttt{[U].St} and a multi-set of the attributes of immediate friends like political affiliation \texttt{[U,F,U].Aff}, and so on. The values of the relational variables \texttt{[U].St} and \texttt{[U,F,U].Aff} for user \texttt{Ava} are \{\texttt{Anti}\} and \{\texttt{Right,Right,Left}\}, respectively.
}

\note{Similar to SCM, the causal mechanism assumptions of node and edge attributes can be encoded by using \textbf{\textit{relational dependencies}} to link relational random variables to their causes. For example, the dependency \texttt{[U,F,U].Aff}$\rightarrow$\texttt{[U].St} captures that the stance of a user is influenced by the political affiliations of immediate friends. 
A RCM alone is insufficient for capturing all conditional independences in the data and thus to reason about sufficient conditioning set for causal identification~\cite{maier-thesis14,arbour-kdd16}.} 
\commentout{RCM only defines a template for how dependencies should be applied to data instantiations. A ground graph is a directed graphical model with nodes as attributes of entity instances (like {\small \texttt{Bob.Aff}} and {\small \texttt{Ava.St}})  and relationship instances (like {\small \texttt{<Bob-Ava>.Dur}}), and arcs as dependencies between the nodes applied from RCM. An acyclic ground graph has the same semantics as a standard graphical model~\cite{getoor-isrl07}, and it facilitates reasoning about the conditional independencies of attribute instances using standard d-separation~\cite{pearl-book88}. However, it is more desirable to reason about dependencies that generalize across all ground graphs by using expressive relational variables and abstract concepts.} 
\note{A higher-order representation known as an \textbf{\textit{abstract ground graph}} (AGG)~\cite{maier-thesis14,lee-thesis18} has been developed for this purpose with nodes including relational variables and arcs including dependencies $\mathbf{D}$ in the RCM and additional extended dependencies to capture all conditional independence. For example, dependency \texttt{[U,F,U].Aff}$\rightarrow$\texttt{[U].St}, suggesting influence from immediate friends, is extended to other dependencies like \texttt{[U].Aff}$\rightarrow$\texttt{[U,F,U].St} and \texttt{[U,F,U,F,U].Aff}$\rightarrow$\texttt{[U,F,U].St}. These extended dependencies are automatically identified by the \textit{extend operator}~\cite{maier-thesis14}, which is described in the Appendix. AGG is theoretically shown to be sound and complete in abstracting all conditional independences} encodeded in the RCM~\cite{maier-thesis14,lee-uai15}. 

\section{Network Structural Causal Model (NSCM)}\label{sec:repr}
\note{To enable reasoning about heterogeneous peer influence, we propose the Network Structural Causal Model (NSCM) and Network Abstract Ground Graph (NAGG), which extend RCM and AGG and provide a more expressive representation of causal assumptions. Specifically, unlike RCM, NSCM can encode causal assumptions related to relationship existence (e.g., related to homophily or preferential attachment), which are essential to reasoning about heterogeneous peer influence. Moreover, this is the first work to consider the implications of latent attributes and selection bias in RCM-like models. NSCM and NAGG are of independent interest for causal modeling and reasoning in networks beyond the scope of this paper on IDE estimation.
}

\textbf{Modeling relationship existence}. 
\newchange{We refer the network {\small $G(\mathcal{V}, \mathcal{E})$} as a single entity type and single relationship type (SESR) model~\cite{maier-thesis14} with a schema {\small $\mathcal{S} = (E,R, \mathcal{A})$}.} 
We \note{first} redefine relational skeleton, \note{i.e., the schema instantiation}, to incorporate relationship existence uncertainty, similar to DAPER~\cite{heckerman-isrl-07} and PRM~\cite{getoor-isrl07} models.

\begin{definition}[Relational skeleton with relationship existence]\label{def:skeleton}
A \textit{relational skeleton} {\small $\sigma \in \Sigma_{\mathcal{S}}$} for a schema {\small $\mathcal{S}$} specifies sets of entity instances {\small $\sigma(E)=\mathcal{V} \in G$} and relationship instances {\small $\sigma(R)=\Sigma_{\sigma(E)}$}, where {\small $\Sigma_{\sigma(E)}$} is a set of edges in a complete graph with {\small $|\sigma(E)|$} entity instances as nodes and {\small $\Sigma_{\mathcal{S}}$} is all possible skeletons.
\end{definition}
\begin{definition}[Relationship existence indicator]\label{def:ex_ind}
A relationship existence indicator {\small $Exists \in \mathcal{A}(R) \wedge Exists \in \{0,1\}$} is a relationship class attribute that indicates the existence of an edge in {\small $G$}.
\end{definition}
The above two definitions distinguish between relationship instances {\small $\sigma(R)$} and edges {\small $\mathcal{E} \in G$}. 
Relationship instances, 
connect all pairs of users, whereas user edges exist only for relationship instances whose Exists attribute equals 1.
We allow attributes in the schema {\small $\mathcal{S}$} to be marked as latent type or selection type. A \textit{latent attribute set} {\small $\mathbf{L} \in \mathcal{A}$} indicates attributes that are always unmeasured. A \textit{selection attribute set} {\small $\mathbf{S} \in \mathcal{A}$} indicates attributes that are responsible for selection (bias).
Next, we present new theoretical properties regarding relational paths due to these extended definitions.

\newchange{A relational path's beginning and ending item classes are called the base item class and terminal item class, respectively. A terminal set comprises instances of the terminal item class reached from an instance of the base item class after traversing a relational path. For example, the terminal sets for a base user \texttt{Ava} and paths \texttt{[U]} and \texttt{[U,F,U]}, respectively, consists of the self, i.e., \{\texttt{Ava}\} and others, i.e., \{\texttt{Bob, ..., Fin}\}. The path traversal follows \textit{bridge-burning semantics} (BBS)~\cite{maier-thesis14} where nodes at lower depth are not visited again. A breadth-first traversal preserves BBS by exploring all nodes at each depth and not visiting explored nodes. \note{As shown in Figure \ref{fig:eg2-paths}}, BBS traverses {\small $\sigma$} along {\small $P$} in breadth-first order, beginning at {\small $i_j$}, to obtain terminal set {\small $P|_j^\sigma$} as tree leaves at level {\small $|P|$}. A relational path {\small $P$} is valid if its terminal set for any skeleton is non-empty, i.e. {\small $\exists_{\sigma \in \Sigma_{\mathcal{S}}},\exists_{j \in \sigma(I)},P|_j^\sigma \ne \emptyset$}.
\begin{lemma}\label{lemma:int}
In a SESR model, the terminal sets of any two relational paths under BBS, starting with a base item class $I$, do not intersect.\end{lemma}
\begin{proposition}\label{prop:paths}
A SESR model with relationship existence uncertainty has the maximum length of valid relationship paths $[I_j,...,I_k]$ of four for $I_j \in E$ and five for $I_j \in R$ under bridge burning semantics (BBS) for any skeleton $\sigma \in \Sigma_\mathcal{S}$.
\end{proposition}
The proofs are presented in \S \ref{appendix:proof} and are based on Definitions \ref{def:skeleton} and \ref{def:ex_ind} and BBS exploration. Lemma \ref{lemma:int} allows us to rule out the complexity of intersecting paths~\cite{maier-thesis14,lee-uai15} in the SESR model.}
\newchange{The primary implication of Definitions \ref{def:skeleton}, \ref{def:ex_ind}, and Proposition \ref{prop:paths} is that relational paths are now interpreted differently than they were under RCM. The entire network is represented from the perspective of users by four paths {\small \texttt{[U]}, \texttt{[U,F]}, \texttt{[U,F,U]},} and {\small \texttt{[U,F,U,F]}} mapping to an ego, the ego's relationships, other nodes, and relationships between other nodes. This differs from RCM that captures entity or relationship instances at n-hops~\cite{arbour-kdd16}. We allow n-hop relations to change due to interventions or over time, and it is reflected in the {\small $Exists$} attribute of instances reached from paths {\small \texttt{[U,F]}} and {\small \texttt{[U,F,U,F]}}.}

Next, we can define relational variables and relational dependencies to encode relationship existence, latent variables, and selection bias.
To encode that a user's stance is influenced by affiliations of immediate friends, we need two dependencies {\small \texttt{[U,F,U].Aff}$\rightarrow$\texttt{[U].St}} and {\small \texttt{[U,F].Ex}$\rightarrow$\texttt{[U].St}}. 
The dependency {\small \texttt{[F,U].Dem$\rightarrow$[F].Ex}} from the relationship class perspective captures the existence of an edge due to homophily where users with similar demographics tend to be friends.
We model latent confounding and selection bias with dependencies of the form {\small $P.L_i \rightarrow [I].Y$} and {\small $P.X \rightarrow [I].S_i$}, respectively, where {\small $L_i \in \mathbf{L}$} and {\small $S_i \in \mathbf{S}$}. For example, the dependencies {\small \texttt{[F,U].L}$\rightarrow$\texttt{[F].Ex}} and {\small \texttt{[U].L}$\rightarrow$\texttt{[U].St}} jointly represent latent homophily and latent confounding. 
To represent contagion, we could introduce relational dependency with time-lagged attributes of the form {\small $[I].Y_{L}\rightarrow [I].Y$} and {\small $[I,...,I].Y_{L}\rightarrow [I].Y$}. The time-lagged attributes should mirror incoming and outgoing dependencies of {\small $Y$}, and these attributes could be unobserved i.e., {\small $Y_{L} \in \mathbf{L}$}.
We define \textit{implicit dependencies} {\small $\mathbf{D_i}:=\{\forall_{Y \in {Exists'}}, [R].Exists \rightarrow [R].Y\}$} to capture that the relationship attributes except {\small $Exists$}, i.e., {\small ${Exists'}:=\mathcal{A}(R)\setminus Exists$}, have non-default values only if the edge exists. For example, the duration of friendship is defined only if friendship exists.
\textbf{Network structural causal model (NSCM)}. A \textit{network structural causal model} {\small $\mathcal{M}(\mathcal{S}, \mathbf{D}, \mathbf{f})$} encapsulates a schema {\small $\mathcal{S}$}, a set of relational dependencies {\small $\mathbf{D}$}, and a set of functions {\small $\mathbf{f}$}. In contrast to RCM, the dependencies {\small $\mathbf{D}$} are more expressive, and the parameters are defined as functions {\small $\mathbf{f}$}. The functions relate each canonical (i.e., of path length $1$) variable {\small $[I].Y$} to its causes, i.e., {\small $[I].Y = f_{Y}(pa([I].Y, \mathbf{D}), \epsilon_Y)$}, where {\small $f_{Y} \in \mathbf{f}$}, {\small $pa([I].Y, \mathbf{D}):=\{P.X|P.X \rightarrow [I].Y \in \mathbf{D}\}$}, and {\small $\epsilon_Y$} is an exogenous noise.
Generally, the functions {\small $\mathbf{f}$} in the model {\small $\mathcal{M}$} are unknown and estimated from the data. 

\textbf{Network abstract ground graph (NAGG)}.
Similar to an abstract ground graph (AGG), a network abstract ground graph ({\small $NAGG_{\mathcal{M}B}$}), for the NSCM {\small $\mathcal{M}$} and a perspective {\small $B \in I$}, is a graphical model to enable reasoning about relational d-separation\note{, an extension of d-separation for relational domains~\cite{maier-thesis14}}. 
\note{To exemplify the power of this representation based on NSCM, rather than a RCM,} we consider a NAGG from the perspective of an entity class (e.g., User). \note{Let} {\small $RV$} \note{denote the} set of all relational variables of the form {\small $[B,...,I_k].X$}, with a combination of all paths {\small $\{ \texttt{[U]}, \texttt{[U,F]}, \texttt{[U,F,U]},  \texttt{[U,F,U,F]}\}$} (Proposition \ref{prop:paths}) and all attributes $X$ of instances reached following the path. \note{Let {\small $RVE \subset RV \times RV$} denote} a set of arcs between pairs of relational variables obtained from explicit dependencies $\mathbf{D}$, implicit dependencies $\mathbf{D_i}$, and extended dependencies obtained \note{with the} extend operator, similar to RCM.
Figure \ref{fig:overview}d depicts a NAGG from the perspective of User class for an NSCM {\small $\mathcal{M}(\mathcal{S}, \mathbf{D})$} \note{for} the toy social network schema. The NSCM encodes the causal hypotheses with explicit relational dependencies {\small $\mathbf{D}$} indicated by black and brown \note{arrows}. These dependencies capture: (1) the existence of friendship depends on demographics \note{(enabled by NSCM but not RCM)}; (2) stance is affected by one's affiliation; (3) one's demographics confound affiliation and stance while peer demographics influence one's affiliation; and (4) a user's current stance is influenced by other users' previous latent stances (enabled by NSCM but not RCM), depending on friendship duration. All the explicit dependencies and the implicit dependency (see blue arrow) for relationship existence uncertainty, are extended (see purple arrows) to capture dependencies between all relational variables for a given perspective. 

We extend the theoretical properties of soundness and completeness of AGG~\cite{maier-thesis14} to NAGG with relationship existence uncertainty, latent attributes, and selection attributes. With these theoretical properties, presented in Appendix \ref{appendix:proof}, NAGG enables reasoning about the identification of causal effects in networks using counterfactual reasoning with the SCM~\cite{pearl-book09}.
\textbf{Identification of causal effects}.\label{sec:iden}
Given a NAGG, we can reason \note{about} the identifiability of causal effects using do-calculus (e.g., backdoor adjustment)~\cite{pearl-book09}. Here, we discuss IDE identification in Eq. \ref{eq:dir_eff} using NSCM and NAGG under Assumption \ref{assum:pre}:

\begin{assumption}[A1]\label{assum:pre}
The network {\small $G$} and its attributes are measured before treatment assignments and treatments are immutable from assignment to outcome measurement.
\end{assumption}

Let node attributes and edge attributes be denoted by {\small $\mathbf{Z_n}$} and {\small $\mathbf{Z_e}$}. We denote relational variables from entity class perspective {\small $[E].\mathbf{Z_n}, [E,R].\mathbf{Z_e}$, $[E,R,E].\mathbf{Z_n}$}, and {\small $[E,R,E,R].\mathbf{Z_e}$} with the notations {\small $Z_i\in \mathbb{R}^{d}, Z_r\in \mathbb{R}^{N'\times d'}, Z_{-i}\in \mathbb{R}^{N'\times d},$} and {\small $Z_{-r}\in \mathbb{R}^{N'\times N'-1 \times d'}$}, respectively, where {\small $N'=|\mathcal{V}|-1$ and $<d,d'>$} are constants. Let {\small $E_r \in \{0,1\}^{N'\times 1}$} and {\small $E_{-r} \in \{0,1\}^{N'\times N'-1 \times 1}$} be the relationship existence indicator variables {\small $[E,R].Exists$} and {\small $[E,R,E,R].Exists$}, respectively. The underlying data generation of treatment and outcome can be described by functions {\small $f_X$} and {\small $f_Y$} in the NSCM under Assumption \ref{assum:pre} as: \\
{\small $X_i = f_X(Z_i, Z_r, Z_{-i}, Z_{-r}, E_r, E_{-r}, X_i^L, X_{-i}^L,\epsilon_X)$} and \\
{\small $Y_i = f_Y(X_i, X_{-i}, Z_i, Z_r, Z_{-i}, Z_{-r}, E_r, E_{-r}, Y_i^L, Y_{-i}^L, \epsilon_Y)$,}
where {\small $\{X_i^L, X_{-i}^L,Y_i^L, Y_{-i}^L\}$} are latent time-lagged variables for treatment and outcome.

\begin{lemma}\label{lemma:adjustment}
     If there exists any adjustment set {\small $\mathbf{W}$} that satisfies the unconfoundedness condition~\cite{rubin-jasa80}, i.e., {\small $\{Y_i(X_i=1),Y_i(X_i=0)\} \perp\!\!\!\!\perp  X_i | \mathbf{W}$}, then the adjustment set {\small ${X_{-i}} \cup \mathcal{Z}_i$}, where {\small $\mathcal{Z}_i = \{Z_i, Z_r, Z_{-i}, Z_{-r}, E_r, E_{-r}\}$} also satisfies the condition under Assumption \ref{assum:pre}.
\end{lemma}
The proof, in \S \ref{sec:ap-iden}, is based on the fact that {\small ${X_{-i}} \cup \mathcal{Z}_i$} is a valid maximal adjustment set for the NSCM under A\ref{assum:pre}. In general, given a NSCM, we can construct a NAGG and find a maximal adjustment set, if exists, for the treatment and outcome of interest.
\begin{corollary}\label{prop:iden}
 Under Assumption \ref{assum:pre}, the counterfactual term {\small $E[Y_i(X_i=\pi_i)|X_{-i},\mathcal{Z}_i]$} can be estimated as {\small $E[Y_i|X_i=\pi_i, X_{-i}, \mathcal{Z}_i]$} from experimental data or from observational data if Lemma \ref{lemma:adjustment}'s condition holds.
\end{corollary}

The proof for Corollary \ref{prop:iden}, in \S \ref{sec:ap-iden}, follows from do-calculus~\cite{pearl-book09} applied to NAGG derived from NSCM. \S \ref{sec:ap-iden} discusses the implications on the identifiability when the assumptions are violated.
\section{IDE Estimation Framework}\label{sec:estimation}
Now that we have defined the graphical models necessary for reasoning about IDE identifiability, we are ready to describe the IDE Estimator which is the last step in the \emph{IDE-Net} framework (Fig. \ref{fig:overview}e)). The estimator consists of two key steps: (1) mapping raw inputs to feature representations, and (2) estimating IDE using \textit{counterfactual outcome (CFO)} prediction. While the second step is fairly straightforward and can be performed with supervised learning (e.g.,~\cite{shalit-icml17}), a key challenge for IDE estimation is the representation of the relational variables whose values are multi-sets and their ability to capture the underlying heterogeneous peer influence (HPI) strength to control indirect effects. We address this challenge by capturing relational variables in the maximal adjustment set and their interactions with graph neural networks (GNNs) to learn an expressive representation of potential HPI contexts. The novelty of our work compared to previous research which uses GNNs for causal effect estimation (e.g., ~\cite{jiang-cikm22,cai-cikm23}) is the consideration of unknown HPI contexts.
Next, we describe feature and exposure embeddings to map the raw attributed network to relational variables and their interactions to capture potential confounders, effect modifiers, and HPI contexts.


\textbf{Feature mappings}. 
The inputs to the framework are adjacency matrix {\small $\mathbf{A}$} with {\small $n$} nodes and {\small $m/2$} undirected edges, treatment {\small $X \in \{0,1\}^n$}, outcome {\small $Y \in \mathbb{R}^n$}, node attributes {\small $\mathbf{Z}_n \in \mathbb{R}^{n \times d}$} and edge attributes {\small $\mathbf{Z}_e \in \mathbb{R}^{m \times d'}$}. Ideally, a feature mapping should capture sufficient statistics about the underlying distribution of relational variables in {\small $\mathcal{Z}_i$} and their interactions. The feature mapping {\small $\phi_f$} maps network structure, node attributes, and edge attributes in {\small $\mathcal{Z}_i$} to an embedding vector of dimension {\small $l$}, i.e., {\small $\phi_f(\mathbf{Z}_n, \mathbf{Z}_e, \mathbf{A}) \rightarrow \mathbb{R}^{n\times l}$}. 

\newchange{We propose GNN feature embeddings {\small $h_f$} by concatenating outputs of ego multi-layer perception (MLP) {\small $h_i$}, the peer graph convolution network (GCN) {\small $h_{-i}$}, and edge GCN {\small $h_r$} modules capturing relational variables {\small $\phi_i(Z_i)$}, {\small $\phi_{-i}(Z_{-i}, Z_r, E_r)$}, and {\small $\phi_{r}(Z_r, Z_{-r}, E_r, E_{-r})$}, respectively. Let {\small $\Theta$} denote MLP with ReLU, and $<\mathbf{v_i}, \mathbf{v_j}>$ be the indices of nodes with edges, i.e., non-zero $\mathbf{A}$, then, we define,
{\small
\begin{align*}
    \begin{split}
    &h_i = \Theta(\mathbf{Z_n}),\\
    &h_{-i}=\mathbf{D}^{-\frac{1}{2}}\overset{row}{\textstyle \sum}\mathbf{A} \odot \Theta(\mathbf{Z_n[v_j]} || \mathbf{Z_e}),\\
    &h'_{r}=\overset{row}{\textstyle \sum}\mathbf{A} \odot \Theta(\mathbf{Z_e}), h_r= h'_{r}||\mathbf{D}^{-\frac{1}{2}}\mathbf{A}h'_{r},\text{ and }\\ &h_f = h_i||h_{-i}||h_{r},
    \end{split}
\end{align*}
} where {\small $\mathbf{D}$} is a degree matrix used for normalization, {\small $\odot$} is index-wise product operator, and {\small $||$} is a concatenation operator. Here, {\small $h_{-i}$} is a normalized GCN and {\small $h_{r}$} is GCN after peer aggregation {\small ($h'_{r}$)}. Existing work has shown mean aggregation, i.e., normalization using {\small $\mathbf{D}^{-1}$} could make GNNs less expressive than sum aggregation~\cite{xu-iclr18}. However, using the sum aggregation may make representations unstable, especially in scale-free networks (e.g., online social networks) that have some nodes with high degrees, and we choose $\mathbf{D}^{-\frac{1}{2}}$ for normalization. 
We can extract l-hop deep features using $h^l = \mathbf{D}^{-\frac{1}{2}}\mathbf{A}\Theta(h^{l-1})$, where $h^1 = h_f$.}


\textbf{Exposure mapping}. The exposure mapping $\phi_e$ maps the treatment of other units, the feature embeddings, raw network, and edge attributes to an embedding vector of dimension $k$, i.e., $\phi_e(X, h_{f}, \mathbf{A}, \mathbf{Z_e}) \rightarrow \mathbb{R}^{n \times k}$.
Ideally, the exposure mapping under unknown HPI mechanisms should be expressive to capture relevant contexts but also invariant to irrelevant contexts. \newchange{To maintain expressiveness, we consider exposure mapping $h_e$ as a weighted fraction of peers with multiple candidate weights based on edge attributes, peer features, similarity with peers, and network structure, i.e.,
\begin{align*}
    h_e = \frac{\overset{row}{\textstyle \sum} \mathbf{A} \odot h_w \odot X_i[\mathbf{v_j}]}{\overset{row}{\textstyle \sum} \mathbf{A} \odot h_w}| h_w := \Theta(\mathbf{Z_e})||(\mathbf{A}\odot(\mathbf{A}\mathbf{A}^T))[\mathbf{v_i}, \mathbf{v_j}]||h_f[\mathbf{v_j}]||sim(h_f[\mathbf{v_i}], h_f[\mathbf{v_j}]),
\end{align*}
where $h_w \in \mathbb{R}^{m \times k}$ represents edge weights capturing underlying peer influence strength, [.] indicates indexing, and $sim(a,b)=e^{-(a-b)^2}$ captures similarity. The edge attributes, network connections, peer features, and similarity with peers are incorporated in $h_w$ as potential influence strengths.} The embedding $h_e$ is passed to the downstream IDE estimator to strengthen invariance to irrelevant mechanisms for outcome prediction.

\textbf{Individual direct effect (IDE) estimation}. 
The feature and exposure embeddings, capturing representation for each node in the network, can be plugged into existing ITE estimators. For exposition and reproducibility, we demonstrate IDE estimation with the TARNet approach~\cite{shalit-icml17} that uses CFO prediction agnostic of treatments, i.e., {\small $\hat{Y}(1) = \Theta_1(h), \hat{Y}(0)=\Theta_0(h)$}, and {\small 
$\hat{Y}=\hat{Y}(1) \text{ if } X=1 \text{ else } \hat{Y}(0)$}, where {\small $\Theta$} is a MLP and {\small $h=\Theta(h_f||h_e)$}. We optimize the model to minimize the loss
    {\small $\mathcal{L} = (Y - \hat{Y})^2 + \lambda_s e^{-\gamma \times \sigma_{\hat{\tau}}^2} \sigma_{\hat{\tau}}^2$}, where {\small $\hat{\tau}=\hat{Y(1)}-\hat{Y(0)}$} is estimated effect and {\small $\sigma_{\hat{\tau}}^2$} is the variance.
Here, the second term is a regularization to smooth variance in estimated causal effects controlled by scale parameter $\lambda_s$ and decay parameter $\gamma$. Due to the decay term, the regularization gets stronger for smaller variances and weaker for larger variances smoothing the variance of estimated effects. 


Here, the choice of embeddings and estimator in the framework is for illustrative purposes. The embeddings could use a different GNN (e.g., GAN~\cite{velickovic-iclr18}) and the estimator with a different learner (e.g., X-learner~\cite{kunzel-pnas19}). One of the reasons to choose this specific representation and estimator is based on the fact that recent works without HPI relied on them~\cite{guo-wsdm20,jiang-cikm22}, which helps with fair empirical comparison to the baselines.
The choice of estimator in \emph{IDE-Net} is driven by the fact that we needed an estimator that can handle network data. While some ITE estimators have been extended to work for network data, including causal trees~\cite{bargagli-arxiv20}, inverse probability weighting (IPW)~\cite{qu-arxiv21}, and doubly-robust~\cite{mcnealis-arxiv23}, not all of them are suitable for observational data, unknown HPI contexts, or are agnostic to network model assumptions. 
\vspace{-0.5em}
\begin{figure*}[t]
    \centering
    \subfigure[Synthetic networks]{
        \includegraphics[width=0.40\linewidth,keepaspectratio]{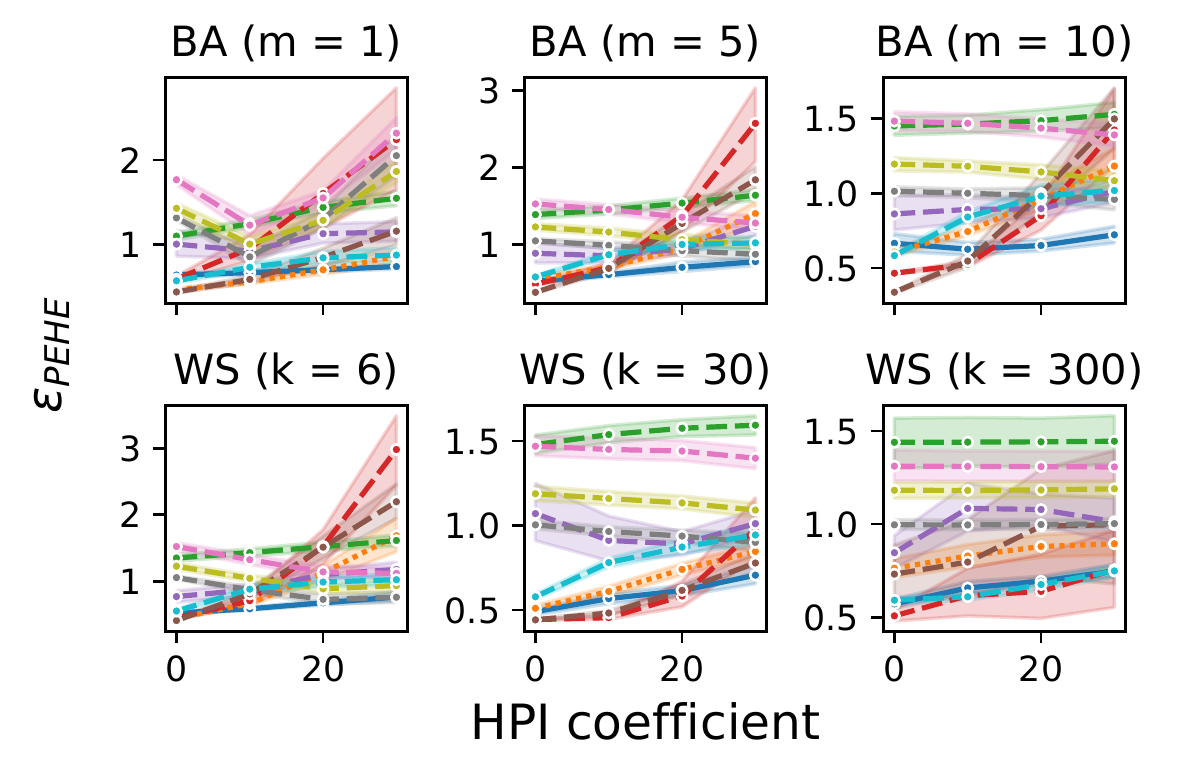}

    \label{fig:syn_cane_all_pehe}
    }
    \subfigure[Semi-synthetic networks]{
        \includegraphics[width=0.54\linewidth,keepaspectratio]{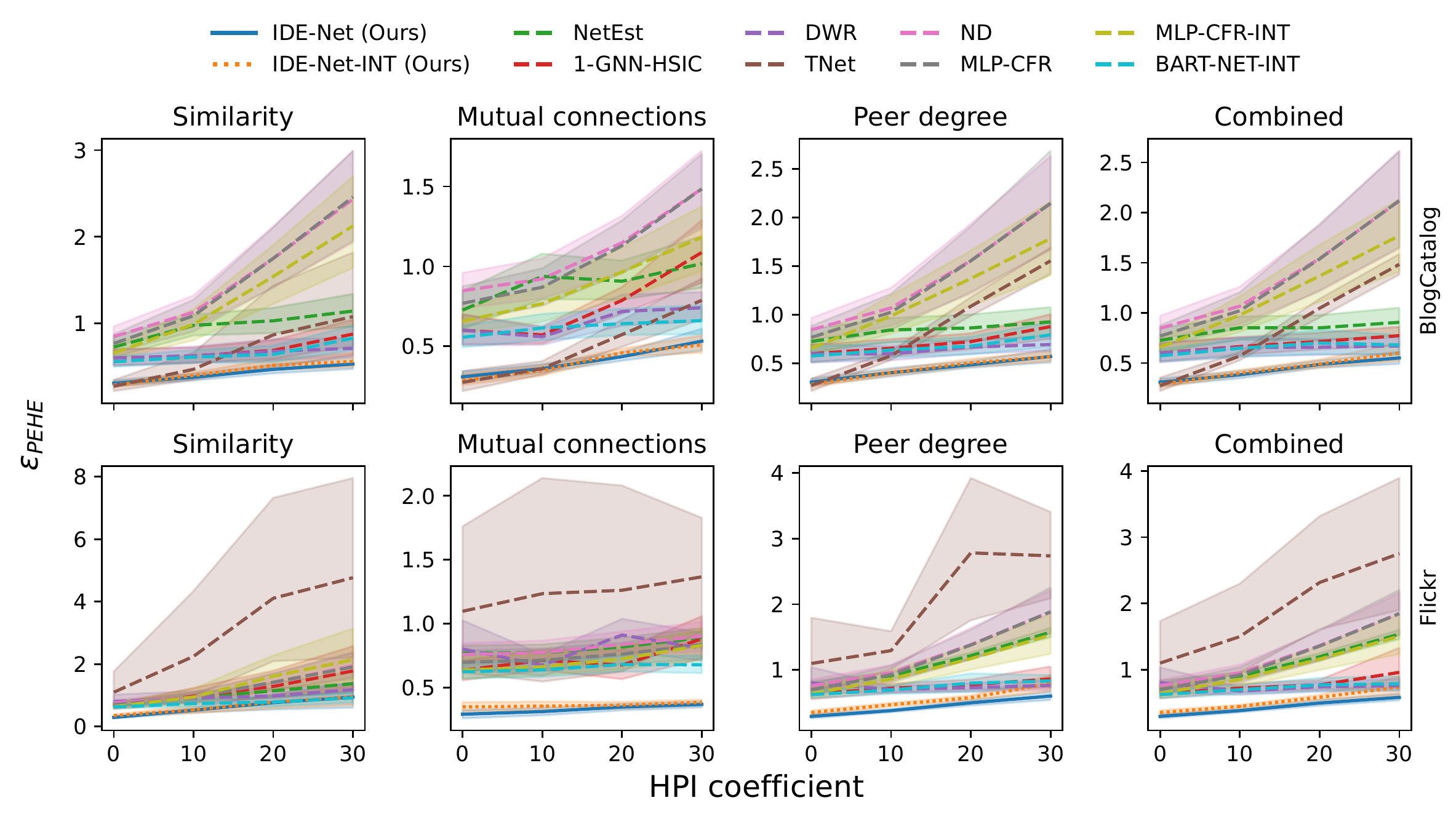}
    \label{fig:semi_cane_all_pehe}
    }
    \caption{Estimation error for {IDEs with effect modification} and HPI based on similarity, mutual connections, and peer degree.}
    \label{fig:cane_pehe}
    \vspace{-1em}
\end{figure*}


{\section{Experiments} \label{sec:experiment}
Here, we investigate the necessity of addressing HPI for robust IDE estimation.
We evaluate errors in the estimation of constant and heterogeneous direct effects in the presence of HPI due to various mechanisms. 


\subsection{Experimental Setup} \label{sec:exp_setup}
As it is common in causal inference tasks, we rely on synthetic and semi-synthetic datasets for evaluation because the ground truth causal effects and heterogeneity contexts are unknown in real-world data. 
We investigate the robustness of estimators for various network topologies and edge densities using synthetic data. For semi-synthetic data, we use real-world networks and attributes and generate the treatments and outcomes to evaluate estimators in more realistic and complex network settings. Similar to other state-of-the-art (SOTA) approaches~\cite{arbour-kdd16,jiang-cikm22,yuan-www21}, we assume one-hop interference.}

\textbf{Synthetic data}. 
We consider two random network models: (1) the Barab{\'a}si Albert (BA) model~\cite{albert-rmp02}, which captures preferential attachment phenomena (e.g., social networks), and (2) the Watts Strogatz (WS) model~\cite{watts-nature98}, which captures small-world phenomena. We generate both networks fixing {\small $N=3000$} and controlling the sparsity of edges (with the preferential attachment parameter {\small $m$} for the BA model and the mean degree parameter {\small $k$} with a fixed rewiring probability of {\small $0.5$} for the WS model). The topology and sparsity of networks should be considered to evaluate robustness because, as shown in \S \ref{sec:results}, they impact the underlying peer influence and thus the estimation of IDE. 
We 
defer the generation of node attributes {\small $\{C,Z\} \subset \mathbf{Z_n}$} (confounder and effect modifier) and edge attributes {\small $\{Z_r\} \subset \mathbf{Z_e}$} to \S \ref{sec:ap-data-gen}.
 
\textit{Treatment model}. For observational data, a unit's treatment may be influenced by its own covariates and peer treatments or covariates. We generate treatment {\small $X_i$} as {\small $X_i \sim \theta (a(\tau_c\mathbf{W}_x \cdot \frac{\sum E_r \odot C_{-i}}{\sum E_r} + (1-\tau_c)\mathbf{W}_x \cdot C_i))$}, where {\small $\theta$} denotes Bernoulli distribution, {\small $a$} is an activation function, {\small $\tau_c\in[0,1]$} controls influence from peers, and {\small $\mathbf{W}_x$} is the weight matrix.

\textit{Outcome model}. We define the outcome {\small $Y_i$} as a function of unit's treatment ({\small $\tau_d$}), attribute effect modifier ({\small $\tau_a$}), network effect modifier ({\small $\tau_n$}), HPI contexts ({\small $\tau_p, \tau_{Z_r},\tau_{Z_{-i}},\tau_{E_{-r}}$}), confounder ({\small $\tau_{c1}$}), and random noise ({\small $\epsilon$}). The detailed functional form of the model is in the Appendix \ref{sec:ap-data-gen}. 
The coefficients {\small $\tau_a$} and {\small $\tau_n$} control treatment effect modification based on node attributes and network location (e.g., degree centrality), respectively. We model HPI based on the strength of ties ({\small $\tau_{Z_r}$}), attribute similarity ({\small $\tau_{Z_{-i}}$}), local structure captured by mutual friend count ({\small $\tau_{E_{-r}}$}), and their combination. The peer exposure coefficient ({\small $\tau_p$}) controls the strength of HPI.

\textbf{Semi-synthetic data}. Following previous work~\cite{guo-wsdm20}, we use two real-world social networks, BlogCatalog and Flickr. 
We use LDA~\cite{blei-jmlr03} to reduce the dimensionality of raw features to $50$. The 50-dimensional features are randomly assigned to either confounding set {\small $\mathbf{C}$} or effect modifier set {\small $\mathbf{Z}$}.
Both real-world networks lack edge attributes and we consider the HPI based on friends' degrees instead of the strength of ties. The detailed functional forms for treatment and outcome models are included in \S \ref{sec:ap-data-gen}.

\textbf{Evaluation metrics}. 
To evaluate the performance of heterogeneous effect estimation, we use the \textit{Precision in the Estimation of Heterogeneous Effects} ({\small ${\epsilon_{PEHE}}$})~\cite{hill-jcgs11} metric defined as {\small ${\epsilon_{PEHE}} = \sqrt{\frac{1}{N}\sum_i (\tau_i - \hat{\tau_i})^2},$} where {\small $\tau_i$} and {\small $\hat{\tau_i}=\hat{Y_i(1)} - \hat{Y_i(0)}$} are true and estimated IDEs. {\small $\epsilon_{PEHE}$} (lower better) measures the deviation of estimated IDEs from true IDEs. To evaluate the performance of average IDE estimation, we use {\small $\epsilon_{ATE}$} metric defined as {\small $\epsilon_{ATE}=|\frac{1}{N}\sum_{v_i \in \mathbf{V}} \tau_i - \frac{1}{N}\sum_{v_i \in \mathbf{V}} \hat{\tau}_i|$} that gives absolute deviation from true average IDEs. 

{\textbf{Baselines}. We compare our proposed estimator, IDE-Net, with recent methods: DWR~\cite{zhao-arxiv22}, 1-GNN-HSIC~\cite{ma-aistats21}, NetEst~\cite{jiang-cikm22}, \newrevision{TNet~\cite{chen-icml24}}, ND~\cite{guo-wsdm20}, and MLP-CFR~\cite{shalit-icml17}. Dual Weighting Regression (DWR) explicitly addresses HPI with attention weights and the calibration step reweighs samples for pseudo-randomization. 1-GNN-HSIC implicitly addresses HPI by summarizing covariates of treated peers with 1-GNN~\cite{morris-aaai19} in addition to homogeneous exposure embedding, i.e., the fraction of treated peers. NetEst uses GCN~\cite{kipf-iclr16} with adversarial training and \newrevision{TNet uses targeted learning for doubly robust estimation of causal effects but} both methods consider the fraction of treated peers as exposure embedding. ND uses GCN for feature representation and CFR~\cite{shalit-icml17} as an estimator but does not account for interference. MLP-CFR considers MLP for feature representation and does not address interference while MLP-CFR-INT is like MLP-CFR but uses homogeneous exposure embedding.  We include a baseline BART-NET-INT, based on BART~\cite{hill-jcgs11}, a tree-based estimator, with mean peer covariates and homogeneous exposure embedding because it had a competitive performance on real-world datasets~\cite{guo-wsdm20}. For the ablation study, we use IDE-Net-INT as a baseline that only utilizes feature embeddings and considers exposure embedding like NetEst.} Figure \ref{fig:cane_em_pehe} in the appendix, compares the performance of the causal network motifs approach (related work considering HPI due to local structure) with IDE-Net. 
The Appendix \S \ref{sec:hyperparam} discusses the hyperparameter settings in detail. 
\subsection{Results}\label{sec:results}
Here, we discuss the main takeaways from our experiments 
and include additional results in the Appendix \ref{sec:addn_results}.

\textbf{Our proposed method is robust when both treatment effect modification and HPI are present}. 
In the first experiment, {we generate ground truth IDEs due to effect modifications based on unit attributes and unit's degree compared to peers by varying the values of coefficients $\tau_d$, $\tau_a$, and $\tau_n$ in the outcome model.} We enable HPI based on tie strength ($\tau_{Z_r}=1$), mutual connections ($\tau_{E_{-r}}=1$), attribute similarity ($\tau_{Z_{-i}}=1$), or a linear combination of these contexts. We vary the strength of peer exposure, i.e., HPI coefficient $\tau_p$, and run the experiment with three configurations of BA and WS networks with sparse to dense edge densities $15$ times for each parameter setting. Figure \ref{fig:syn_cane_all_pehe} shows the $\epsilon_{PEHE}$ with $95\%$ confidence interval for different network configurations averaged for all four HPI mechanisms. We observe our approach is robust to unknown HPI and effect modification but the baselines suffer high bias depending on network topology and edge density. {Most baselines incur high bias even for low HPI coefficient. 1-GNN-HSIC \newrevision{and TNet} perform well for low HPI coefficient and the dense WS network but poorly for higher HPI coefficient. Interestingly, the MLP-CFR baseline that ignores network features and interference has a competitive performance to other baselines for higher HPI coefficient which demonstrates the bias due to misspecified peer exposure.}
Figure \ref{fig:semi_cane_all_pehe} shows the performance of estimators for each underlying peer influence mechanism for two real-world networks. Our proposed approach has robust performance for all four mechanisms in both real-world networks. {Baselines that partially handle HPI and BART-NET-INT have better performances than baselines considering homogeneous influence or ignoring interference. \newrevision{While TNet performs well for low HPI coefficient in the BlogCatalog dataset, the performance for higher HPI coefficient and Flickr dataset is poor.} Interestingly, IDE-Net-INT has competitive performance for semi-synthetic data highlighting the strength of our feature embedding.}
Our method is robust for ATE estimation as well (Figure \ref{fig:cane_ate}).

\textbf{Ignoring HPI results in biased heterogeneous and average direct effects}. In the second experiment, we evaluate the estimation of constant direct effects $\tau_d=-1$ for a similar setup as the first experiment. Here, we have no effect modifications to investigate bias only due to HPI. In Figure \ref{fig:ane_pehe} and \ref{fig:ane_ate}, we can observe the robustness of IDE-Net for estimating heterogeneous and average direct effects for both synthetic and semi-synthetic data settings. All baselines except BART-NET-INT suffer high bias depending on network topology and edge density. This may be because BART-NET-INT has access to aggregated peer features but GNN-based methods have to learn the relevant representation automatically. Although IDE-Net is designed to handle heterogeneity and effect modifications, it performs well for constant effects.
Figure \ref{fig:cane_em_pehe} in \S \ref{sec:addn_results} includes an ablation study to support the robustness of our method for estimating direct effects when peer exposures are effect modifiers.
\vspace{-1em}
\section{Conclusion}
We motivated and investigated the problem of estimating individual direct effects under unknown heterogeneous peer influence mechanisms. We proposed flexible models (NSCM and NAGG) to identify potential heterogeneous contexts as relational variables and an estimator (\emph{IDE-Net}) for representation learning of these relational variables to estimate robust individual direct effects. 
While our framework is not a replacement for related work dealing with specific HPI~\cite{tran-aaai22,yuan-www21,qu-arxiv21,ma-kdd22}, it is general enough to capture their sources of heterogeneity. 
For instance, Figure \ref{fig:cane_em_pehe} in the Appendix compares \emph{IDE-Net} with the causal network motifs approach designed to address HPI due to local neighborhood structure. We discuss the concrete mappings to each work in the Appendix \ref{sec:ap-discussion}.  

Our work can be extended to peer and total effect estimation which would require addressing counterfactual peer exposures with unknown HPI contexts. 
Future investigations could relax the assumption of pre-treatment network structure and study heterogeneity in complex interventions on/or affecting network ties. The implications of our work include identifying subpopulations with heterogeneous effects. The potential societal impacts could include designing targeted interventions or discovering intervention policies that maximize desired utility in online or offline social networks.

\small{
\bibliography{references}
}
\onecolumn
\appendix
\section{Appendix}
\subsection{Results (continued...)}\label{sec:addn_results}
This section shows Figures \ref{fig:ane_pehe}, \ref{fig:ane_ate}, and \ref{fig:cane_ate} referenced and discussed in the main paper. These results, along with Figure \ref{fig:cane_pehe} in the main paper, establish the robustness of our proposed estimator when heterogeneous peer influence (HPI) acts as a nuisance or bias in the estimation of individual direct causal effects. For experimental settings in Figures \ref{fig:cane_pehe}–\ref{fig:cane_ate}, IDE (as defined in Equation \ref{eq:dir_eff}) is equivalent to "insulated" individual effects~\cite{jiang-cikm22} because influences from peers do not interact with the treatment assignments of individuals. The robustness of our method, IDE-Net, for estimating heterogeneous and average effects in the presence of HPI and effect modification due to contexts related to individual attributes and network position is demonstrated in Figures \ref{fig:cane_pehe} and \ref{fig:cane_ate}, respectively. Figures \ref{fig:ane_pehe} and \ref{fig:ane_ate} illustrate the effectiveness of our methods when the underlying IDEs are constant but the peer influence is heterogeneous.

Figure \ref{fig:cane_em_pehe} shows the results of a follow-up experimental setting when there is effect modification due to peer exposure and the true IDEs capture heterogeneous influences from peers as referenced (but not discussed) in \S \ref{sec:results}. We generate the outcome as $Y_i = f_Y(\tau_p \phi(X_{-i},\mathcal{Z}_i) + \tau_p \phi(X_{-i},\mathcal{Z}_i)X_i + ...)$, where "$...$" denotes other terms, $\tau_p$ is the peer exposure coefficient, and $\phi(X_{-i},\mathcal{Z}_i)$ capture strength of peer exposure due to some underlying mechanism. The true IDEs in this setting are $\tau_p \phi(X_{-i},\mathcal{Z}_i)$ and we evaluate how well the estimators estimate the IDE with the $\epsilon_{PEHE}$ metric. The IDEs in this setting are different from "insulated" individual causal effects~\cite{jiang-cikm22} due to the interaction with peer exposure. This experimental setting also acts as an \textbf{ablation study} to examine the contributions of our feature and exposure embeddings. In previous experimental settings, we used the baseline estimators and mostly default hyperparameters (see \S \ref{sec:hyperparam} for details) implemented in the NetEst~\cite{jiang-cikm22} and Network Deconfounder (ND)~\cite{guo-wsdm20} papers. Here, baselines also use counterfactual prediction modules and hyperparameters like our IDE-Net estimator. We use GCN-TARNet and GCN-CFR with causal network motifs~\cite{yuan-www21} as exposure embeddings for baselines instead of NetEst~\cite{jiang-cikm22}, which performs poorly for IDE estimation. 

Causal network motifs~\cite{yuan-www21} capture recurrent subgraphs with treatment assignments to represent heterogeneous peer exposure based on local neighborhood structure. Let the symbol $\bullet$ indicate the treated peer and $\circ$ indicate the peer in the control group. In the ego's neighborhood, we use normalized counts of dyads ($\{\circ\}, \{\bullet\}$), open triads ($\{\bullet~\bullet\}, \{\bullet~\circ\}, \{\circ~\circ\}$), and closed triads ($\{\bullet-\bullet\}, \{\bullet-\circ\}, \{\circ-\circ\}$) to construct the exposure embedding following~\citet{yuan-www21}. We compare our method IDE-Net (with smoothing regularization) against GCN-TARNet and GCN-CFR consisting causal motif fractions (MOTIFS) and  GCN-TARNet-INT that considers fraction of treated peers (HOM-INT). We also use DWR, \newrevision{TNet}, and 1-GNN-HSIC as baselines. Additionally, we include IDE-Net-INT (with smoothing regularization) as a baseline for the ablation study. We generate heterogeneous peer influence due to influence mechanisms like mutual connections, attribute similarity, tie strength, and peer degree, as discussed in the main paper. Additionally, we include structural diversity as an underlying influence mechanism that considers the number of connected components among treated peers as peer exposure~\cite{yuan-www21}. We use peer exposure coefficient $\tau_p=1$ for structural diversity and $\tau_p=20$ for other influence mechanisms. Figure \ref{fig:cane_em_pehe} shows the performance of proposed and baseline estimators for synthetic scale-free (BA) and small-world (WS) networks with different edge densities. Due to the computational cost of extracting exposure embeddings with causal network motifs, we use a WS graph with a cluster size of K = 15 instead of K = 300 in the previous settings.

We observe in Figure \ref{fig:cane_em_pehe} that causal motifs capture peer influences due to structural diversity (with open triads) and mutual connections (with closed triads) well. However, as expected, causal motifs cannot capture peer influences due to attribute similarity, tie strength, and peer degree in scale-free networks with a non-uniform degree distribution. IDE-Net on the other hand is robust for different underlying peer influence mechanisms and competitive to or better than causal motifs for structural diversity and mutual connection mechanisms. Estimators assuming homogeneous influence suffer high bias unless the network topology and influence mechanisms produce uniform influences (e.g., peer degree in small-world networks). Estimators assuming no network interference have a large bias in every setting when there is effect modification due to peer exposure.

\begin{figure}[ht]
    \centering
    \subfigure[Synthetic networks and one or more peer influence mechanisms]{
        \includegraphics[width=0.99\linewidth,keepaspectratio]{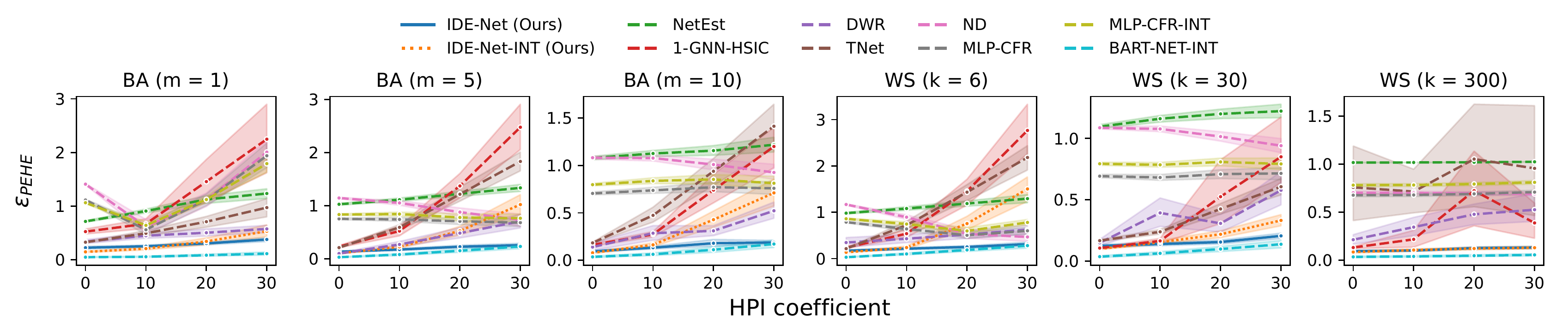}
    \label{fig:syn_all_pehe}
    \vspace{-1em}
    }
    
    \subfigure[Semi-synthetic networks and different peer influence mechanisms]{
        \includegraphics[width=0.99\linewidth,keepaspectratio]{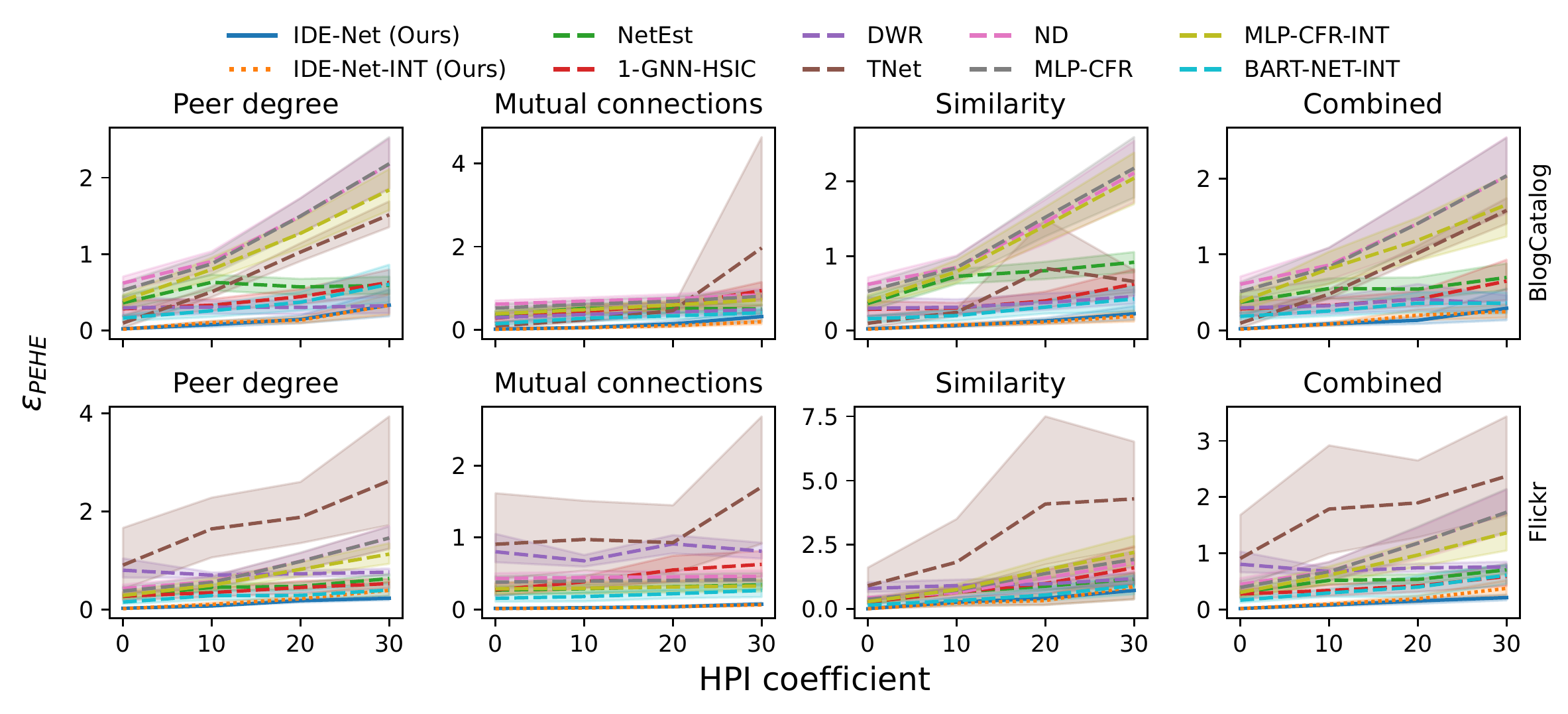}
    \label{fig:semi_all_pehe}
    \vspace{-1em}
    }
    \caption{$\epsilon_{PEHE}$ errors for \textbf{constant true individual direct effects} and heterogeneous peer influence.}
    \label{fig:ane_pehe}
    \vspace{-1em}
\end{figure}

\begin{figure}[th]
    \centering
    \subfigure[Synthetic networks and one or more peer influence mechanisms]{
        \includegraphics[width=0.99\linewidth,keepaspectratio]{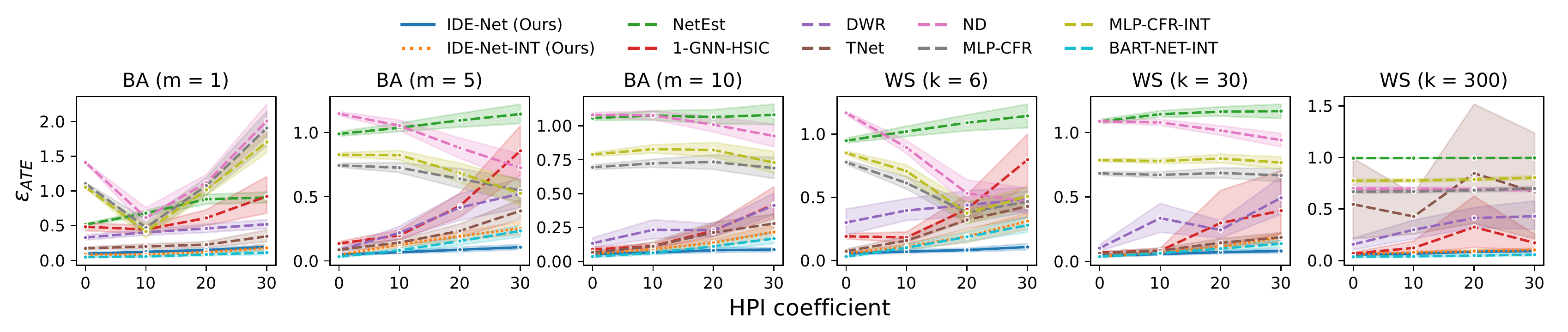}
    \label{fig:syn_all_ate}
    \vspace{-1em}
    }
    
    \subfigure[Semi-synthetic networks and different peer influence mechanisms]{
        \includegraphics[width=0.99\linewidth,keepaspectratio]{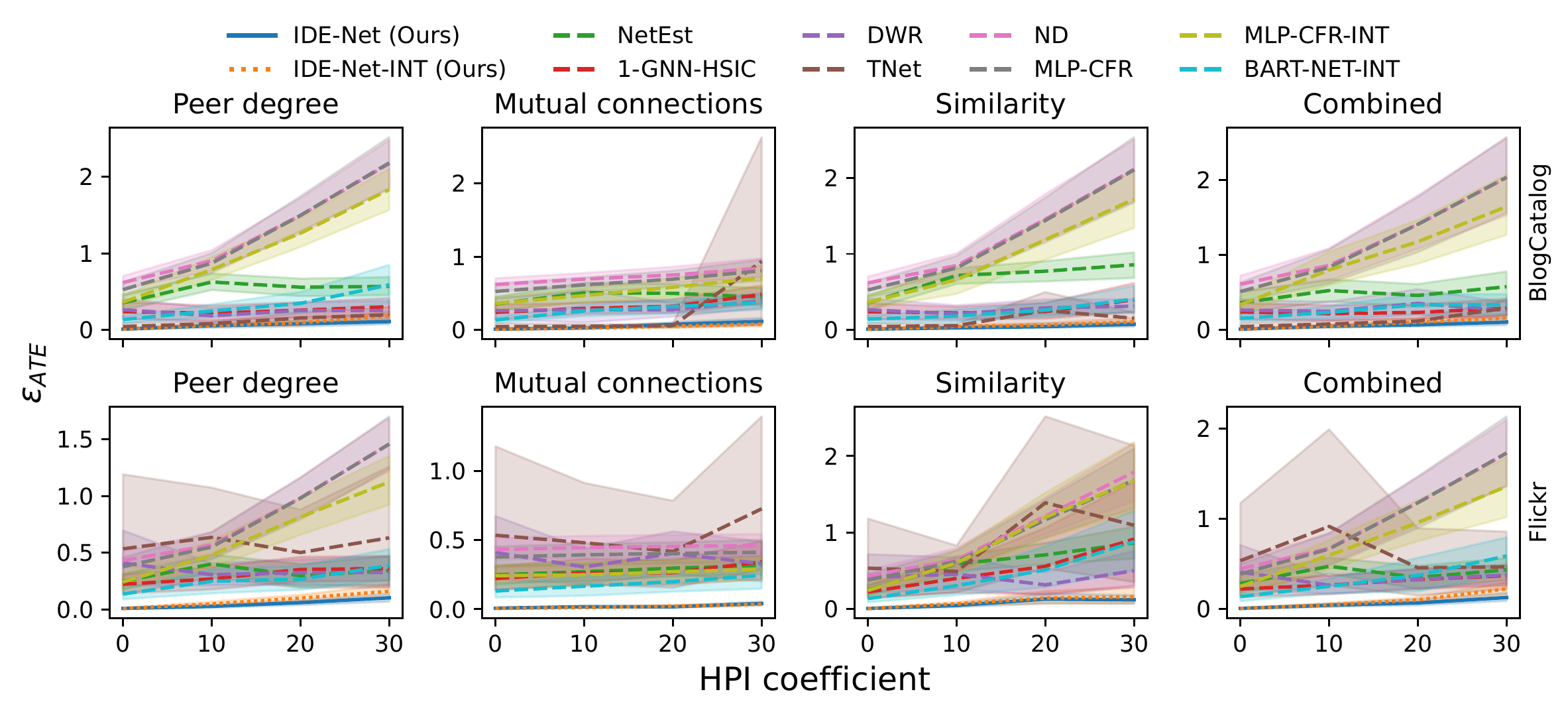}
    \label{fig:semi_all_ate}
    \vspace{-1em}
    }
    \caption{$\epsilon_{ATE}$ errors for \textbf{constant true individual direct effects} and heterogeneous peer influence.}
    \label{fig:ane_ate}
    \vspace{-1em}
\end{figure}

\begin{figure}[th]
    \centering
    \subfigure[Synthetic networks and one or more peer influence mechanisms]{
        \includegraphics[width=0.99\linewidth,keepaspectratio]{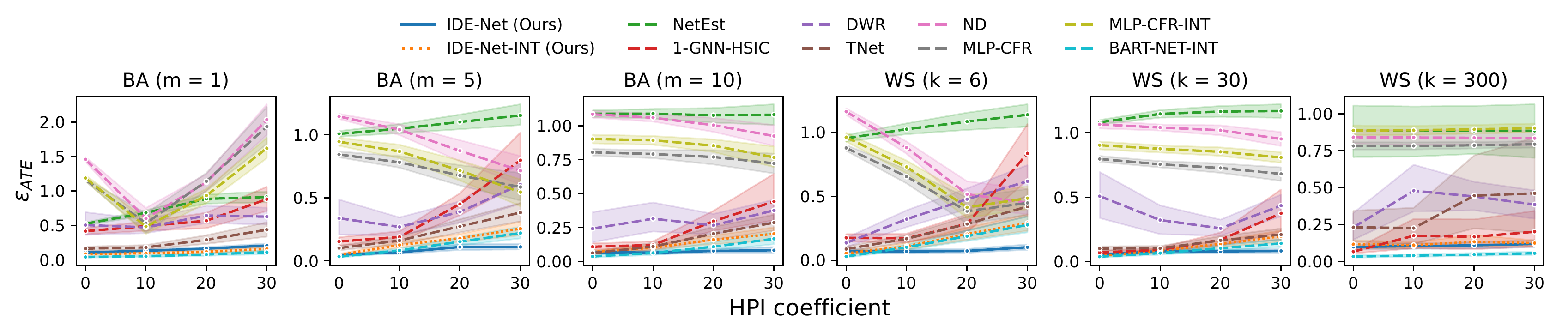}
    \label{fig:syn_all_ate_cane}
    \vspace{-1em}
    }
    
    \subfigure[Semi-synthetic networks and different peer influence mechanisms]{
        \includegraphics[width=0.99\linewidth,keepaspectratio]{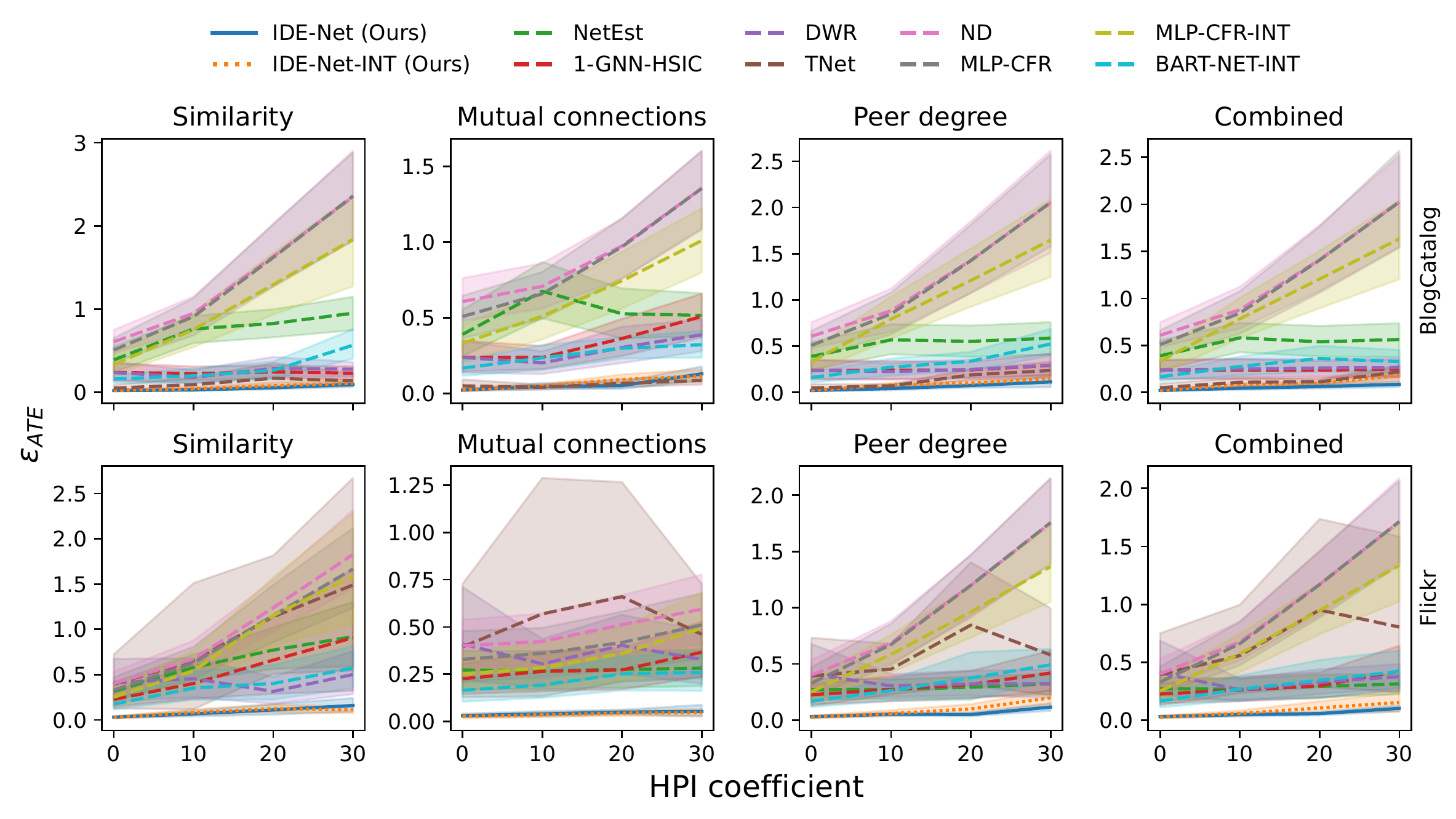}
    \label{fig:semi_all_ate_cane}
    \vspace{-1em}
    }
    \caption{$\epsilon_{ATE}$ errors for individual direct effects with {effect modification} and heterogeneous peer influence.}
    \label{fig:cane_ate}
    \vspace{-1em}
\end{figure}

\begin{figure}[th]
    \centering
    \includegraphics[width=\linewidth]{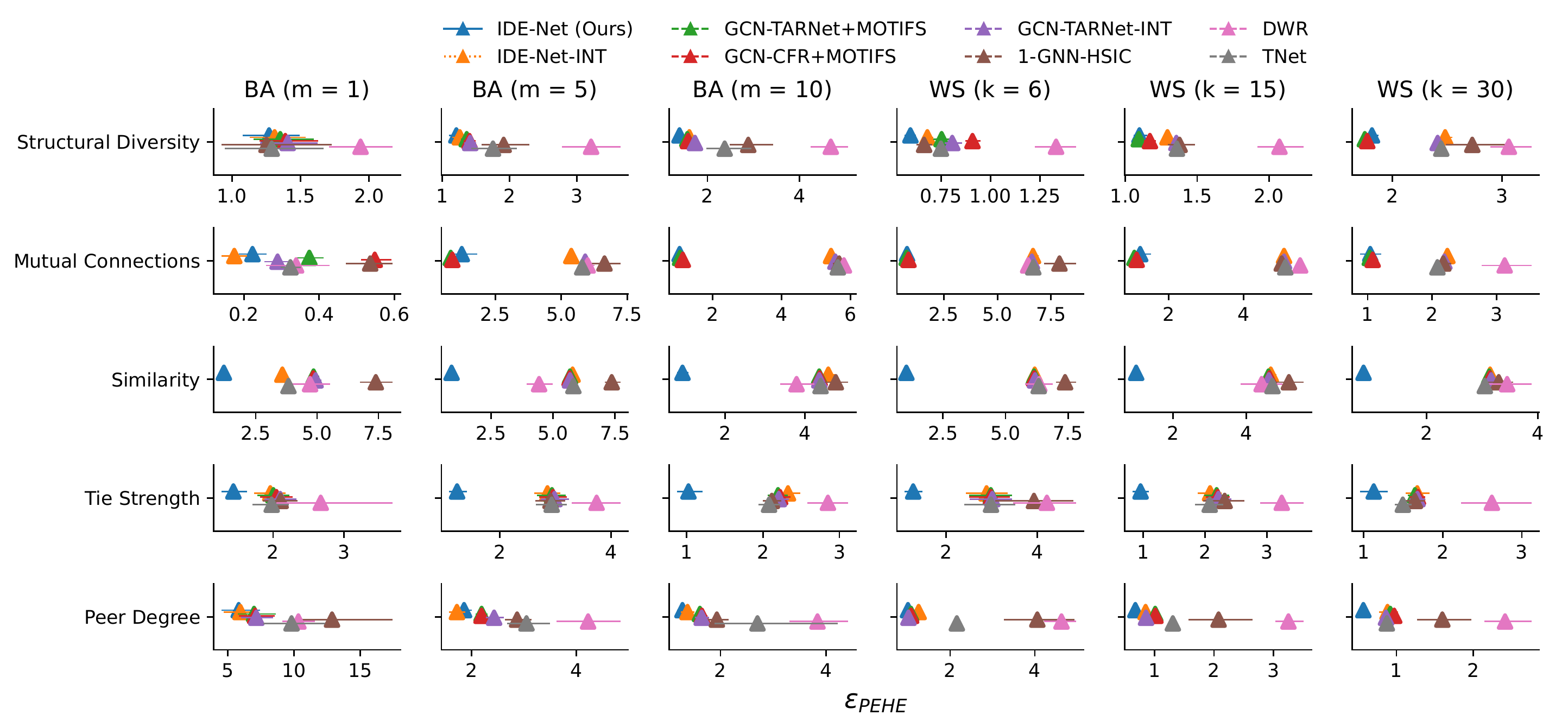}
    \caption{$\epsilon_{PEHE}$ errors, for different underlying heterogeneous peer influence mechanisms and synthetic networks with different topologies and edge densities, when \textbf{the peer exposure is an effect modifier} and defines individual effects. While causal motifs capture peer influences due to structural diversity and mutual connections well, they cannot capture peer influences due to attribute similarity, tie strength, and peer degree in scale-free networks with a non-uniform degree distribution. IDE-Net, on the other hand, is robust for different underlying peer influence mechanisms and competitive to causal motifs for structural diversity and mutual connection mechanisms.}
    \label{fig:cane_em_pehe}
\end{figure}

\clearpage

\subsection{Data generation}\label{sec:ap-data-gen}
First, we discuss the data generation for confounder $C_i$, context $Z_i$, and tie strength $Z_r$ for synthetic data as referenced in \S \ref{sec:exp_setup}. We generate $C_i \sim Beta(0.6,0.6)$ to model demographics predictive of political polarity, $Z_i \sim Categorical(K=5, \mathbf{p} \sim Dirichlet(5,5,5,5,5))$ to model membership in one of five demographic subgroups, and $Z_r \sim Uniform(1,10)$ to model friendship duration as tie strength. Here, the $Beta$ distribution with the parameters gives values in the range $(0,1)$ with a slightly polar distribution most values near $0$ and $1$ as fewer values toward $0.5$. The $Dirichlet$ distribution gives probabilities values that sum to $1$ for the $Categorical$ distribution.

\textit{Treatment model}. For observational data, a unit's treatment may be influenced by its own covariates as well as the treatments or covariates of other units, depending on the interference conditions. We generate treatment $X_i$ for unit $v_i$ as $X_i \sim \theta (a(\tau_c\mathbf{W}_x \cdot \frac{\sum E_r \odot C_{-i}}{\sum E_r} + (1-\tau_c)\mathbf{W}_x \cdot C_i))$, where $\theta$ denotes Bernoulli distribution, $a:\mathbb{R}\mapsto [0,1]$ is an activation function, $\tau_c\in[0,1]$ controls spillover influence from unit $v_i$'s peers, and $\mathbf{W}_x$ is the weight matrix.

\textit{Outcome model}. We define unit $v_i$'s outcome $Y_i$ as a function of unit's treatment ($\tau_d$), unit's attribute effect modifier ($\tau_a$), network effect modifier ($\tau_n$), heterogeneous influence from peers ($\tau_p, \tau_{Z_r},\tau_{Z_{-i}},\tau_{E_{-r}}$), confounder ($\tau_{c1}$), and random noise ($\epsilon \sim \mathcal{N}(0,1)$).
{\small
\begin{equation}
\begin{split}
Y_i= \tau_0 + \tau_{c1} C_i +\tau_d X_i + \tau_{a}X_iZ_i + \tau_{n} X_i\mathbb{1}[{(\textstyle \sum E_r)>g_p(E_{-r})}]+\epsilon + \\\tau_p X_{-i}\cdot g_n(\tau_{Z_r}.\phi(Z_r) +\tau_{Z_{-i}}.{\phi(E_r \odot rbf(Z_i, Z_{-i}))}+ \tau_{E_{-r}}.{\phi(E_r \odot g(E_r, E_{-r}))}),
\end{split}
\end{equation}
}
where $\tau_0$ is a constant, $g_p$ gives $p^{th}$ percentile value of peers' degree $g(E_{-r})$, $g_n(\vec{Z})=\frac{\vec{Z}}{\sum Z}$ gives normalized weight, $\phi(Z)$ is one of $\{Z^2, \sqrt{Z}, Z\}$, and $rbf$ is a radial basis function. The coefficients $\tau_a$ and $\tau_n$ control treatment effect modification based on node attributes (e.g., age group) and network location (e.g., degree centrality), respectively. We model variable influence from peers based on the strength of ties ($\tau_{Z_r}$), the attribute similarity and friendship existence ($\tau_{Z_{-i}}$), local structure captured by the count of mutual friends ($\tau_{E_{-r}}$), and peer exposure coefficient ($\tau_p$). 

The generated data is used in the treatment and outcome models. For synthetic data, we use $\tau_c=0.5$ in the treatment model to control spillover from peers' confounders. Similarly, we use weight matrix $\mathbf{W}_x=1$ and identity activation $a$ in the treatment model to generate treatment assignment $X_i$. For the outcome model, we use the square of tie strength and the square root of mutual friend count to introduce non-linear influences in the underlying model. The radial basis function in the outcome model uses $gamma=2$.

For the semi-synthetic data, we randomly assign 50-dimensional features of each unit to confounders or effect modifiers. Confounders are the features that affect both treatment and outcome whereas effect modifiers (except confounders) do not affect the treatment. We generate masks drawn from the Bernoulli distribution and weights drawn from the uniform distribution to get weights $\mathbf{W}_x$ and $\mathbf{W}_y$ for all 50-dimensional features, i.e.,
\begin{align*}
    M_{x} &\sim Bernoulli(0.6), \\
    M_{y} &\sim Bernoulli(0.6), \\
    W_{x} &\sim  M_x \times Uniform(-3,3), \text{ and } \\
    W_{y} &\sim  M_y \times Uniform(-3,3).
\end{align*}
The treatment model for semi-synthetic data is similar to that of synthetic data except we use the weights $\mathbf{W}_x$ and activation function $a$ is Sigmoid, i.e.,
\begin{align*}
    X_i \sim \theta (Sigmoid(\tau_c\mathbf{W}_x \cdot \frac{\sum E_r \odot C_{-i}}{\sum E_r} + (1-\tau_c)\mathbf{W}_x \cdot C_i)),
\end{align*}
where $C_{i}$ represent the 50-dimensional features. The outcome model is similar to the synthetic data but we multiply $C_i$ with $\mathbf{W}_y$ and we include features with $M_x=0$ and $M_y=1$ to $Z_i$.

\subsection{Feature and exposure embeddings}\label{sec:ap-mapping}
\textbf{Feature embeddings}. We propose GNN feature embeddings {\small $h_f$} by concatenating outputs of ego multi-layer perception (MLP) {\small $h_i$}, the peer graph convolution network (GCN) {\small $h_{-i}$}, and edge GCN {\small $h_r$} modules capturing relational variables {\small $\phi_i(Z_i)$}, {\small $\phi_{-i}(Z_{-i}, Z_r, E_r)$}, and {\small $\phi_{r}(Z_r, Z_{-r}, E_r, E_{-r})$}, respectively. Let {\small $\Theta$} denote MLP with ReLU, and $<\mathbf{v_i}, \mathbf{v_j}>$ be the indices of nodes with edges, i.e., non-zero $\mathbf{A}$, then, we define,
{\small
\begin{align*}
    \begin{split}
    &h_i = \Theta(\mathbf{Z_n}),\\
    &h_{-i}=\mathbf{D}^{-\frac{1}{2}}\overset{row}{\textstyle \sum}\mathbf{A} \odot \Theta(\mathbf{Z_n[v_j]} || \mathbf{Z_e}),\\
    &h'_{r}=\overset{row}{\textstyle \sum}\mathbf{A} \odot \Theta(\mathbf{Z_e}), h_r= h'_{r}||\mathbf{D}^{-\frac{1}{2}}\mathbf{A}h'_{r},\text{ and }\\ &h_f = h_i||h_{-i}||h_{r},
    \end{split}
\end{align*}
} where {\small $\mathbf{D}$} is a degree matrix used for normalization, {\small $\odot$} is index-wise product operator, and {\small $||$} is a concatenation operator. Here, {\small $h_{-i}$} is a normalized GCN and {\small $h_{r}$} is GCN after peer aggregation {\small ($h'_{r}$)}. Existing work has shown mean aggregation, i.e., normalization using {\small $\mathbf{D}^{-1}$} could make GNNs less expressive than sum aggregation~\cite{xu-iclr18}. However, using the sum aggregation may make representations unstable, especially in scale-free networks (e.g., online social networks) that have some nodes with high degrees, and we choose $\mathbf{D}^{-\frac{1}{2}}$ for normalization. 
We can extract l-hop deep features using $h^l = \mathbf{D}^{-\frac{1}{2}}\mathbf{A}\Theta(h^{l-1})$, where $h^1 = h_f$.

\textbf{Exposure embeddings}. Let $<\mathbf{v_i}, \mathbf{v_j}>$ be the indices of nodes with edges, i.e., non-zero $\mathbf{A}$, then, we define exposure embedding $h_e$ as a weighted mean, i.e.,
\begin{align*}
    h_e = \frac{\overset{row}{\textstyle \sum} \mathbf{A} \odot h_w \odot X_i[\mathbf{v_j}]}{\overset{row}{\textstyle \sum} \mathbf{A} \odot h_w}| h_w := \Theta(\mathbf{Z_e})||(\mathbf{A}\odot(\mathbf{A}\mathbf{A}^T))[\mathbf{v_i}, \mathbf{v_j}]||h_f[\mathbf{v_j}]||sim(h_f[\mathbf{v_i}], h_f[\mathbf{v_j}]),
\end{align*}
where $h_w \in \mathbb{R}^{m \times k}$ represents edge weights capturing underlying peer influence strength, [.] indicates indexing, and $sim(a,b)=e^{-(a-b)^2}$ captures similarity. The edge attributes, network connections, peer features, and similarity with peers are incorporated in $h_w$ as potential influence strengths.

\subsection{Discussion}\label{sec:ap-discussion}
\textbf{Connections to other HPI work}. 
We characterize existing work based on peer influence mechanisms as approaches considering homogeneous influence, known heterogeneous influence, or specific heterogeneous mechanisms. However, these works have different scopes and solve different problems. Our work highlights the importance of modeling unknown heterogeneous peer influence mechanisms for individual direct effect estimation. While our framework is not a replacement for existing approaches dealing with specific HPI~\cite{tran-aaai22,yuan-www21,qu-arxiv21,ma-kdd22}, it is general enough to capture their sources of heterogeneity. 
Moreover, these approaches could use our work to cover additional sources of heterogeneity. For instance, Figure \ref{fig:cane_em_pehe} in the Appendix compares IDE-Net with the causal network motifs approach designed to address HPI due to local neighborhood structure.
IDE-Net leverages graph convolution networks (GCN) and triangle counts ($AA^T$) for exposure embeddings to capture the representation of local neighborhoods. Other sources of heterogeneity due to variable susceptibilities of units ~\cite{tran-aaai22} and unit's covariates~\cite{qu-arxiv21} are captured by node attributes while group interactions captured with hypergraphs~\cite{ma-kdd22} could be modeled by IDE-Net with either node or edge attributes. 

\textbf{Scalability}. Extending our approach to make it scalable to larger networks is the next step for practical applications. A potential solution could rely on capturing 1-hop neighborhoods for each node vs. needing to feed the full network, as the current solution does. Let $N$, $M$, and $D$, respectively denote the number of nodes, edges, and the dimension of features or representations. The scalability of the current approach, like GCN, depends on time complexity $O(ND^2)$ for feature transformation and $O(MD)$, the number of edges (mostly linearly given constant dimension) for peer aggregation because we use sparse matrix multiplication whose complexity depends on the number of non-zero elements in the adjacency matrix ($A$). However, one term $AA^T$ in the exposure mapping introduced to explicitly capture structural property (number of walks of the path two between $v_i$ and $v_j$) incurs time complexity of $O(MN)$. This time complexity and overall space complexity can be improved by using batches of ego-network and ongoing work to apply our framework to large real-world networks for discovering populations with heterogeneous effects is addressing the scalability. The scope of the current work focuses on evaluating different network and data generation settings for IDE estimation under HPI.

\subsection{Preliminaries on Relational Causal Model (RCM)}\label{sec:rcm-bg}
{The relational causal model (RCM) framework provides a principled way to define random variables for causal modeling and reasoning in relational settings by accounting for relationships among units in addition to the attributes.
We focus on relational domains like social networks or contact networks, as depicted in Figure \ref{fig:eg1-sn}, abstracted by a \textbf{\textit{relational schema}} {\small $\mathcal{S} = (E,R, \mathcal{A})$} comprising an entity class {\small $E$} (e.g., User), a relationship class {\small $R=<E, E>$} (e.g., Friend) with many-to-many cardinality between the entities, entity attributes {\small $\mathcal{A}(E)$}, and relationship attributes {\small $\mathcal{A}(R)$}. The relational schema's partial instantiation, known as a \textbf{\textit{relational skeleton}}, specifies the entity instances (e.g., user nodes) and relationship instances (e.g., friendship edges) in the domain.}\\
{
A \textbf{\textit{relational random variable}} is defined from the perspective of an \textit{item class} (e.g., User or Friend) and consists of a \textbf{\textit{relational path}} and an attribute reached following the path. A relational path {\small $P=[I_j,...,I_k]$} is an alternating sequence of entity and relationship classes and captures path traversal in the schema, where \textit{item class} {\small $I\in\{E,R\}$} is either an entity or a relationship class. For the social network example, the relational paths capture the self, i.e., \texttt{[U]}, and friendship paths from the perspective of a user such as immediate friends, i.e., \texttt{[U,F,U]}, and friends of friends, i.e., \texttt{[U,F,U,F,U]}. The relational variables capture the self's attributes like stance \texttt{[U].St}, attributes of immediate friends like political affiliation \texttt{[U,F,U].Aff}, and so on. The relational variables with a trivial path of length one (e.g., \texttt{[U].St}) are known as \textit{canonical variables}.
The values of such relational random variables are determined using the concepts of \textbf{\textit{terminal set}} and \textbf{\textit{path semantics}} for traversal. A relational path's beginning and ending item classes are called the base item class and terminal item class, respectively. A terminal set comprises instances of the terminal item class reached from an instance of the base item class after traversing a relational path. For example, the terminal sets for a base user \texttt{Ava} and paths \texttt{[U]} and \texttt{[U,F,U]}, respectively, consists of the self, i.e., \{\texttt{Ava}\} and immediate friends, i.e., \{\texttt{Bob, Dan, Fin}\}. The path traversal follows \textit{bridge-burning semantics} (BBS)~\cite{maier-thesis14} where nodes at lower depth are not visited again. A breadth-first traversal preserves BBS by exploring all nodes at each depth and not visiting explored nodes. For example, the terminal set for the base user \texttt{Ava} and path \texttt{[U,F,U,F,U]} consists only \{\texttt{Eve}\} excluding the immediate friends and the self.
The value of a relational variable $P.X$ is a multi-set, i.e., a set that allows repeated elements, obtained by collecting values of the terminal item class's attribute $X$ in the terminal set of path $P$. For example, the values of the relational variables \texttt{[U].St} and \texttt{[U,F,U].Aff} for base user \texttt{Ava} are \{\texttt{Anti}\} and \{\texttt{Right,Right,Left}\}, respectively. These concepts apply to relational variables with edge attributes (e.g., \texttt{[U,F].Dur} with value \{\texttt{Long,Short,Short}\} for base user \texttt{Ava}) and from the perspective of relationship class (e.g., \texttt{[F,U].Aff} with value \{\texttt{Left,Right}\} for base friendship edge \texttt{Ava-Bob}).
}

{Similar to the structural causal model (SCM), the data-generating assumptions of node and edge attributes can be encoded by using \textbf{\textit{relational dependencies}} to link relational random variables to their causes. A relational dependency $P.X \rightarrow [I].Y$, where $X \ne Y$, links a canonical relational variable $[I].Y$ to its cause $P.X$ which is another relational variable of arbitrary length with the same base item class $I$. The use of dependencies with canonical variables enables capturing data generation from the perspective of an entity or a relationship class. For example, the dependency \texttt{[U,F,U].Aff}$\rightarrow$\texttt{[U].St} captures that the stance of a user is influenced by the political affiliations of immediate friends. A \textbf{\textit{relational causal model}} (RCM) {\small $\mathcal{M}(\mathcal{S}, \mathbf{D}, \mathcal{P})$} encapsulates a schema {\small $\mathcal{S}$}, a set of relational dependencies {\small $\mathbf{D}$}, and optionally parameters {\small $\mathcal{P}$}. For all canonical variables $[I].Y$ in the schema, the parameters {\small $\mathcal{P}$} consists of distributions $P([I].Y| pa([I].Y, \mathbf{D}))$, where $pa([I].Y, \mathbf{D}):=\{P.X| P.X \rightarrow [I].Y \in \mathbf{D}\}$ is a set of relational variables that are causes of $[I].Y$ in $\mathbf{D}$.
}\\
{
A graphical model constructed to reflect dependencies $\mathbf{D}$ in the RCM is invalid because it fails to capture all conditional independence in the data. RCM only defines a template for how dependencies should be applied to data instantiations. A ground graph is a directed graphical model with nodes as attributes of entity instances (like {\small \texttt{Bob.Aff}} and {\small \texttt{Ava.St}})  and relationship instances (like {\small \texttt{<Bob-Ava>.Dur}}), and edges as dependencies between the nodes applied from RCM. An acyclic ground graph has the same semantics as a standard graphical model~\cite{getoor-isrl07}, and it facilitates reasoning about the conditional independencies of attribute instances using standard d-separation~\cite{pearl-book88}. However, it is more desirable to reason about dependencies that generalize across all ground graphs by using expressive relational variables and abstract concepts. A higher-order representation known as an \textbf{\textit{abstract ground graph}} (AGG)~\cite{maier-thesis14,lee-thesis18} has been developed for this purpose. An AGG is a graphical model with relational variables as nodes, and the edges include dependencies $\mathbf{D}$ in the RCM and additional extended dependencies to capture all conditional independence. For example, dependency \texttt{[U,F,U].Aff}$\rightarrow$\texttt{[U].St}, suggesting influence from immediate friends, is extended to other dependencies like \texttt{[U].Aff}$\rightarrow$\texttt{[U,F,U].St} and \texttt{[U,F,U,F,U].Aff}$\rightarrow$\texttt{[U,F,U].St}. These extended dependencies are automatically identified by the \textit{extend operator}, which is described in the Appendix. AGG is theoretically shown to be sound and complete in abstracting all conditional independence in the ground graph for social networks~\cite{maier-thesis14}.
}
{
AGG can be used for causal reasoning, including using d-separation and do-calculus similar to SCM~\cite{pearl-book09} for expressing counterfactual term in Equation \ref{eq:dir_eff} in terms of observational or experimental distribution and identifying heterogeneous contexts $\mathcal{Z}$ in terms of relational variables.
}

\subsection{Network Structural Causal Model and Network Abstract Ground Graph (NAGG)}\label{sec:nscm-detail}

Here, we formally see the new theoretical properties due to Definitions \ref{def:skeleton} and \ref{def:ex_ind} first. Then, we present NSCM and NAGG in more detail including the theoretical guarantees of soundness and completeness. Definitions \ref{def:skeleton} and \ref{def:ex_ind} distinguish between relationship instances {\small $\sigma(R)$} and edges {\small $\mathcal{E}$}. \note{As depicted in Figure \ref{fig:eg2-paths}}, the relationship instances, \note{indicated by diamond shapes}, connect all nodes completely, but the {\small $Exists$} attribute indicates the presence of an edge.
\newchange{Next, we examine how these extended definitions alter RCM concepts like relational paths and their terminal sets and highlight some theoretical properties.}

We refer the network {\small $G(\mathcal{V}, \mathcal{E})$} as a single entity type and single relationship type (SESR) model~\cite{maier-thesis14} with a schema {\small $\mathcal{S} = (E,R, \mathcal{A})$}.
Formally, a \textit{terminal set} for a relational path $P=[I_j,...,I_k]$ and a base item {\small $i_j\in \sigma(I_j)$}, denoted by {\small $P|_j^\sigma$}, is a set of terminal items {\small $\mathbf{i_k} \subseteq \sigma(I_k)$} reachable after traversal according to \textit{bridge burning semantics} (BBS)~\cite{maier-thesis14} in a skeleton {\small $\sigma$}. \note{As shown in Figure \ref{fig:eg2-paths}}, BBS traverses {\small $\sigma$} along {\small $P$} in breadth-first order, beginning at {\small $i_j$}, to obtain terminal set {\small $P|_j^\sigma$} as tree leaves at level {\small $|P|$}. A relational path {\small $P$} is valid if its terminal set for any skeleton is non-empty, i.e. {\small $\exists_{\sigma \in \Sigma_{\mathcal{S}}},\exists_{j \in \sigma(I)},P|_j^\sigma \ne \emptyset$}.
The proofs are presented in \S \ref{appendix:proof} and are based on Definitions \ref{def:ex_ind} and \ref{def:skeleton} and BBS exploration. Lemma \ref{lemma:int} allows us to rule out the complexity of intersecting paths~\cite{maier-thesis14,lee-uai15} in the SESR model.
{As illustrated in Figure \ref{fig:eg2-paths}, the entire network is represented from the perspective of users by the terminal sets of paths {\small $[E],[E,R],[E,R,E],$} and {\small $[E,R,E,R]$} mapping to an ego, the ego's relationships, other nodes, and relationships between other nodes. This differs from RCM's SESR model~\cite{arbour-kdd16}, where relational paths can be infinitely long and capture entity or relationship instances at n-hops. In our case, we allow n-hop relations to change due to interventions or over time, and it is reflected in the {\small $Exists$} attribute of terminal sets of path {\small $[E, R]$} and {\small $[E, R, E, R]$}.

\textbf{}{Network structural causal model (NSCM)}.
A \textit{network structural causal model} {\small $\mathcal{M}(\mathcal{S}, \mathbf{D}, \mathbf{f})$} encapsulates a schema {\small $\mathcal{S}$}, a set of relational dependencies {\small $\mathbf{D}$}, and a set of functions {\small $\mathbf{f}$}. The functions relate each canonical variable {\small $[I].Y$} to its causes, i.e., {\small $[I].Y = f_{Y}(pa([I].Y, \mathbf{D}), \epsilon_Y)$}, where {\small $f_{Y} \in \mathbf{f}$}, {\small $pa([I].Y, \mathbf{D}):=\{P.X|P.X \rightarrow [I].Y \in \mathbf{D}\}$}, and {\small $\epsilon_Y$}} is an exogenous noise.
Generally, the functions {\small $\mathbf{f}$} in the model {\small $\mathcal{M}$} are unknown and estimated from the data. An NSCM to encode that a user's and her friends' affiliations affect the user's stance needs three dependencies {\small \texttt{[U].Aff}$\rightarrow$\texttt{[U].St}}, {\small \texttt{[U,F,U].Aff}$\rightarrow$\texttt{[U].St}}, and {\small \texttt{[U,F].Ex}$\rightarrow$\texttt{[U].St}}. 


NSCM, like RCM, defines a template for how dependencies should be applied to data instantiations. A realization of NSCM {\small $\mathcal{M}$} with a relational skeleton {\small $\sigma$} is referred to as a ground graph {\small $GG_{\mathcal{M}\sigma}$}~\cite{maier-thesis14} (definition in \S \ref{appendix:proof}). It is a directed graph with vertices as attributes of entity instances (like {\small \texttt{Bob.Aff}} and {\small \texttt{Ava.St}})  and relationship instances (like {\small \texttt{<Bob-Ava>.Dur}}), and edges as dependencies between them.
An acyclic ground graph has the same semantics as a standard graphical model~\cite{getoor-isrl07}, and it facilitates reasoning about the conditional independencies of attribute instances using standard d-separation~\cite{pearl-book88}. However, it is more desirable to reason about dependencies that generalize across all ground graphs by using expressive relational variables and abstract concepts. A higher-order representation known as an abstract ground graph (AGG)~\cite{maier-thesis14,lee-thesis18} has been developed for this purpose. Relational d-separation~\cite{maier-thesis14} is the standard d-separation applied on AGG.




\textbf{}{Network abstract ground graph (NAGG)}.
A network abstract ground graph ({\small $NAGG_{\mathcal{M}B}$}), for the NSCM {\small $\mathcal{M}$} and a perspective {\small $B \in I$}, is a representation to enable reasoning about relational d-separation consistent across all ground graphs using relational variables ({\small $RV$}) and edges between them ({\small $RVE$}). An {\small $RV$} is a set of all relational variables of the form {\small $[B,...,I_k].X$}, with {\small $|RV|$} dependent on {\small $B$} for the SESR model (Proposition \ref{prop:paths}). 
An {\small $RVE\subset RV \times RV$} is a set of edges between pairs of relational variable obtained after applying \textit{extend}~\cite{maier-thesis14} operator on relational dependencies, i.e., {\small $RVE=\{[B,...,I_k].X\rightarrow[B,...,I_j].Y|[I_j,...,I_k].X \rightarrow [I_j].Y \in \mathbf{D}\}$}, where  {\small $[B,...,I_k] \ \in extend([B,...,I_j], [I_j,...,I_k])$} and {\small $[B,...,I_k]$} is a valid path (Proposition \ref{prop:paths}).

A relational model is acyclic if its \textit{class dependency graph} with vertices as attribute classes and edges as dependencies between them 
is acyclic~\cite{getoor-isrl07}.
For simplicity of exposition, we will assume the NSCM model {\small $\mathcal{M}(\mathcal{S}, \mathbf{D})$} is acyclic.
The relational d-separation queries can be answered by applying standard d-separation on {\small $NAGG_{\mathcal{M}B}$} with vertices {\small $RV$} and edges {\small $RVE$}. We extend the theoretical properties of AGG~\cite{maier-thesis14} to NAGG with relationship existence uncertainty, latent attributes, and selection attributes. Let sets {\small $\mathbf{L}$} and {\small $\mathbf{S}$} also include latent and selection relational variables. The proofs, presented in \S \ref{appendix:proof}, follow the ~\citet{maier-thesis14}'s proof of soundness and completeness considering the implications of selection and latent variables.
\begin{theorem}\label{th:abstraction}
For every acyclic NSCM {\small $\mathcal{M}$} and perspective {\small $B \in \{E,R\}$}, the {\small $NAGG_{\mathcal{M}B}$} is sound and complete for all ground graphs {\small $GG_{M\sigma}$} with {\small $\sigma \in \Sigma_{\mathcal{S}}$}.
\end{theorem}
\begin{corollary}\label{th:acyclicity}
For every acyclic NSCM {\small $\mathcal{M}$} and perspective {\small $B \in \{E,R\}$}, the {\small $NAGG_{\mathcal{M}B}$} is directed and acyclic.
\end{corollary}
\begin{theorem}\label{th:rel-d-sep}
Relational d-separation is sound and complete for NAGG. Let {\small $\mathcal{M}$} be an acylic NSCM and let {\small $\mathbf{X}$}, {\small $\mathbf{Y}$}, and {\small $\mathbf{Z}|\mathbf{L} \not\in \mathbf{Z} \wedge \mathbf{S} \in \mathbf{Z} \wedge \mathbf{L} \cup \mathbf{S} \in \mathcal{A}$} be three distinct sets of relational variables for perspective {\small $B \in \{E,R\}$} defined over schema {\small $\mathcal{S}$}. Then, {\small $\mathbf{X}$} and {\small $\mathbf{Y}$} are d-separated by {\small $\mathbf{Z}$} on the {\small $NAGG_{\mathcal{M}B}$} if and only if for all skeletons {\small $\sigma \in \Sigma_{\mathcal{S}}$} and for all {\small $b \in \sigma(B)$}, {\small $\mathbf{X}|_b$} and {\small $\mathbf{Y}|_b$} are d-separated by {\small $\mathbf{Z}|_b$} in ground graph {\small $GG_{\mathcal{M}\sigma}$}.
\end{theorem}

With these theoretical properties, NAGG enables reasoning about the identification of causal effects in networks using counterfactual reasoning with the SCM~\cite{pearl-book09}.

\subsection{Proofs for NSCM and NAGG} \label{appendix:proof}

The proof for Lemma \ref{lemma:int}, similar to ~\citet{arbour-kdd16}, is as follows:
\begin{proof}
Proof by contradiction. Let $P=[I,...,I_k]$ and $Q=[I,...,I_l]$ be two distinct relational paths sharing base item class $I \in \{E,R\}$. Assume for contradiction that $P$ and $Q$ intersect such that their terminal sets $P|_{i\in \sigma(I)} \cap Q|_{i\in \sigma(I)} \ne \phi$. Since we know both $P$ and $Q$ share base item class and are two different paths, one path is a prefix of another path for the SESR model. According to the bridge burning semantics, terminal set of prefix path is not visited again by the other path. This contradicts with the previous assumption. Therefore $P$ and $Q$ do not intersect.
\end{proof}

The proof for Proposition \ref{prop:paths} is as follows:
\begin{proof}
Let $\mathcal{S}=(E,R,\mathcal{A})$ be a relational schema for a single-entity and single-relationship (SESR) model with arbitrary skeleton $\sigma \in \Sigma_{\mathcal{S}}$. For enabling relationship existence uncertainty, Definition \ref{def:skeleton} defines relational skeleton in terms of entity instances $\sigma(E)$ and fully-connected relationship instances $\sigma(R) = \Sigma_{\sigma(E)}$. Let $N=|\sigma(E)|$, then there are $|\sigma(R)| = \frac{N\times N-1}{2}$ relationship instances connecting each entity instance with other entities instances. Let us consider an arbitrary path $P=[I_j,...,I_k]$ to be traversed under bridge burning semantics (BBS).
\begin{claim}\label{claim:pathE}
$P|I_j \in E$ has a maximum length of valid relational path of four.
\end{claim}
\begin{proof}
As illustrated in Figure \ref{fig:eg1-sn}, the terminal set for $P|_{i_j\in \sigma(E)}$ has 1 entity instance for path $E$, $N-1$ relationship instances for path $[E,R]$, $N-1$ entity instances for path $[E,R,E]$, and remaining $|\sigma(R)|-(N-1)$ relationship instances for path $[E,R,E,R]$. By this point all entity and relationship instances are visited and path $[E,R,E,R,E]$ produces an empty terminal set, making it an invalid path.
\end{proof}
\begin{claim}\label{claim:pathR}
$P|I_j \in R$ has a maximum length of valid relational path of five.
\end{claim}
\begin{proof}
Similar to Claim \ref{claim:pathE}, the terminal set for $P|_{i_j \in \sigma(R)}$ has 1 relationship instance for path $R$ and two entity instances for path $[R,E]$. Next, the terminal set for path $[R,E,R]$ consists of $2\times N-2$ relationship instances. The terminal sets from previous level connect to remaining $N-2$ entity instances for path $[R,E,R,E]$. Finally, the terminal set for path $[R,E,R,E,R]$ traverses remaining relationship instances. By this point all entity and relationship instances are visited and path $[R,E,R,E,R,E]$ produces empty terminal set, making it an invalid path.
\end{proof}
\end{proof}

Here, we present formal definition of ground graph following \citet{maier-thesis14}'s definition.
\begin{definition}[Ground graph]
A \textit{ground graph} $GG_{\mathcal{M}\sigma}(V,E)$ is a realization of the NSCM $\mathcal{M}(\mathcal{S},\mathbf{D},\mathbf{f})$ for skeleton $\sigma$ with nodes $V=\{i.X|i \in \sigma(I) \wedge X \in \mathcal{A}(I)\}$ and edges $E=\{i_k.X \rightarrow i_j.Y|[I_j,...,I_k].X\rightarrow [I_j].Y \in \mathbf{D}\wedge i_k \in [I_j,...,I_k]|_{i_j}^\sigma \wedge \{i_k.X,i_j.Y\}\in V\}$.
\end{definition}

The proof for Theorem \ref{th:abstraction} follows the proof by \citet{maier-thesis14} and is provided below:

\begin{proof}
Let $\mathcal{M(\mathcal{S}, \mathbf{D})}$ be an arbitrary acyclic NSCM and let $B \in \{E,R\}$ be an arbitrary perspective. As a consequence of Lemma \ref{lemma:int} and Proposition \ref{prop:paths}, $NAGG_{\mathcal{M}B}$ does not have intersection variables and contains limited relational variables. Let $P_k = [B,...,I_k]$ and $P_j=[B,...,I_j]$ be two arbitrary valid paths for perspective $B$.

\textbf{Soundness}: To prove the soundness of $NAGG_{\mathcal{M}B}$, we must show that for every edge $P_k.X \rightarrow P_j.Y$ in the $NAGG_{\mathcal{M}B}$, there exist a corresponding edge $i_k.X \rightarrow i_j.Y|\exists_{b \in \sigma(B)}, i_k \in P_k|_b \wedge i_j \in P_j|_b$ in the ground graph $GG_{\mathcal{M}\sigma}$ for some skeleton $\sigma \in \Sigma_{\mathcal{S}}$.
By construction, for the perspective $B=E$, any skeleton with $N \ge 3$, where $N=|\sigma(E)|$, has all possible paths $[E], [E,R], [E,R,E],$ and $[E,R,E,R]$ with non-empty terminal sets, making them valid paths (Proposition \ref{prop:paths}). Similarly, for $B=R$, any skeleton with $N \ge 4$ has all possible paths $[R], [R,E], [R,E,R], [R,E,R,E],$ and $[R,E,R,E,R]$ with non-empty terminal sets, making them valid paths. This shows the existence of $b \in \sigma(B)$ with $i_k \in P_k|_b \wedge i_j \in P_j|_b$ for some skeleton $\sigma \in \Sigma_{\mathcal{S}}$.

Assume for contradiction, there exists no edge $i_k.X \rightarrow i_j.Y$ in any ground graph i.e., $\forall_{\sigma \in \Sigma_{\mathcal{S}}} \forall_{b \in \sigma(B)} \forall_{i_k \in [B,...,I_k]|_b} \forall_{i_j \in [B,...,I_j]|_b}, i_k.X \rightarrow i_j.Y \not\in GG_{\mathcal{M}\sigma}$. However, according to the definition of $NAGG_{\mathcal{M}B}$, if the relational edge $[B,...,I_k].X \rightarrow [B,...,I_j].Y$ is present, then there must be a dependency $[I_j,...,I_k].X \rightarrow [I_j].Y \in \mathbf{D}$ where $[B,...,I_k]$ $\in extend([B,...,I_j],[I_j,...,I_k])$. This implies the ground graphs must have an edge such that: $$\forall_{\sigma \in \Sigma_{\mathcal{S}}} \forall_{i_j \in \sigma(I_j)} \forall_{i_k \in [I_j,...,I_k]|_{i_j}}, i_k.X \rightarrow i_j.Y \in GG_{\mathcal{M}\sigma}.$$
Since $[B,...,I_j]$ is a valid path, we know $\exists_{i_j \in \sigma(I_j)},i_j \in [B,...,I_j]|_b$. To contradict the previous assumption, we prove the following Lemma about sufficient and necessary conditions for the soundness of $extend$~\cite{maier-thesis14} operator:
\begin{lemma}\label{lemma:extend}
If $[B,..,I_k] \in extend([B,...,I_j], [I_j,...,I_k]),$ then the following must be true: $$\exists_{i_j \in [B,...,I_j]|_b}\exists_{i_k \in [I_j,...,I_k]|_{i_j}}, i_k \in [B,...,I_k]|_b \wedge \forall_{S \in \mathbf{P_B}\setminus \{[B,...,I_k]\}}, i_k \not \in S|_b,$$ where $\mathbf{P_B}$ is a set of valid paths for the perspective B.
\end{lemma}
\begin{proof}
(1) Sufficient condition~\cite{lee-thesis18}: $i_k \in [B,...,I_k]|_b$.

The proof follows directly from definition of $extend$ and \citet{maier-thesis14}'s Lemma 4.4.1. Let $P_{o}=[B,...,I_j]$ and $P_{e}=[I_j,...,I_k]$ denote original and extension paths. The $extend$ function is defined as follows:
\begin{equation*}
\begin{split}
        extend(P_o, P_e)=\{P=P_o^{1:n_o-i+1}+P_e^{i+1:n_e}|i\in pivots(reverse(P_o),P_e) \\ \wedge isValid(P)\},
        \text{ with }pivots(P_1,P_2)=\{i|P_1^{1:i},P_2^{1:i}\},
\end{split}
\end{equation*}
where $n_o=|P_o|$, $n_e=|P_e|$, $P^{i:j}$ is i-inclusive and j-inclusive path slicing, $+$ denotes path concatenation, and $reverse$ is a function that reverses the order of a path. Let $c \in pivots(reverse(P_o),P_e)$ be a value obtained after applying the $pivots$ function, then there are two subcases:

(a) $c=1$ and $P=[B,...,I_j,...,I_k]$. Since $P$ is a valid path, the terminal set is not empty i.e., $\exists_{i_k \in \sigma(I_k)}, P|_b$.
As $I_j$ is in the intermediate path, $\exists_{i_j \in [B,...,I_j]|_b}\exists_{i_k \in [I_j,...,I_k]|_{i_j}}, i_k \in P|_b$ for $P|_b$ to be non-empty.

(b) $c>1$ and $P=[B,...,I_c,...,I_k]$. We must have $P_o=[B,...,I_c,...I_j]$ and $P_e=[I_j,...,I_c,...,I_k]$ for pivot $c>1$ to produce $P$.
Similar to case (a), since $P$ is a valid path, the terminal set is not empty i.e., $\exists_{i_k \in \sigma(I_k)}, P|_b$. Also, we know from case (a), $\exists_{i_c \in [B,..,I_c]|_b}\exists_{i_k \in [I_c,...,I_k]|_{i_c}}, i_k \in P|_b$. We then must show that $\exists_{i_j \in [I_c,...,I_j]|_{i_c}}, i_c \in [I_j,...,I_c]|_{i_j}$. The fully connected skeleton with $N=|\sigma(E)|$ entity instances and $|\sigma(R)|=\frac{N\times N-1}{2}$ relationship instances according to Definition \ref{def:skeleton} allows traversal from $i_c$ to $i_j$ and vice versa.

(2) Necessary condition~\cite{lee-thesis18}: $\forall_{S \in \mathbf{P_B}\setminus \{[B,...,I_k]\}}, i_k \not \in S|_b$

\citet{lee-thesis18} points out this necessary condition is required to ensure the parents of $[B,...,I_j].Y$ should not be redundant and propose a method called \textit{newextend} for handling both necessary and sufficient conditions. Here, we show that the $extend$ method satisfies necessary condition as well for $NAGG_{\mathcal{M}B}$. Suppose for contradiction the necessary condition is false. This means the item $i_k$ appears in the terminal sets of $[B,...,I_k]$ (sufficient condition proved before) as well as some other path from the perspective of item $b$.
However, due to the BBS traversal and Lemma \ref{lemma:int}, the same item cannot appear in terminal sets of two different paths, and no paths can intersect. This contradicts the assumption.
\end{proof}

\textbf{Completeness}: To prove the completeness of $NAGG_{\mathcal{M}B}$, we must show that every edge $i_k.X \rightarrow i_j.Y$ in every ground graph $GG_\mathcal{M}\sigma$, where $\sigma \in \Sigma_{\mathcal{S}}$, has a set of corresponding edges in $NAGG_{\mathcal{M}B}$. Let $\sigma \in \Sigma_{\mathcal(S)}$ be an arbitrary skeleton and $i_k.X \rightarrow i_j.Y \in GG_\mathcal{M}\sigma$ be an arbitrary edge as a result of dependency $[I_j,...,I_k].X \rightarrow [I_j].Y \in \mathbf{D}$. The edge yields two sets of relational variables $\mathbf{P_k.X}=\{P_k.X|\exists_{b \in \sigma(B)}, i_k \in P_k|_b\}$ and $\mathbf{P_j.Y}=\{P_j.Y|\exists_{b \in \sigma(B)}, i_j \in P_j|_b\}$.
The construction of $NAGG_{\mathcal{M}B}$ adds edges corresponding to $i_k.X \rightarrow i_j.Y \in GG_\mathcal{M}\sigma$ as follows:
\begin{enumerate}
    \item If $P_k \in extend(P_j, [I_j,...,I_k])$, then $P_k.X \rightarrow P_j.Y$ is added according to the definition of $NAGG_{\mathcal{M}B}$.
    \item If $P_k \not \in extend(P_j, [I_j,...,I_k])$, but $\exists_{P_k'},P_k'.X \in \mathbf{P_k.X} \wedge P_k' \in extend(P_j, [I_j,...,I_k])$, then $P_k'.X \rightarrow P_j.Y$ is added according to the definition of $NAGG_{\mathcal{M}B}$.
\end{enumerate}
\begin{claim}\label{claim:2}
Above two are only possible conditions for $NAGG_{\mathcal{M}B}$ i.e.,
$\forall_{P_j|P_j.Y \in \mathbf{P_j.Y}}\exists_{P_k|P_k.X \in \mathbf{P_k.X}}, P_k \in extend(P_j, [I_j,...,I_k])$.
\end{claim}
\begin{proof}
From Proposition \ref{prop:paths}, we know we can have paths of maximum length four for $B=E$ and five for $B=R$. Moreover, due to the fully connected skeleton with $N=|\sigma(E)|$ entity instances and $|\sigma(R)|=\frac{N\times N-1}{2}$ relationship instances according to Definition \ref{def:skeleton}, the sets $\mathbf{P_j.Y}$ and $\mathbf{P_k.X}$ consist all valid paths of the form $[B,...,I_j]$ and $[B,...,I_k]$, respectively. Thus, $extend([B,...,I_j], [I_j,...,I_k])$ is always a subset of valid paths of the form $[B,...,I_k]$, and two conditions above are always true.
\end{proof}
\end{proof}

The proof for Corollary \ref{th:acyclicity} is as follows:
\begin{proof}
The proof follows directly from Theorem \ref{th:abstraction} that shows $NAGG_{\mathcal{M}\sigma}$ with extension of relationship existence uncertainty correctly abstracts all ground graphs for the acyclic model $\mathcal{M}$. Other extensions of latent and selection attributes are special roles given to the attributes. Thus, \citet{maier-thesis14}'s Theorem 4.5.3 can be directly applied to $NAGG_{\mathcal{M}\sigma}$ to prove $NAGG_{\mathcal{M}\sigma}$ is acylic.
\end{proof}

The proof for Theorem \ref{th:rel-d-sep} follows the proof by \citet{maier-thesis14} and is provided below:
\begin{proof}
To establish the soundness of relational d-separation, we must prove that d-separation on a $NAGG_{\mathcal{M}B}$ implies d-separation on all ground graphs it represents. To show the completeness of relational d-separation, we must prove that d-separation facts that hold across all ground graphs are also entailed by the d-separation on the $NAGG_{\mathcal{M}B}$. The introduction of latent and selection attributes only constraint the conditioning set $\mathbf{Z}$ such that it must always include selection attributes $\mathbf{S}$ and never include latent attributes $\mathbf{L}$ for reasoning about d-separation on both ground graphs and $NAGG_{\mathcal{M}B}$. The extension for relationship existence uncertainty constraints the relational skeleton and adds implicit dependencies in both ground graphs and $NAGG_{\mathcal{M}B}$. Thus,~\citet{maier-thesis14}'s proof of soundness and completeness of relational d-separation applies for $NAGG$. \citet{lee-uai15,lee-uai16} have challenged the completeness of AGG and \citet{maier-thesis14}'s proof of relational d-separation which is based on completeness of abstraction. But by Theorem \ref{th:abstraction}, we have shown that $NAGG$ is sound and complete.

\textbf{Soundness}: Assume $\mathbf{X}$ and $\mathbf{Y}$ are d-separated by $\mathbf{Z}$ on $NAGG_{\mathcal{M}B}$. Assume for contradiction that there exists an item instance $b \in \sigma(B)$, for an arbitrary skeleton $\sigma$, such that $\mathbf{X}|_b$ and $\mathbf{Y}|_b$ are \textit{not} d-separated by $\mathbf{Z}|_b$ in the ground graph $GG_{\mathcal{M}\sigma}$. This implies there must exist a d-connecting path $p$ from some $x \in \mathbf{X}|_b$ to some $y \in \mathbf{Y}|_b$ given all $z \in \mathbf{Z}|_b$. Theorem \ref{th:abstraction} shows that $NAGG_{\mathcal{M}B}$ is complete suggesting all the edges in the $GG_{\mathcal{M}\sigma}$ is captured by $NAGG_{\mathcal{M}B}$. So, the path $p$ must be represented by some node in $\{N_x|x \in N_x|_b\}$ connecting to some node in $\{N_y|y \in N_y|_b\}$, where $N_x$ and $N_y$ are nodes in $NAGG_{\mathcal{M}B}$. If the path $p$ is d-connected in $GG_{\mathcal{M}\sigma}$, then it is also d-connected in $NAGG_{\mathcal{M}B}$, implying that $\mathbf{X}$ and $\mathbf{Y}$ are \textit{not} d-separated by $\mathbf{Z}$. Thus, $\mathbf{X}|_b$ and $\mathbf{Y}|_b$ must be d-separated by $\mathbf{Z}|_b$.

\textbf{Completeness}: Assume $\mathbf{X}|_b$ and $\mathbf{Y}|_b$ are d-separated by $\mathbf{Z}|_b$ in the ground graph $GG_{\mathcal{M}\sigma}$ for all skeleton $\sigma$ and for all $b \in \sigma(B)$. Assume for contradiction that $\mathbf{X}$ and $\mathbf{Y}$ are \textit{not} d-separated by $\mathbf{Z}$ on $NAGG_{\mathcal{M}B}$. This implies there must exist a d-connecting path $p$ for some relational variable $X \in \mathbf{X}$ to some $Y \in \mathbf{Y}$ given all $Z \in \mathbf{Z}$. Theorem \ref{th:abstraction} shows that $NAGG_{\mathcal{M}B}$ is sound suggesting every edge in $NAGG_{\mathcal{M}B}$ must correspond to some pair of variables in some ground graph. So, if the path $p$ is responsible for d-connection in $NAGG_{\mathcal{M}B}$, then there must exist some skeleton $\sigma$ such that $p$ is d-connecting in $GG_{\mathcal{M}\sigma}$ for some $b \in \sigma(B)$, implying that d-separation does not hold for that ground graph. Thus, $\mathbf{X}$ and $\mathbf{Y}$ must be d-separated by $\mathbf{Z}$ on $NAGG_{\mathcal{M}B}$.
\end{proof}

\subsection{Proofs for identification of individual direct effects}\label{sec:ap-iden}

The proof for Lemma \ref{prop:iden} first formalizes the Assumption \ref{assum:pre} in terms of the Network Structural Causal Model (NSCM), then uses Network Abstract Ground Graph (NAGG) for reasoning about possible adjustment sets that satisfy the backdoor criterion~\cite{pearl-book09} in the presence of selection and latent variables.

\begin{proof}
    
    The formal implications of Assumption \ref{assum:pre} ``The network $G$ and its attributes are measured before treatment assignments and treatments are immutable from assignment to outcome measurement" on the NSCM $\mathcal{M}(\mathcal{S},\mathbf{D},\mathbf{f})$ are as follows:
    \begin{enumerate}
        \item The entity and relationship attributes , i.e., ${\mathbf{Z_n}\subseteq \mathcal{A}(E)}$ and ${\{Exists, \mathbf{Z_n}\}\subseteq \mathcal{A}(R)}$, are never the descendants of treatment attribute $X \in \mathcal{A}(E)$ and outcome attribute $Y \in \mathcal{A}(E)$. These attributes are commonly referred to as background covariates.
        \item There is no selection bias on treatment and outcome but there could be selection bias on the background covariates, i.e., $\mathcal{A(E)} \cup \mathcal{A(R)} \setminus \{X, Y\} \in \mathbf{S}$, where $\mathbf{S}$ indicates a set of attributes marked as selected. Note that the absence of selection bias is the stronger assumption and the experimental or observational studies for causal inference studies typically select certain populations of interest based on the background covariates. For the selected population, there is no further selection based on treatment and outcome is a reasonable assumption for experimental or prospective observational data.
        \item Although all attributes could be auto-regressive (e.g., $[E].Y^L \rightarrow [E].Y$) or contagious (e.g., $[E,R,E].Y^L \rightarrow [E].Y$), the treatment is assumed to be immutable during the study period. This implies the time-lagged treatment attributes do not directly affect the outcome, i.e.,\{$[E].X^L \rightarrow [E].Y, [E,R,E].X^L \rightarrow [E].Y\} \not \in \mathbf{D}$. All time-lagged attributes $\{X^L, Y^L, \mathbf{Z_n}^L, \mathbf{Z_e}^L, Exists^L\} \subseteq \mathbf{L}$, where $\mathbf{L}$ is a set of attributes marked as latent.
    \end{enumerate}

    For node attributes {\small $\mathbf{Z_n}$} and edge attributes {\small $\mathbf{Z_e}$}, we denote relational variables {\small $[E].\mathbf{Z_n}, [E,R].\mathbf{Z_e}$, $[E,R,E].\mathbf{Z_n}$}, and {\small $[E,R,E,R].\mathbf{Z_e}$} with the notations {\small $Z_i\in \mathbb{R}^{d}, Z_r\in \mathbb{R}^{N'\times d'}, Z_{-i}\in \mathbb{R}^{N'\times d},$} and {\small $Z_{-r}\in \mathbb{R}^{N'\times N'-1 \times d'}$}, respectively, where {\small $N'=|\mathcal{V}|-1$ and $<d,d'>$} are constants. Let {\small $E_r \in \{0,1\}^{N'\times 1}$} and {\small $E_{-r} \in \{0,1\}^{N'\times N'-1 \times 1}$} be the relationship existence indicator variables {\small $[E,R].Exists$} and {\small $[E,R,E,R].Exists$}, respectively.
    Let $\mathbf{L_v}$ and $\mathbf{S_v}$ denote sets of latent and selection variables, respectively.

    Since our causal estimand of individual direct effects (IDE), has peer treatments as a conditional in Equation \ref{eq:dir_eff}, the identification should consider $X_{-i}$, i.e., $[E,R,E].X$ as observed and we need to find an adjustment set $\{X_{-i}, \mathcal{Z}_i\}$. To prove Lemma \ref{lemma:adjustment}, we have to show that if there exists any valid adjustment set $\mathbf{W}$, where $X_{-i} \in \mathbf{W} \wedge \mathbf{W} \cap \mathbf{L_v} = \emptyset \wedge \mathbf{W} \cap \mathbf{S_v}=\mathbf{S_v}$, then the set $X_{-i}, \mathcal{Z}_{i}, \mathbf{S_v}\}$, where $\mathcal{Z}_i=\{Z_i, Z_r, Z_{-i}, Z_{-r}, E_r, E_{-r}\}$ is also a valid adjustment set. The restrictions in $\mathbf{W}$ follow from the implications of latent and selection variables in the adjustment set.

    Although the unconfoundedness condition, i.e., $\{Y_i(X_i=1),Y_i(X_i=0)\} \perp\!\!\!\!\perp  X_i | \mathbf{W}$, is untestable, causal reasoning of d-separation between $X_i$ and $Y_i$ can be used in the ``mutilated" SCM ($\mathcal{G}_{\bar{X_i}}$) where outgoing edges from treatment $X_i$ are removed~\cite{pearl-book09,bareinboim-pnas16}.

    There can be arbitrary dependencies between all possible background variables and these variables may or may not affect treatment and outcome.
    
    Proof by contradiction. Let $\mathbf{W}$ satisfy the unconfoundedness condition, i.e., $\mathbf{W}$ d-separates $X_i$ and $Y_i$ in $\mathcal{G}_{\bar{X_i}}$, but the set $\{X_{-i}, Z_i, Z_r, Z_{-i}, Z_{-r}, E_r, E_{-r}\}$ does not for contradiction.

    \textbf{Case I}: \textit{There is latent confounding between $X_i$ and $Y_i$, i.e., $X_i \leftrightarrow Y_i$.}\\
    This contradicts the condition that $\mathbf{W}$ d-separates $X_i$ and $Y_i$ in $\mathcal{G}_{\bar{X_i}}$ because there exists an open path due to unobserved confounding.

    \textbf{Case II}: \textit{There exist paths $X_i \leftrightarrow C$ and $D \leftrightarrow Y_i$ for any colliders $C$ and $D$ such that $\{C,D\} \not \in \{X_i,Y_i,Y_{-i}\}$, where $C=D$ or there exist open path from $C$ to $D$ (e.g., collider chaining), with a latent common cause between $C$ and the treatment $X_i$ and a separate latent common cause between $D$ and the outcome $Y_i$.}\\
    In this case, if collider $C$ is peer treatments, i.e., $X_{-i}$, then it must be observed for estimating Equation \ref{eq:dir_eff}. The conditioning opens a path from $X_i$ to $Y_i$ via the collider. If the collider is any background covariate then, the selection of the background covariates opens the path $X_i$ to $Y_i$ via the collider. Any open paths from $C$ and $D$ for $C\ne D$ follows the same reasoning. These cases contradict the condition that $\mathbf{W}$ d-separates $X_i$ and $Y_i$ in $\mathcal{G}_{\bar{X_i}}$.

    \textbf{Case III}: \textit{There exist a path $X_i \leftarrow L_v \rightarrow Y_{-i}$, where $L_v \in \{L_i, L_r, L_{-i}, L_{-r}\}$.}\\
    In this case, although we have $X_i \leftrightarrow Y_{-i} \leftrightarrow Y_i$, $Y_{-i}$ is not observed and does not belong to the selection set and cannot directly open paths like case II. However,  $X_i \leftrightarrow Y_{-i}$ suggests a dependency $L_v \rightarrow Y_{-i}$, where $L \in \mathbf{L}$ be a latent node or edge attribute and $L_v$ the corresponding variable. NSCM specifies dependencies in the canonical form. Therefore, for this extended dependency $L_v \rightarrow Y_{-i}$ (purple arrows in Figure \ref{fig:proof-adjust}), there must be some explicit dependency that needs to be considered for the identifiability of causal effects. Figure \ref{fig:proof-adjust} illustrates equivalence for $L_v \in \{L_i, L_r, L_{-i}, L_{-r}\}$ to cases II, I, II, and I respectively. The purple arrow shows $L_v \rightarrow Y_{-i}$ dependency while the black arrow shows $L_v \rightarrow X_i$ dependency. These dependencies are extended and shown as dotted edges. For this case the condition that $\mathbf{W}$ d-separates $X_i$ and $Y_i$ in $\mathcal{G}_{\bar{X_i}}$ does not hold.

    For all other cases, from Proposition \ref{prop:paths}, we have a maximum of four paths for entity perspective. The adjustment set $\{X_{-i}, Z_i, Z_r, Z_{-i}, Z_{-r}, E_r, E_{-r}\}$ covers all observed variables except $Y_{-i}$. Adjusting on $Y_{-i}$ can open the colliders. Thus, $\{X_{-i}, Z_i, Z_r, Z_{-i}, Z_{-r}, E_r, E_{-r}\}$ is the maximal adjacency set that does not open any more colliders and block all the mediated backdoor paths if first three cases do not exist. This contradicts the set $\{X_{-i}, Z_i, Z_r, Z_{-i}, Z_{-r}, E_r, E_{-r}\}$ does not satisfy unconfoundedness when a set $\mathbf{W}$ does.

\begin{figure}
    \centering
    \includegraphics[width=0.4\textwidth]{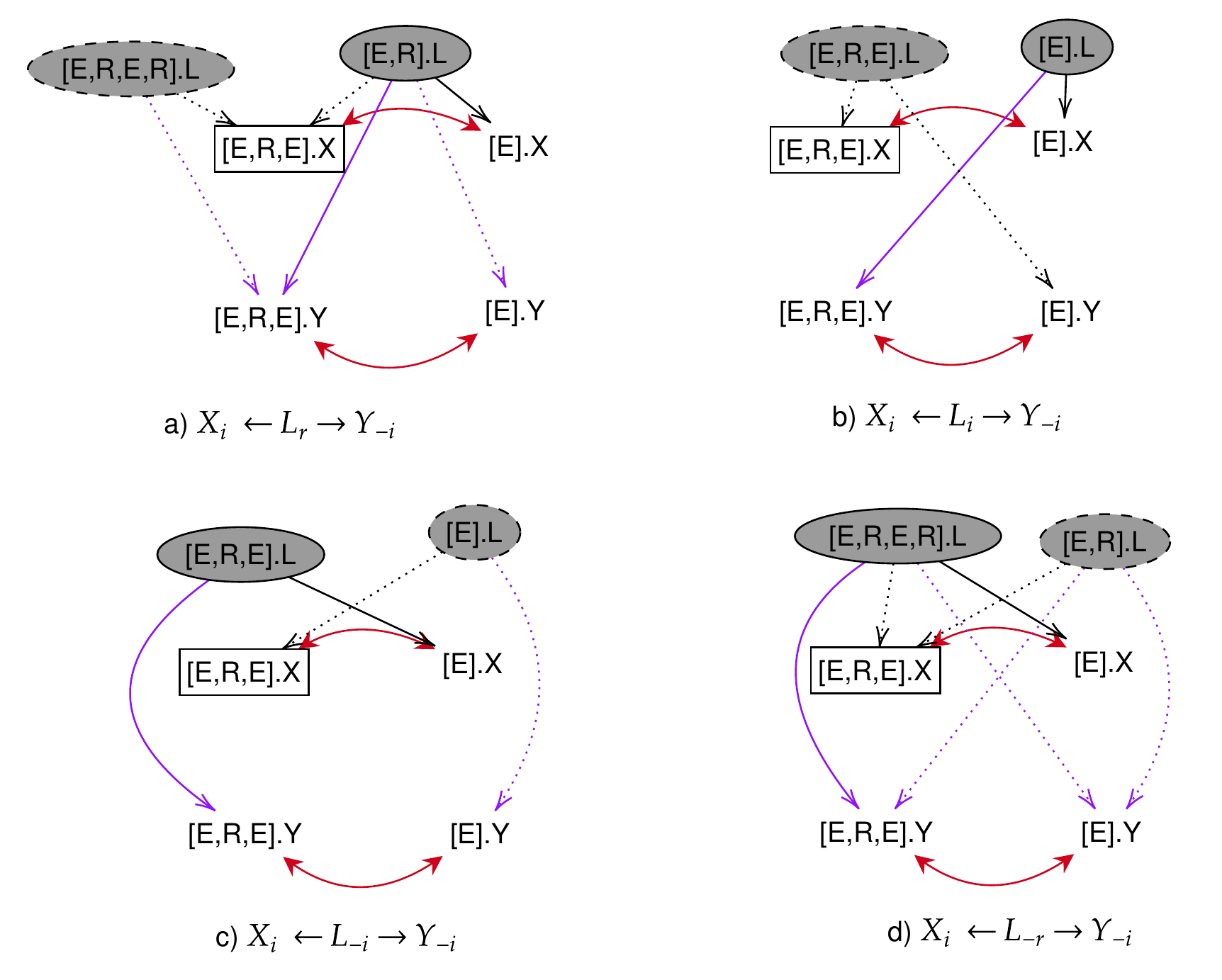}
    \caption{Illustration of contradiction of condition $\mathbf{W}$ d-separates $X_i$ and $Y_i$ in $\mathcal{G}_{\bar{X_i}}$ that for Case III: \textit{There exist a path $X_i \leftarrow L_v \rightarrow Y_{-i}$, where $L_v \in \{L_i, L_r, L_{-i}, L_{-r}\}$.} Depending on the path of the latent variable in Case III, other dependencies derived (dotted) from NSCM/NAGG show equivalence with Case I and Case II.}
    \label{fig:proof-adjust}
\end{figure}
\end{proof}
\textbf{Discussion}. The implication of Lemma \ref{lemma:adjustment} is that under Assumption \ref{assum:pre}, the underlying data generation of treatment and outcome can be described by functions {\small $f_X$} and {\small $f_Y$} in our general NSCM as: \\
\begin{align}
    &X_i = f_X(Z_i, Z_r, Z_{-i}, Z_{-r}, E_r, E_{-r}, X_i^L, X_{-i}^L,\epsilon_X) \text{ and} \\
    &Y_i = f_Y(X_i, X_{-i}, Z_i, Z_r, Z_{-i}, Z_{-r}, E_r, E_{-r}, Y_i^L, Y_{-i}^L, \epsilon_Y),
\end{align}
where {\small $\{X_i^L, X_{-i}^L,Y_i^L, Y_{-i}^L\}$} are latent time-lagged variables for treatment and outcome. The NSCM allows modeling latent homophily, contagious treatments or outcomes, and selection bias. Although all relational variables may not affect the treatment/outcome, the above NSCM does not change the identifiability of causal effects. Considering the maximal adjustment set is helpful to close the path between separate colliders $C$ and $D$ for case II as well as the backdoor paths~\cite{pearl-book09} mediated by the maximal set. Also, having a maximal adjustment set helps to capture underlying influence mechanisms and effect modifications. We have identified conditions under which observational data is not sufficient to guarantee the identifiability of causal effects. Further assumptions or experiments are required in such cases.

The proof of Corollary \ref{prop:iden} applies do-calculus for experimental data with incoming edges to unit treatment and peer treatments removed but allowing correlation between treatments. For observational data, we use Lemma \ref{lemma:adjustment}'s condition, i.e., the unconfoundedness assumption.
\begin{proof}

\textbf{Case 1: Experimental data}. In this case, the treatments are assigned to units randomly or according to some experimental design. Therefore, the treatments $<X_i, X{-i}>$ are exogenous and do not depend on observed contexts $\mathcal{Z}_{-i}$ or unobserved confounders. The treatment assignments, however, can be correlated depending on the experimental design. Figure \ref{fig:prf_exp} depicts the causal graph for the experimental data with exogenous but correlated treatments (green bi-directed edge). Although the relational variables $\{Z_i, Z_r, Z_{-i}, Z_{-r}, E_r, E_{-r}\}$ can have arbitrary causal dependence among themselves, we represent these variables with a single node in the causal graph for simplicity. The black edges in the causal diagram show dependencies between treatment, outcome, and context variables. The blue edges show the dependencies due to contagion where a unit's outcome, measured some time steps after treatment assignment, may be influenced by peers' past outcomes. The dashed blue bi-directed edge indicates the time-lagged outcomes may be correlated due to contagion in earlier time steps. Figure \ref{fig:prf_exp} depicts treatment $X_i$ in the blue circle, outcome $Y_i$ in the red circle, conditional variables $\mathcal{Z}_i$ and $X_{-i}$ in rectangles, and latent variables in the shaded circle.

The counterfactual $E[Y_i(X_i=\pi_i)|X_{-i},\mathcal{Z}_i]$ can be written in terms of do-expression as $E[Y_i|do(X_i=\pi_i),X_{-i},\mathcal{Z}_i]$ because both $\{X_{-i},\mathcal{Z}_i\}$ are non-descendants of $X_i$~\cite{pearl-book09}.
Then, $E[Y_i|do(X_i=\pi_i),X_{-i},\mathcal{Z}_i]$ can be estimated with $E[Y_i|X_i=\pi_i,X_{-i},\mathcal{Z}_i]$ in the experimental data because $\{X_{-i},\mathcal{Z}_i\}$ satisfies backdoor criterion~\cite{pearl-book09} from $X_i$ to $Y_{i}$. In the causal graph with edges outgoing from treatment $X_i$ removed, the conditional set $\{X_{-i},\mathcal{Z}_i\}$ d-separates $X_i$ and $Y_i$ and satisfies Pearl's second rule of do-calculus~\cite{bareinboim-pnas16}, enabling the substitution.

\begin{figure}
    \centering
    \includegraphics[width=0.4\linewidth]{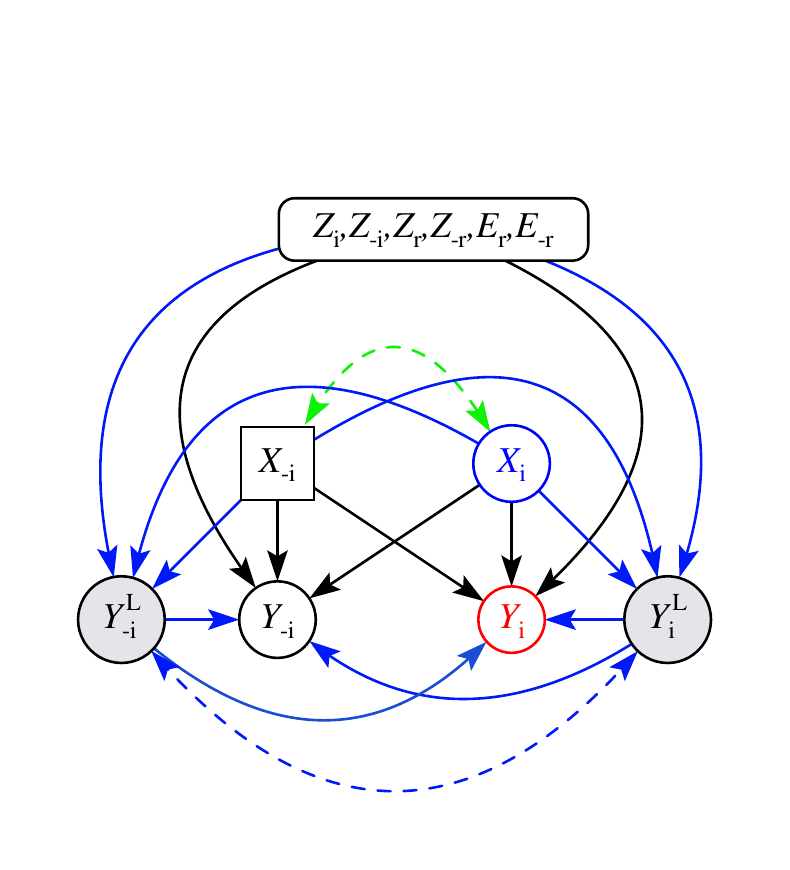}
    \caption{Causal diagram for experimental data ($G^{exp}$).}
    \label{fig:prf_exp}
    \vspace{-1em}
\end{figure}

\textbf{Case 2: Observational data}. For the identification of causal effects from observational data, we have an additional assumption of unconfoundedness, i.e., $\{Y_i(X_i=1),Y_i(X_i=0)\} \perp\!\!\!\!\perp  X_i | X_{-i},\mathcal{Z}_i$ (Lemma \ref{lemma:adjustment}). The counterfactual $E[Y_i(X_i=\pi_i)|X_{-i},\mathcal{Z}_i]$ can be written in terms of do-expression as $E[Y_i|do(X_i=\pi_i),X_{-i},\mathcal{Z}_i]$ because both $\{X_{-i},\mathcal{Z}_i\}$ are non-descendants of $X_i$~\cite{pearl-book09}. The term $E[Y_i|do(X_i=\pi_i),X_{-i},\mathcal{Z}_i]$ can be replaced by conditional $E[Y_i|X_i=\pi_i,X_{-i},\mathcal{Z}_i]$ using Pearl's second rule of do-calculus~\cite{bareinboim-pnas16} if  $\{X_{-i},\mathcal{Z}_i\}$ d-separates $X_i$ and $Y_i$ in the mutilated graph with outgoing edges from $X_i$ removed. Assume, for contradiction, this d-separation does not exist. However, this contradicts the unconfoundedness assumption where $X_i$ is independent of counterfactual outcomes $<Y_i(1), Y_i(0)>$ given $\{X_{-i},\mathcal{Z}_i\}$. Therefore, the counterfactual $E[Y_i(X_i=\pi_i)|X_{-i},\mathcal{Z}_i]$ can be estimated with $E[Y_i|X_i=\pi_i,X_{-i},\mathcal{Z}_i]$ under unconfoundedness assumption.
\end{proof}

\subsection{Estimators and hyperparameters}\label{sec:hyperparam}
We start with the description of hyperparameters of the proposed IDE-Net estimator. Then, we discuss hyperparameters for baselines in each experimental setting. Finally, we describe the computational resources used for the experiments.

We separate IDE-Net into two components: individual direct effects (IDE) encoder to learn feature and exposure mapping and TARNet for counterfactual prediction. Since our feature and exposure embeddings try to capture all potential heterogeneity contexts, most of these contexts could be irrelevant for underlying data generation. So, we include a Batch Normalization layer with $tanh$ activation function after the first $MLP$ layer in the TARNet module. The Batch Normalization layer adds regularization to mitigate the effect of irrelevant features. 

The fundamental challenge with causal effect estimation tasks is that we do not have ground truth for any causal effects to tune hyperparameters and the models have to rely on generalizing for the observed factual outcomes. The regularization to avoid over/underfitting relies on smoothing the variance in estimated causal effects. This is particularly important for our causal inference under heterogeneous peer influence (HPI) task to maintain invariance to non-relevant contexts while retaining expressiveness. Also, we train our models starting with a higher learning rate and gradually decreasing the learning rates.

\textbf{Hyperparameters for IDE-Net estimator}.
\begin{itemize}
    \item \texttt{maxiter}: Maximum number of epochs to train (default 300)
    \item \texttt{val}: Fraction of nodes to be used in validation set for model selection or early stopping (default $0.2$)
    \item \texttt{lr}: Learning rate for the encoder, i.e., feature and exposure embeddings, (default $0.02$)
    \item \texttt{lrest}: Learning rate for the estimator, i.e., counterfactual predictors, (default $0.2$)
    \item \texttt{lrstep}: Change learning rate after l epochs (default $50$)
    \item \texttt{lrgamma}: Decay learning rate by multiplying it with this value (default $0.5$)
    \item \texttt{clip}: Clip gradient values of estimator (default $3$)
    \item \texttt{max\_patience}: Early stopping if validation loss (without regularization) does not improve for given epochs (default $300$, i.e., no early stopping)
    \item \texttt{weight\_decay}: L2 regularization parameter for the Adam optimizer (default $1e-5$)
    \item \texttt{fdim}: Dimension of hidden units for node and peer features (default $32$ for synthetic and $64$ for semi-synthetic data)
    \item \texttt{edim}: Dimension of hidden units for edge features (default $4$)
    \item \texttt{inlayers}: MLP layers for the encoder (default $2$). We chose $2$ layers to capture non-linear feature mapping before summarization with GCN.
    \item \texttt{dropout}: Dropout probability (default $0$). We rely on L2 regularization, batch normalization, and smoothing regularization and do not use dropout.
    \item \texttt{normY}: Whether outcome should be normalized for training (default False)
    \item \texttt{alpha}: Parameter for representation balancing regularization~\cite{shalit-icml17,guo-wsdm20} (default $0$ for TARNet and $0.5$ for CFR). Although IDE-Net could be converted to IDE-CFR by passing any alpha except $0$, we do not report experiments for this estimator because it has additional computation cost and hyperparameter selection. We substitute the regularization obtained by representation balancing with cheaper smoothing regularization explained next.
    \item \texttt{reg}: Whether smoothing regularization should be enabled (default True). If set to true, we fix decay parameter $\gamma=3$ and select the model with the best validation loss (without regularization) between two settings of scaling parameter $\lambda_s=0.1$ and $\lambda_s=1.0$. The regularization is only enabled after $150$, i.e., $50\%$, epochs to estimate the variance of predicted causal effects and decide the strength of smoothing. Without smoothing, the performance of IDE-Net is slightly worse, when the underlying causal effects are uniform.
\end{itemize}

\textbf{Hyperparmeters for baselines}. We compare the performance of our method with the NetEst~\cite{jiang-cikm22} estimator and baselines (Network Deconfounder (ND)~\cite{guo-wsdm20}, ND-INT, GCN-TARNet, GCN-TARNet-INT, MLP-CFR, and MLP-CFR-INT) implemented in the NetEst paper\footnote{https://github.com/songjiang0909/Causal-Inference-on-Networked-Data}. NetEst~\cite{jiang-cikm22} focuses on estimating "insulated" individual effects, peer effects, and overall effects for within-sample data (i.e., data available for training the model) as well as out-of-sample data (i.e., data unavailable for training the model). NetEst uses regularization by making features less predictive of treatment assignments and homogenous peer exposure. We adapt NetEst to estimate individual effects instead of "insulated" individual effects but find the performance very poor. This may be because NetEst concatenates features, peer exposure, and treatment before counterfactual prediction, and it is not agnostic to treatment like TARNet and CFR architectures~\cite{shalit-icml17}. The value of treatment may be lost in our settings where individual effects have a small coefficient and peer exposure has a large coefficient. Although we focus on individual direct effect estimation for within-sample data, we use NetEst with  "insulated" individual effects as a baseline for two experiment settings because "insulated" individual effects are equivalent to individual effects without interaction with peer exposure.

We use all default hyperparameters of NetEst except we set the epochs to train the treatment predictor and peer exposure predictor to $1$ instead of $50$ to save computation time. Interestingly, we find the performance of NetEst with parameters $dstep=1$ and $d\_zstep=1$ is better than the default settings and report these performances in Figures \ref{fig:cane_pehe}-\ref{fig:cane_ate}. For the BART baseline, we use python-based implementation\footnote{https://github.com/JakeColtman/bartpy} and use the default parameters. BART is popular for good performances without hyperparameters tuning~\cite{hill-jcgs11}. For other baselines, we use default parameters but adapt them to estimate individual direct effects instead of "insulated" individual effects. For the ablation study in Figure \ref{fig:cane_em_pehe}, the baselines use our TARNet/CFR module and similar hyperparameters to that of the IDE-Net estimator. For the ablation study, we follow \citet{yuan-www21}'s implementation\footnote{https://github.com/facebookresearch/CausalMotifs} to generate data for structural diversity as a peer influence mechanism and to extract causal motifs.

The DWR~\cite{zhao-arxiv22} baseline is trained for $300$ epochs with $5$ epoch each for calibration by setting a learning rate of $0.2$. 
1-GNN-HSIC~\cite{ma-aistats21} is trained for $300$ epochs with representation balancing parameter $alpha=0.5$ by setting a learning rate of $0.2$. Other hyperparameters like gradient clipping, feature dimension, and weight\_decay are similar to IDE-Net for both variants of DWR as well as 1-GNN-HSIC. \newrevision{TNet is trained for $10$ epochs with $50$ iterations on each epoch for targeted learning with other default hyperparameters.}

\textbf{Computational resources}. All the experiments are performed in a machine with the following resources.
\begin{itemize}
    \item CPU: AMD EPYC 7662 64-Core Processor (128 CPUs)
    \item Memory: 256 GB RAM
    \item Operating system: Ubuntu 20.04.4 LTS
    \item GPU: NVIDIA RTX A5000 (24 GB)
    \item CUDA Version: 11.4
\end{itemize}

\subsection{Results in Tabular Form}
\subsubsection{Estimation error for {IDEs with effect modification} and HPI based on similarity, mutual connections, and peer degree (Fig \ref{fig:cane_pehe}).}
\begin{table}[!h]
\centering
\label{synCaneBA}
\begin{adjustbox}{max width=\textwidth}
\begin{tabular}{lllllllllllll}
{Metric} & \multicolumn{12}{l}{$\epsilon_{PEHE}$} \\
{Network} & \multicolumn{12}{l}{Barabasi Albert} \\
{Sparsity (p)} & \multicolumn{4}{r}{1} & \multicolumn{4}{r}{5} & \multicolumn{4}{r}{10} \\
{exposure} & {0} & {10} & {20} & {30} & {0} & {10} & {20} & {30} & {0} & {10} & {20} & {30} \\
{method} & {} & {} & {} & {} & {} & {} & {} & {} & {} & {} & {} & {} \\
MLP-CFR & {\cellcolor[HTML]{E43027}} 1.32±0.10 & {\cellcolor[HTML]{FC7F5F}} 0.86±0.36 & {\cellcolor[HTML]{E63328}} 1.30±0.38 & {\cellcolor[HTML]{B21218}} 2.05±0.89 & {\cellcolor[HTML]{F34A36}} 1.05±0.09 & {\cellcolor[HTML]{FB7B5B}} 0.99±0.14 & {\cellcolor[HTML]{FCBBA1}} 0.92±0.24 & {\cellcolor[HTML]{FFECE3}} 0.88±0.30 & {\cellcolor[HTML]{F34935}} 1.01±0.10 & {\cellcolor[HTML]{FA6849}} 1.00±0.10 & {\cellcolor[HTML]{FC8D6D}} 0.98±0.17 & {\cellcolor[HTML]{FCAD90}} 0.96±0.24 \\
ND & {\cellcolor[HTML]{67000D}} 1.77±0.24 & {\cellcolor[HTML]{75030F}} 1.23±0.44 & {\cellcolor[HTML]{880811}} 1.55±0.37 & {\cellcolor[HTML]{67000D}} 2.32±0.76 & {\cellcolor[HTML]{67000D}} 1.53±0.23 & {\cellcolor[HTML]{67000D}} 1.46±0.24 & {\cellcolor[HTML]{C1161B}} 1.36±0.31 & {\cellcolor[HTML]{FCB296}} 1.28±0.37 & {\cellcolor[HTML]{67000D}} 1.48±0.25 & {\cellcolor[HTML]{67000D}} 1.47±0.23 & {\cellcolor[HTML]{840711}} 1.44±0.26 & {\cellcolor[HTML]{B31218}} 1.39±0.30 \\
MLP-CFR-INT & {\cellcolor[HTML]{CC191E}} 1.43±0.24 & {\cellcolor[HTML]{E83429}} 1.01±0.34 & {\cellcolor[HTML]{EA362A}} 1.29±0.32 & {\cellcolor[HTML]{D72322}} 1.86±0.55 & {\cellcolor[HTML]{CE1A1E}} 1.23±0.18 & {\cellcolor[HTML]{E93529}} 1.16±0.19 & {\cellcolor[HTML]{FC7F5F}} 1.07±0.26 & {\cellcolor[HTML]{FEE1D3}} 1.00±0.31 & {\cellcolor[HTML]{C9181D}} 1.20±0.18 & {\cellcolor[HTML]{DB2824}} 1.18±0.16 & {\cellcolor[HTML]{F34A36}} 1.14±0.19 & {\cellcolor[HTML]{FB7858}} 1.09±0.25 \\
\end{tabular}
\end{adjustbox}
\end{table}
\begin{table}[!h]
\centering
\label{synCaneWS}
\begin{adjustbox}{max width=\textwidth}
\begin{tabular}{lllllllllllll}
{Metric} & \multicolumn{12}{l}{$\epsilon_{PEHE}$} \\
{Network} & \multicolumn{12}{l}{Watts Strogatz} \\
{Sparsity (k)} & \multicolumn{4}{r}{6} & \multicolumn{4}{r}{30} & \multicolumn{4}{r}{300} \\
{exposure} & {0} & {10} & {20} & {30} & {0} & {10} & {20} & {30} & {0} & {10} & {20} & {30} \\
{method} & {} & {} & {} & {} & {} & {} & {} & {} & {} & {} & {} & {} \\
MLP-CFR & {\cellcolor[HTML]{F34A36}} 1.06±0.09 & {\cellcolor[HTML]{FCA082}} 0.87±0.17 & {\cellcolor[HTML]{FFEBE2}} 0.73±0.26 & {\cellcolor[HTML]{FFF5F0}} 0.76±0.27 & {\cellcolor[HTML]{F75C41}} 1.00±0.09 & {\cellcolor[HTML]{FB7555}} 0.96±0.09 & {\cellcolor[HTML]{FC9777}} 0.94±0.12 & {\cellcolor[HTML]{FDC5AE}} 0.90±0.18 & {\cellcolor[HTML]{F96044}} 1.00±0.10 & {\cellcolor[HTML]{FB7353}} 1.00±0.09 & {\cellcolor[HTML]{FB7A5A}} 1.00±0.09 & {\cellcolor[HTML]{FC9777}} 1.00±0.09 \\
ND & {\cellcolor[HTML]{67000D}} 1.52±0.24 & {\cellcolor[HTML]{A60F15}} 1.32±0.29 & {\cellcolor[HTML]{F6583E}} 1.14±0.36 & {\cellcolor[HTML]{FDD5C4}} 1.12±0.36 & {\cellcolor[HTML]{6B010E}} 1.47±0.25 & {\cellcolor[HTML]{900A12}} 1.45±0.23 & {\cellcolor[HTML]{A91016}} 1.44±0.23 & {\cellcolor[HTML]{C1161B}} 1.40±0.26 & {\cellcolor[HTML]{A91016}} 1.31±0.33 & {\cellcolor[HTML]{AF1117}} 1.31±0.33 & {\cellcolor[HTML]{B01217}} 1.31±0.33 & {\cellcolor[HTML]{B81419}} 1.31±0.33 \\
MLP-CFR-INT & {\cellcolor[HTML]{CE1A1E}} 1.23±0.19 & {\cellcolor[HTML]{F75C41}} 1.04±0.20 & {\cellcolor[HTML]{FCBCA2}} 0.89±0.28 & {\cellcolor[HTML]{FEE8DE}} 0.93±0.35 & {\cellcolor[HTML]{D32020}} 1.19±0.17 & {\cellcolor[HTML]{E83429}} 1.16±0.16 & {\cellcolor[HTML]{F6583E}} 1.13±0.15 & {\cellcolor[HTML]{FC8060}} 1.09±0.18 & {\cellcolor[HTML]{D32020}} 1.18±0.17 & {\cellcolor[HTML]{DD2A25}} 1.18±0.17 & {\cellcolor[HTML]{E12D26}} 1.18±0.17 & {\cellcolor[HTML]{EB372A}} 1.19±0.17 \\
\end{tabular}
\end{adjustbox}
\end{table}
\begin{table}[!h]
\centering
\label{TblSemiCaneBC}
\begin{adjustbox}{max width=\textwidth}
\begin{tabular}{lllllllllllllllll}
{Metric} & \multicolumn{16}{l}{$\epsilon_{PEHE}$} \\
{Network} & \multicolumn{16}{l}{BlogCatalog} \\
{Source} & \multicolumn{4}{r}{Combined} & \multicolumn{4}{r}{Mutual connections} & \multicolumn{4}{r}{Peer degree} & \multicolumn{4}{r}{Similarity} \\
{exposure} & {0} & {10} & {20} & {30} & {0} & {10} & {20} & {30} & {0} & {10} & {20} & {30} & {0} & {10} & {20} & {30} \\
{method} & {} & {} & {} & {} & {} & {} & {} & {} & {} & {} & {} & {} & {} & {} & {} & {} \\
MLP-CFR & {\cellcolor[HTML]{A91016}} 0.77±0.22 & {\cellcolor[HTML]{8A0812}} 1.02±0.39 & {\cellcolor[HTML]{67000D}} 1.54±0.70 & {\cellcolor[HTML]{67000D}} 2.12±1.00 & {\cellcolor[HTML]{A91016}} 0.77±0.22 & {\cellcolor[HTML]{A10E15}} 0.87±0.24 & {\cellcolor[HTML]{75030F}} 1.13±0.33 & {\cellcolor[HTML]{67000D}} 1.49±0.48 & {\cellcolor[HTML]{A91016}} 0.77±0.22 & {\cellcolor[HTML]{840711}} 1.03±0.40 & {\cellcolor[HTML]{67000D}} 1.56±0.72 & {\cellcolor[HTML]{67000D}} 2.15±1.03 & {\cellcolor[HTML]{A91016}} 0.77±0.22 & {\cellcolor[HTML]{800610}} 1.09±0.38 & {\cellcolor[HTML]{67000D}} 1.75±0.73 & {\cellcolor[HTML]{67000D}} 2.46±1.08 \\
ND & {\cellcolor[HTML]{67000D}} 0.85±0.24 & {\cellcolor[HTML]{67000D}} 1.07±0.39 & {\cellcolor[HTML]{67000D}} 1.54±0.69 & {\cellcolor[HTML]{67000D}} 2.12±0.99 & {\cellcolor[HTML]{67000D}} 0.85±0.24 & {\cellcolor[HTML]{77040F}} 0.92±0.26 & {\cellcolor[HTML]{67000D}} 1.15±0.34 & {\cellcolor[HTML]{67000D}} 1.49±0.48 & {\cellcolor[HTML]{67000D}} 0.85±0.24 & {\cellcolor[HTML]{67000D}} 1.07±0.40 & {\cellcolor[HTML]{67000D}} 1.56±0.72 & {\cellcolor[HTML]{69000D}} 2.14±1.02 & {\cellcolor[HTML]{67000D}} 0.85±0.24 & {\cellcolor[HTML]{67000D}} 1.13±0.38 & {\cellcolor[HTML]{67000D}} 1.75±0.74 & {\cellcolor[HTML]{6D010E}} 2.43±1.06 \\
\end{tabular}
\end{adjustbox}
\end{table}

\begin{table}[!h]
\centering
\label{TblSemiCaneFL}
\begin{adjustbox}{max width=\textwidth}
\begin{tabular}{lllllllllllllllll}
{Metric} & \multicolumn{16}{l}{$\epsilon_{PEHE}$} \\
{Network} & \multicolumn{16}{l}{Flickr} \\
{Source} & \multicolumn{4}{r}{Combined} & \multicolumn{4}{r}{Mutual connections} & \multicolumn{4}{r}{Peer degree} & \multicolumn{4}{r}{Similarity} \\
{exposure} & {0} & {10} & {20} & {30} & {0} & {10} & {20} & {30} & {0} & {10} & {20} & {30} & {0} & {10} & {20} & {30} \\
{method} & {} & {} & {} & {} & {} & {} & {} & {} & {} & {} & {} & {} & {} & {} & {} & {} \\
MLP-CFR & {\cellcolor[HTML]{FA6849}} 0.70±0.17 & {\cellcolor[HTML]{FB7050}} 0.92±0.25 & {\cellcolor[HTML]{FB7353}} 1.35±0.47 & {\cellcolor[HTML]{F34A36}} 1.85±0.69 & {\cellcolor[HTML]{FA6849}} 0.70±0.17 & {\cellcolor[HTML]{FC8060}} 0.71±0.17 & {\cellcolor[HTML]{FB7A5A}} 0.76±0.18 & {\cellcolor[HTML]{FB7353}} 0.84±0.23 & {\cellcolor[HTML]{FA6849}} 0.70±0.17 & {\cellcolor[HTML]{F14331}} 0.93±0.26 & {\cellcolor[HTML]{FC8F6F}} 1.38±0.48 & {\cellcolor[HTML]{F14432}} 1.88±0.71 & {\cellcolor[HTML]{FA6849}} 0.70±0.17 & {\cellcolor[HTML]{FCBDA4}} 0.94±0.29 & {\cellcolor[HTML]{FDCCB8}} 1.42±0.55 & {\cellcolor[HTML]{FCB99F}} 1.91±0.83 \\
ND & {\cellcolor[HTML]{F34C37}} 0.76±0.17 & {\cellcolor[HTML]{FA6648}} 0.95±0.25 & {\cellcolor[HTML]{FB7252}} 1.36±0.46 & {\cellcolor[HTML]{F34C37}} 1.84±0.69 & {\cellcolor[HTML]{F34C37}} 0.76±0.17 & {\cellcolor[HTML]{FA6849}} 0.78±0.18 & {\cellcolor[HTML]{F75C41}} 0.84±0.20 & {\cellcolor[HTML]{F85F43}} 0.90±0.24 & {\cellcolor[HTML]{F34C37}} 0.76±0.17 & {\cellcolor[HTML]{EB372A}} 0.96±0.26 & {\cellcolor[HTML]{FC8F6F}} 1.38±0.47 & {\cellcolor[HTML]{F14432}} 1.88±0.71 & {\cellcolor[HTML]{F34C37}} 0.76±0.17 & {\cellcolor[HTML]{FCBDA4}} 0.94±0.30 & {\cellcolor[HTML]{FDCDB9}} 1.40±0.57 & {\cellcolor[HTML]{FCB79C}} 1.94±0.83 \\
\end{tabular}
\end{adjustbox}
\end{table}

\clearpage
\subsubsection{$\epsilon_{PEHE}$ errors for \textbf{constant true individual direct effects} and heterogeneous peer influence (Fig \ref{fig:ane_pehe}).}
\begin{table}[!h]
\centering
\label{TblSynAneBA}
\begin{adjustbox}{max width=\textwidth}
\begin{tabular}{lllllllllllll}
{Metric} & \multicolumn{12}{l}{$\epsilon_{PEHE}$} \\
{Network} & \multicolumn{12}{l}{Barabasi Albert} \\
{Sparsity (p)} & \multicolumn{4}{r}{1} & \multicolumn{4}{r}{5} & \multicolumn{4}{r}{10} \\
{exposure} & {0} & {10} & {20} & {30} & {0} & {10} & {20} & {30} & {0} & {10} & {20} & {30} \\
{method} & {} & {} & {} & {} & {} & {} & {} & {} & {} & {} & {} & {} \\
MLP-CFR & {\cellcolor[HTML]{C2161B}} 1.11±0.06 & {\cellcolor[HTML]{F34935}} 0.56±0.37 & {\cellcolor[HTML]{CA181D}} 1.12±0.46 & {\cellcolor[HTML]{AB1016}} 1.94±0.95 & {\cellcolor[HTML]{EA362A}} 0.75±0.06 & {\cellcolor[HTML]{EA362A}} 0.74±0.14 & {\cellcolor[HTML]{FB7757}} 0.71±0.26 & {\cellcolor[HTML]{FDCBB6}} 0.68±0.31 & {\cellcolor[HTML]{EC382B}} 0.70±0.06 & {\cellcolor[HTML]{EC382B}} 0.74±0.12 & {\cellcolor[HTML]{EE3A2C}} 0.77±0.21 & {\cellcolor[HTML]{FB7252}} 0.76±0.29 \\
ND & {\cellcolor[HTML]{67000D}} 1.41±0.05 & {\cellcolor[HTML]{E93529}} 0.61±0.55 & {\cellcolor[HTML]{CA181D}} 1.12±0.49 & {\cellcolor[HTML]{9D0D14}} 2.01±0.89 & {\cellcolor[HTML]{67000D}} 1.15±0.06 & {\cellcolor[HTML]{820711}} 1.05±0.16 & {\cellcolor[HTML]{F14432}} 0.88±0.31 & {\cellcolor[HTML]{FDC5AE}} 0.73±0.44 & {\cellcolor[HTML]{67000D}} 1.08±0.07 & {\cellcolor[HTML]{7C0510}} 1.08±0.13 & {\cellcolor[HTML]{AA1016}} 1.01±0.23 & {\cellcolor[HTML]{F0402F}} 0.93±0.33 \\
MLP-CFR-INT & {\cellcolor[HTML]{CC191E}} 1.06±0.06 & {\cellcolor[HTML]{DE2B25}} 0.64±0.28 & {\cellcolor[HTML]{CA181D}} 1.12±0.21 & {\cellcolor[HTML]{C1161B}} 1.79±0.59 & {\cellcolor[HTML]{D21F20}} 0.84±0.06 & {\cellcolor[HTML]{CF1C1F}} 0.84±0.15 & {\cellcolor[HTML]{F96346}} 0.78±0.26 & {\cellcolor[HTML]{FCBFA7}} 0.77±0.29 & {\cellcolor[HTML]{D01D1F}} 0.80±0.06 & {\cellcolor[HTML]{D11E1F}} 0.84±0.12 & {\cellcolor[HTML]{D82422}} 0.85±0.21 & {\cellcolor[HTML]{F96346}} 0.81±0.29 \\
\end{tabular}
\end{adjustbox}
\end{table}

\begin{table}[!h]
\centering
\label{TblSemiCaneFL}
\begin{adjustbox}{max width=\textwidth}
\begin{tabular}{lllllllllllll}
{Metric} & \multicolumn{12}{l}{$\epsilon_{PEHE}$} \\
{Network} & \multicolumn{12}{l}{Watts Strogatz} \\
{Sparsity (k)} & \multicolumn{4}{r}{6} & \multicolumn{4}{r}{30} & \multicolumn{4}{r}{300} \\
{exposure} & {0} & {10} & {20} & {30} & {0} & {10} & {20} & {30} & {0} & {10} & {20} & {30} \\
{method} & {} & {} & {} & {} & {} & {} & {} & {} & {} & {} & {} & {} \\
MLP-CFR & {\cellcolor[HTML]{E32F27}} 0.79±0.06 & {\cellcolor[HTML]{F6563D}} 0.64±0.18 & {\cellcolor[HTML]{FCBDA4}} 0.50±0.29 & {\cellcolor[HTML]{FEE0D2}} 0.60±0.30 & {\cellcolor[HTML]{F0402F}} 0.69±0.05 & {\cellcolor[HTML]{F5523A}} 0.68±0.09 & {\cellcolor[HTML]{F6563D}} 0.71±0.15 & {\cellcolor[HTML]{F75C41}} 0.72±0.20 & {\cellcolor[HTML]{E63328}} 0.68±0.06 & {\cellcolor[HTML]{E63328}} 0.68±0.06 & {\cellcolor[HTML]{EB372A}} 0.69±0.07 & {\cellcolor[HTML]{E22E27}} 0.71±0.08 \\
ND & {\cellcolor[HTML]{67000D}} 1.17±0.06 & {\cellcolor[HTML]{B91419}} 0.89±0.20 & {\cellcolor[HTML]{FCB69B}} 0.53±0.39 & {\cellcolor[HTML]{FEE8DE}} 0.47±0.39 & {\cellcolor[HTML]{6B010E}} 1.09±0.04 & {\cellcolor[HTML]{8A0812}} 1.08±0.09 & {\cellcolor[HTML]{B01217}} 1.02±0.15 & {\cellcolor[HTML]{CE1A1E}} 0.94±0.22 & {\cellcolor[HTML]{E12D26}} 0.70±0.07 & {\cellcolor[HTML]{E12D26}} 0.70±0.07 & {\cellcolor[HTML]{E83429}} 0.70±0.08 & {\cellcolor[HTML]{E83429}} 0.69±0.09 \\
MLP-CFR-INT & {\cellcolor[HTML]{D11E1F}} 0.86±0.06 & {\cellcolor[HTML]{E43027}} 0.75±0.17 & {\cellcolor[HTML]{FCA486}} 0.60±0.27 & {\cellcolor[HTML]{FDCAB5}} 0.78±0.23 & {\cellcolor[HTML]{D72322}} 0.79±0.05 & {\cellcolor[HTML]{E43027}} 0.79±0.08 & {\cellcolor[HTML]{E93529}} 0.81±0.13 & {\cellcolor[HTML]{F14331}} 0.79±0.19 & {\cellcolor[HTML]{C9181D}} 0.78±0.06 & {\cellcolor[HTML]{C9181D}} 0.78±0.06 & {\cellcolor[HTML]{CE1A1E}} 0.79±0.07 & {\cellcolor[HTML]{C4161C}} 0.81±0.08 \\
\end{tabular}
\end{adjustbox}
\end{table}

\begin{table}[!h]
\centering
\label{TblSemiCaneBC}
\begin{adjustbox}{max width=\textwidth}
\begin{tabular}{lllllllllllllllll}
{Metric} & \multicolumn{16}{l}{$\epsilon_{PEHE}$} \\
{Network} & \multicolumn{16}{l}{BlogCatalog} \\
{Source} & \multicolumn{4}{r}{Combined} & \multicolumn{4}{r}{Mutual connections} & \multicolumn{4}{r}{Peer degree} & \multicolumn{4}{r}{Similarity} \\
{exposure} & {0} & {10} & {20} & {30} & {0} & {10} & {20} & {30} & {0} & {10} & {20} & {30} & {0} & {10} & {20} & {30} \\
{method} & {} & {} & {} & {} & {} & {} & {} & {} & {} & {} & {} & {} & {} & {} & {} & {} \\
MLP-CFR & {\cellcolor[HTML]{B21218}} 0.51±0.26 & {\cellcolor[HTML]{79040F}} 0.83±0.46 & {\cellcolor[HTML]{67000D}} 1.41±0.76 & {\cellcolor[HTML]{67000D}} 2.03±1.05 & {\cellcolor[HTML]{AC1117}} 0.53±0.25 & {\cellcolor[HTML]{9C0D14}} 0.62±0.26 & {\cellcolor[HTML]{940B13}} 0.69±0.30 & {\cellcolor[HTML]{FC9C7D}} 0.81±0.38 & {\cellcolor[HTML]{AC1117}} 0.53±0.25 & {\cellcolor[HTML]{79040F}} 0.87±0.37 & {\cellcolor[HTML]{67000D}} 1.50±0.61 & {\cellcolor[HTML]{67000D}} 2.18±0.83 & {\cellcolor[HTML]{AC1117}} 0.53±0.25 & {\cellcolor[HTML]{6D010E}} 0.84±0.42 & {\cellcolor[HTML]{67000D}} 1.52±0.71 & {\cellcolor[HTML]{67000D}} 2.18±1.02 \\
ND & {\cellcolor[HTML]{67000D}} 0.61±0.23 & {\cellcolor[HTML]{67000D}} 0.86±0.44 & {\cellcolor[HTML]{67000D}} 1.41±0.76 & {\cellcolor[HTML]{67000D}} 2.03±1.04 & {\cellcolor[HTML]{67000D}} 0.62±0.23 & {\cellcolor[HTML]{67000D}} 0.69±0.24 & {\cellcolor[HTML]{67000D}} 0.75±0.28 & {\cellcolor[HTML]{FC9777}} 0.84±0.35 & {\cellcolor[HTML]{67000D}} 0.62±0.23 & {\cellcolor[HTML]{67000D}} 0.90±0.35 & {\cellcolor[HTML]{67000D}} 1.50±0.60 & {\cellcolor[HTML]{67000D}} 2.18±0.81 & {\cellcolor[HTML]{67000D}} 0.62±0.23 & {\cellcolor[HTML]{67000D}} 0.85±0.42 & {\cellcolor[HTML]{7A0510}} 1.46±0.77 & {\cellcolor[HTML]{75030F}} 2.12±1.08 \\
\end{tabular}
\end{adjustbox}
\end{table}

\begin{table}[!h]
\centering
\label{TblSemiCaneFL}
\begin{adjustbox}{max width=\textwidth}
\begin{tabular}{lllllllllllllllll}
{Metric} & \multicolumn{16}{l}{$\epsilon_{PEHE}$} \\
{Network} & \multicolumn{16}{l}{Flickr} \\
{Source} & \multicolumn{4}{r}{Combined} & \multicolumn{4}{r}{Mutual connections} & \multicolumn{4}{r}{Peer degree} & \multicolumn{4}{r}{Similarity} \\
{exposure} & {0} & {10} & {20} & {30} & {0} & {10} & {20} & {30} & {0} & {10} & {20} & {30} & {0} & {10} & {20} & {30} \\
{method} & {} & {} & {} & {} & {} & {} & {} & {} & {} & {} & {} & {} & {} & {} & {} & {} \\
MLP-CFR & {\cellcolor[HTML]{FC8262}} 0.39±0.23 & {\cellcolor[HTML]{FC9C7D}} 0.67±0.36 & {\cellcolor[HTML]{F34935}} 1.18±0.57 & {\cellcolor[HTML]{D82422}} 1.73±0.76 & {\cellcolor[HTML]{FC8666}} 0.38±0.21 & {\cellcolor[HTML]{FC8A6A}} 0.40±0.21 & {\cellcolor[HTML]{FC8565}} 0.41±0.22 & {\cellcolor[HTML]{FDC7B2}} 0.41±0.22 & {\cellcolor[HTML]{FC8969}} 0.38±0.21 & {\cellcolor[HTML]{FCAA8D}} 0.55±0.31 & {\cellcolor[HTML]{FB7353}} 0.98±0.46 & {\cellcolor[HTML]{FA6547}} 1.46±0.59 & {\cellcolor[HTML]{FC8666}} 0.38±0.21 & {\cellcolor[HTML]{FCA78B}} 0.73±0.32 & {\cellcolor[HTML]{FCB296}} 1.38±0.58 & {\cellcolor[HTML]{FC9E80}} 1.93±0.90 \\
ND & {\cellcolor[HTML]{FB7151}} 0.44±0.21 & {\cellcolor[HTML]{FC9C7D}} 0.67±0.35 & {\cellcolor[HTML]{F34935}} 1.18±0.57 & {\cellcolor[HTML]{D82422}} 1.73±0.76 & {\cellcolor[HTML]{FB7555}} 0.43±0.20 & {\cellcolor[HTML]{FB7A5A}} 0.45±0.21 & {\cellcolor[HTML]{FB7757}} 0.45±0.21 & {\cellcolor[HTML]{FCBEA5}} 0.46±0.21 & {\cellcolor[HTML]{FB7757}} 0.43±0.20 & {\cellcolor[HTML]{FCA689}} 0.57±0.30 & {\cellcolor[HTML]{FB7353}} 0.98±0.45 & {\cellcolor[HTML]{FA6547}} 1.46±0.59 & {\cellcolor[HTML]{FB7555}} 0.43±0.20 & {\cellcolor[HTML]{FCBBA1}} 0.64±0.37 & {\cellcolor[HTML]{FCC1A8}} 1.21±0.67 & {\cellcolor[HTML]{FCAB8F}} 1.79±1.00 \\
\end{tabular}
\end{adjustbox}
\end{table}
\clearpage
\subsubsection{$\epsilon_{ATE}$ errors for \textbf{constant true individual direct effects} and heterogeneous peer influence (Fig \ref{fig:ane_ate}).}
\begin{table}[!h]
\centering
\label{TblSynAneBA}
\begin{adjustbox}{max width=\textwidth}
\begin{tabular}{lllllllllllll}
{Metric} & \multicolumn{12}{l}{$\epsilon_{ATE}$} \\
{Network} & \multicolumn{12}{l}{Barabasi Albert} \\
{Sparsity (p)} & \multicolumn{4}{r}{1} & \multicolumn{4}{r}{5} & \multicolumn{4}{r}{10} \\
{exposure} & {0} & {10} & {20} & {30} & {0} & {10} & {20} & {30} & {0} & {10} & {20} & {30} \\
{method} & {} & {} & {} & {} & {} & {} & {} & {} & {} & {} & {} & {} \\
MLP-CFR & {\cellcolor[HTML]{C4161C}} 1.10±0.06 & {\cellcolor[HTML]{E43027}} 0.47±0.44 & {\cellcolor[HTML]{79040F}} 1.08±0.49 & {\cellcolor[HTML]{820711}} 1.90±0.96 & {\cellcolor[HTML]{EC382B}} 0.74±0.06 & {\cellcolor[HTML]{E12D26}} 0.73±0.15 & {\cellcolor[HTML]{F6583E}} 0.64±0.29 & {\cellcolor[HTML]{FC8161}} 0.55±0.37 & {\cellcolor[HTML]{EC382B}} 0.70±0.06 & {\cellcolor[HTML]{E93529}} 0.72±0.13 & {\cellcolor[HTML]{E53228}} 0.73±0.21 & {\cellcolor[HTML]{F14130}} 0.69±0.31 \\
ND & {\cellcolor[HTML]{67000D}} 1.41±0.05 & {\cellcolor[HTML]{9D0D14}} 0.61±0.55 & {\cellcolor[HTML]{67000D}} 1.12±0.49 & {\cellcolor[HTML]{67000D}} 2.01±0.89 & {\cellcolor[HTML]{67000D}} 1.15±0.06 & {\cellcolor[HTML]{67000D}} 1.05±0.16 & {\cellcolor[HTML]{BE151A}} 0.88±0.31 & {\cellcolor[HTML]{F14331}} 0.73±0.43 & {\cellcolor[HTML]{67000D}} 1.08±0.07 & {\cellcolor[HTML]{67000D}} 1.08±0.13 & {\cellcolor[HTML]{840711}} 1.01±0.23 & {\cellcolor[HTML]{AC1117}} 0.93±0.33 \\
MLP-CFR-INT & {\cellcolor[HTML]{CF1C1F}} 1.05±0.06 & {\cellcolor[HTML]{F34C37}} 0.42±0.41 & {\cellcolor[HTML]{9C0D14}} 1.01±0.31 & {\cellcolor[HTML]{B01217}} 1.70±0.61 & {\cellcolor[HTML]{D52221}} 0.83±0.06 & {\cellcolor[HTML]{C7171C}} 0.82±0.15 & {\cellcolor[HTML]{F24734}} 0.68±0.31 & {\cellcolor[HTML]{FC8767}} 0.53±0.39 & {\cellcolor[HTML]{D32020}} 0.79±0.06 & {\cellcolor[HTML]{C9181D}} 0.83±0.12 & {\cellcolor[HTML]{CB181D}} 0.82±0.21 & {\cellcolor[HTML]{E93529}} 0.73±0.31 \\
\end{tabular}
\end{adjustbox}
\end{table}
\begin{table}[!h]
\centering
\label{TblSemiCaneFL}
\begin{adjustbox}{max width=\textwidth}
\begin{tabular}{lllllllllllll}
{Metric} & \multicolumn{12}{l}{$\epsilon_{ATE}$} \\
{Network} & \multicolumn{12}{l}{Watts Strogatz} \\
{Sparsity (k)} & \multicolumn{4}{r}{6} & \multicolumn{4}{r}{30} & \multicolumn{4}{r}{300} \\
{exposure} & {0} & {10} & {20} & {30} & {0} & {10} & {20} & {30} & {0} & {10} & {20} & {30} \\
{method} & {} & {} & {} & {} & {} & {} & {} & {} & {} & {} & {} & {} \\
MLP-CFR & {\cellcolor[HTML]{E53228}} 0.78±0.07 & {\cellcolor[HTML]{F4503A}} 0.61±0.19 & {\cellcolor[HTML]{FCAB8F}} 0.38±0.34 & {\cellcolor[HTML]{FC9B7C}} 0.47±0.33 & {\cellcolor[HTML]{F0402F}} 0.68±0.06 & {\cellcolor[HTML]{F5523A}} 0.67±0.10 & {\cellcolor[HTML]{F4503A}} 0.69±0.15 & {\cellcolor[HTML]{F75B40}} 0.67±0.22 & {\cellcolor[HTML]{E43027}} 0.67±0.06 & {\cellcolor[HTML]{E43027}} 0.67±0.07 & {\cellcolor[HTML]{E22E27}} 0.68±0.07 & {\cellcolor[HTML]{DE2B25}} 0.70±0.09 \\
ND & {\cellcolor[HTML]{67000D}} 1.17±0.06 & {\cellcolor[HTML]{A91016}} 0.89±0.20 & {\cellcolor[HTML]{FB7B5B}} 0.53±0.39 & {\cellcolor[HTML]{FC9B7C}} 0.47±0.40 & {\cellcolor[HTML]{67000D}} 1.09±0.04 & {\cellcolor[HTML]{820711}} 1.08±0.09 & {\cellcolor[HTML]{A50F15}} 1.02±0.15 & {\cellcolor[HTML]{BF151B}} 0.94±0.22 & {\cellcolor[HTML]{DB2824}} 0.70±0.07 & {\cellcolor[HTML]{DB2824}} 0.70±0.07 & {\cellcolor[HTML]{DE2B25}} 0.69±0.08 & {\cellcolor[HTML]{E22E27}} 0.69±0.09 \\
MLP-CFR-INT & {\cellcolor[HTML]{D32020}} 0.85±0.06 & {\cellcolor[HTML]{E12D26}} 0.71±0.18 & {\cellcolor[HTML]{FCAB8F}} 0.38±0.36 & {\cellcolor[HTML]{FC8E6E}} 0.51±0.31 & {\cellcolor[HTML]{D52221}} 0.79±0.05 & {\cellcolor[HTML]{E32F27}} 0.78±0.08 & {\cellcolor[HTML]{E22E27}} 0.80±0.13 & {\cellcolor[HTML]{EC382B}} 0.77±0.19 & {\cellcolor[HTML]{C2161B}} 0.78±0.06 & {\cellcolor[HTML]{C2161B}} 0.78±0.06 & {\cellcolor[HTML]{BF151B}} 0.79±0.07 & {\cellcolor[HTML]{BF151B}} 0.80±0.08 \\
\end{tabular}
\end{adjustbox}
\end{table}

\begin{table}[!h]
\centering
\label{TblSemiCaneBC}
\begin{adjustbox}{max width=\textwidth}
\begin{tabular}{lllllllllllllllll}
{Metric} & \multicolumn{16}{l}{$\epsilon_{PEHE}$} \\
{Network} & \multicolumn{16}{l}{BlogCatalog} \\
{Source} & \multicolumn{4}{r}{Combined} & \multicolumn{4}{r}{Mutual connections} & \multicolumn{4}{r}{Peer degree} & \multicolumn{4}{r}{Similarity} \\
{method} & {} & {} & {} & {} & {} & {} & {} & {} & {} & {} & {} & {} & {} & {} & {} & {} \\
MLP-CFR & {\cellcolor[HTML]{B11218}} 0.51±0.26 & {\cellcolor[HTML]{79040F}} 0.83±0.46 & {\cellcolor[HTML]{67000D}} 1.41±0.76 & {\cellcolor[HTML]{67000D}} 2.03±1.05 & {\cellcolor[HTML]{AB1016}} 0.53±0.25 & {\cellcolor[HTML]{9A0C14}} 0.62±0.26 & {\cellcolor[HTML]{900A12}} 0.69±0.30 & {\cellcolor[HTML]{AC1117}} 0.81±0.38 & {\cellcolor[HTML]{AB1016}} 0.53±0.25 & {\cellcolor[HTML]{79040F}} 0.87±0.37 & {\cellcolor[HTML]{67000D}} 1.50±0.61 & {\cellcolor[HTML]{67000D}} 2.18±0.83 & {\cellcolor[HTML]{AB1016}} 0.53±0.25 & {\cellcolor[HTML]{79040F}} 0.82±0.44 & {\cellcolor[HTML]{6D010E}} 1.44±0.79 & {\cellcolor[HTML]{69000D}} 2.11±1.11 \\
ND & {\cellcolor[HTML]{67000D}} 0.61±0.23 & {\cellcolor[HTML]{67000D}} 0.86±0.44 & {\cellcolor[HTML]{67000D}} 1.41±0.76 & {\cellcolor[HTML]{67000D}} 2.03±1.04 & {\cellcolor[HTML]{67000D}} 0.62±0.23 & {\cellcolor[HTML]{67000D}} 0.69±0.24 & {\cellcolor[HTML]{67000D}} 0.75±0.28 & {\cellcolor[HTML]{9F0E14}} 0.84±0.35 & {\cellcolor[HTML]{67000D}} 0.62±0.23 & {\cellcolor[HTML]{67000D}} 0.90±0.35 & {\cellcolor[HTML]{67000D}} 1.50±0.60 & {\cellcolor[HTML]{69000D}} 2.17±0.82 & {\cellcolor[HTML]{67000D}} 0.62±0.23 & {\cellcolor[HTML]{67000D}} 0.85±0.42 & {\cellcolor[HTML]{67000D}} 1.46±0.77 & {\cellcolor[HTML]{67000D}} 2.12±1.08 \\
\end{tabular}
\end{adjustbox}
\end{table}

\begin{table}[!h]
\centering
\label{TblSemiCaneFL}
\begin{adjustbox}{max width=\textwidth}
\begin{tabular}{lllllllllllllllll}
{Metric} & \multicolumn{16}{l}{$\epsilon_{ATE}$} \\
{Network} & \multicolumn{16}{l}{Flickr} \\
{Source} & \multicolumn{4}{r}{Combined} & \multicolumn{4}{r}{Mutual connections} & \multicolumn{4}{r}{Peer degree} & \multicolumn{4}{r}{Similarity} \\
{exposure} & {0} & {10} & {20} & {30} & {0} & {10} & {20} & {30} & {0} & {10} & {20} & {30} & {0} & {10} & {20} & {30} \\
{method} & {} & {} & {} & {} & {} & {} & {} & {} & {} & {} & {} & {} & {} & {} & {} & {} \\
MLP-CFR & {\cellcolor[HTML]{D01D1F}} 0.39±0.23 & {\cellcolor[HTML]{D42121}} 0.67±0.36 & {\cellcolor[HTML]{67000D}} 1.18±0.57 & {\cellcolor[HTML]{67000D}} 1.73±0.76 & {\cellcolor[HTML]{D52221}} 0.38±0.21 & {\cellcolor[HTML]{B91419}} 0.39±0.21 & {\cellcolor[HTML]{9F0E14}} 0.40±0.21 & {\cellcolor[HTML]{F75C41}} 0.41±0.22 & {\cellcolor[HTML]{D52221}} 0.38±0.21 & {\cellcolor[HTML]{AB1016}} 0.55±0.31 & {\cellcolor[HTML]{67000D}} 0.98±0.46 & {\cellcolor[HTML]{67000D}} 1.46±0.59 & {\cellcolor[HTML]{D52221}} 0.38±0.21 & {\cellcolor[HTML]{77040F}} 0.62±0.35 & {\cellcolor[HTML]{B61319}} 1.16±0.67 & {\cellcolor[HTML]{840711}} 1.69±1.03 \\
ND & {\cellcolor[HTML]{B31218}} 0.44±0.21 & {\cellcolor[HTML]{D42121}} 0.67±0.35 & {\cellcolor[HTML]{67000D}} 1.18±0.57 & {\cellcolor[HTML]{67000D}} 1.73±0.76 & {\cellcolor[HTML]{B91419}} 0.43±0.20 & {\cellcolor[HTML]{860811}} 0.45±0.21 & {\cellcolor[HTML]{67000D}} 0.45±0.21 & {\cellcolor[HTML]{F14130}} 0.46±0.21 & {\cellcolor[HTML]{B91419}} 0.43±0.20 & {\cellcolor[HTML]{9F0E14}} 0.57±0.30 & {\cellcolor[HTML]{67000D}} 0.98±0.45 & {\cellcolor[HTML]{67000D}} 1.46±0.59 & {\cellcolor[HTML]{B91419}} 0.43±0.20 & {\cellcolor[HTML]{67000D}} 0.64±0.37 & {\cellcolor[HTML]{AA1016}} 1.21±0.67 & {\cellcolor[HTML]{67000D}} 1.79±1.00 \\
\end{tabular}
\end{adjustbox}
\end{table}

\clearpage

\subsubsection{$\epsilon_{ATE}$ errors for individual direct effects with {effect modification} and heterogeneous peer influence (Fig \ref{fig:cane_ate}).}
\begin{table}[!h]
\centering
\label{TblSynAneBA}
\begin{adjustbox}{max width=\textwidth}
\begin{tabular}{lllllllllllll}
{Metric} & \multicolumn{12}{l}{$\epsilon_{ATE}$} \\
{Network} & \multicolumn{12}{l}{Barabasi Albert} \\
{Sparsity (p)} & \multicolumn{4}{r}{1} & \multicolumn{4}{r}{5} & \multicolumn{4}{r}{10} \\{exposure} & {0} & {10} & {20} & {30} & {0} & {10} & {20} & {30} & {0} & {10} & {20} & {30} \\
{method} & {} & {} & {} & {} & {} & {} & {} & {} & {} & {} & {} & {} \\
MLP-CFR & {\cellcolor[HTML]{BF151B}} 1.16±0.07 & {\cellcolor[HTML]{C7171C}} 0.53±0.49 & {\cellcolor[HTML]{67000D}} 1.14±0.45 & {\cellcolor[HTML]{7C0510}} 1.94±0.94 & {\cellcolor[HTML]{D21F20}} 0.84±0.08 & {\cellcolor[HTML]{D21F20}} 0.78±0.16 & {\cellcolor[HTML]{F44D38}} 0.68±0.29 & {\cellcolor[HTML]{FB7656}} 0.59±0.38 & {\cellcolor[HTML]{D21F20}} 0.80±0.09 & {\cellcolor[HTML]{D72322}} 0.79±0.10 & {\cellcolor[HTML]{DC2924}} 0.77±0.18 & {\cellcolor[HTML]{EB372A}} 0.72±0.27 \\
ND & {\cellcolor[HTML]{67000D}} 1.46±0.13 & {\cellcolor[HTML]{A50F15}} 0.60±0.56 & {\cellcolor[HTML]{6B010E}} 1.13±0.50 & {\cellcolor[HTML]{67000D}} 2.03±0.87 & {\cellcolor[HTML]{67000D}} 1.15±0.08 & {\cellcolor[HTML]{6B010E}} 1.04±0.17 & {\cellcolor[HTML]{C5171C}} 0.87±0.32 & {\cellcolor[HTML]{F34935}} 0.72±0.44 & {\cellcolor[HTML]{6B010E}} 1.08±0.10 & {\cellcolor[HTML]{75030F}} 1.06±0.12 & {\cellcolor[HTML]{8E0912}} 1.00±0.21 & {\cellcolor[HTML]{AF1117}} 0.92±0.31 \\
MLP-CFR-INT & {\cellcolor[HTML]{B81419}} 1.19±0.16 & {\cellcolor[HTML]{DE2B25}} 0.48±0.46 & {\cellcolor[HTML]{BB141A}} 0.93±0.42 & {\cellcolor[HTML]{BF151B}} 1.62±0.60 & {\cellcolor[HTML]{B51318}} 0.95±0.13 & {\cellcolor[HTML]{B71319}} 0.87±0.20 & {\cellcolor[HTML]{F03F2E}} 0.72±0.32 & {\cellcolor[HTML]{FC8262}} 0.55±0.43 & {\cellcolor[HTML]{B61319}} 0.90±0.13 & {\cellcolor[HTML]{B91419}} 0.89±0.13 & {\cellcolor[HTML]{C4161C}} 0.85±0.20 & {\cellcolor[HTML]{DC2924}} 0.77±0.28 \\
\end{tabular}
\end{adjustbox}
\end{table}

\begin{table}[!h]
\centering
\label{TblSemiCaneFL}
\begin{adjustbox}{max width=\textwidth}
\begin{tabular}{lllllllllllll}
{Metric} & \multicolumn{12}{l}{$\epsilon_{ATE}$} \\
{Network} & \multicolumn{12}{l}{Watts Strogatz} \\
{Sparsity (k)} & \multicolumn{4}{r}{6} & \multicolumn{4}{r}{30} & \multicolumn{4}{r}{300} \\{exposure} & {0} & {10} & {20} & {30} & {0} & {10} & {20} & {30} & {0} & {10} & {20} & {30} \\
{method} & {} & {} & {} & {} & {} & {} & {} & {} & {} & {} & {} & {} \\
MLP-CFR & {\cellcolor[HTML]{CA181D}} 0.88±0.07 & {\cellcolor[HTML]{F0402F}} 0.65±0.19 & {\cellcolor[HTML]{FCA98C}} 0.38±0.35 & {\cellcolor[HTML]{FC9E80}} 0.45±0.34 & {\cellcolor[HTML]{D01D1F}} 0.80±0.08 & {\cellcolor[HTML]{EB372A}} 0.76±0.08 & {\cellcolor[HTML]{F24633}} 0.73±0.13 & {\cellcolor[HTML]{F6583E}} 0.68±0.19 & {\cellcolor[HTML]{A50F15}} 0.78±0.09 & {\cellcolor[HTML]{A60F15}} 0.78±0.08 & {\cellcolor[HTML]{A60F15}} 0.79±0.08 & {\cellcolor[HTML]{A60F15}} 0.79±0.08 \\
ND & {\cellcolor[HTML]{67000D}} 1.16±0.11 & {\cellcolor[HTML]{AB1016}} 0.88±0.20 & {\cellcolor[HTML]{FB7B5B}} 0.52±0.38 & {\cellcolor[HTML]{FC9E80}} 0.45±0.39 & {\cellcolor[HTML]{6B010E}} 1.07±0.12 & {\cellcolor[HTML]{9A0C14}} 1.04±0.09 & {\cellcolor[HTML]{A91016}} 1.02±0.16 & {\cellcolor[HTML]{BC141A}} 0.95±0.22 & {\cellcolor[HTML]{820711}} 0.84±0.20 & {\cellcolor[HTML]{840711}} 0.84±0.20 & {\cellcolor[HTML]{8A0812}} 0.84±0.20 & {\cellcolor[HTML]{8A0812}} 0.84±0.20 \\
MLP-CFR-INT & {\cellcolor[HTML]{B51318}} 0.96±0.14 & {\cellcolor[HTML]{DB2824}} 0.73±0.21 & {\cellcolor[HTML]{FC9E80}} 0.41±0.36 & {\cellcolor[HTML]{FC9576}} 0.48±0.35 & {\cellcolor[HTML]{B31218}} 0.90±0.12 & {\cellcolor[HTML]{CA181D}} 0.88±0.12 & {\cellcolor[HTML]{D82422}} 0.85±0.13 & {\cellcolor[HTML]{E22E27}} 0.81±0.18 & {\cellcolor[HTML]{67000D}} 0.89±0.13 & {\cellcolor[HTML]{67000D}} 0.89±0.13 & {\cellcolor[HTML]{67000D}} 0.90±0.12 & {\cellcolor[HTML]{67000D}} 0.90±0.11 \\
\end{tabular}
\end{adjustbox}
\end{table}

\begin{table}[!h]
\centering
\label{TblSemiCaneBC}
\begin{adjustbox}{max width=\textwidth}
\begin{tabular}{lllllllllllllllll}
{Metric} & \multicolumn{16}{l}{$\epsilon_{ATE}$} \\
{Network} & \multicolumn{16}{l}{BlogCatalog} \\
{Source} & \multicolumn{4}{r}{Combined} & \multicolumn{4}{r}{Mutual connections} & \multicolumn{4}{r}{Peer degree} & \multicolumn{4}{r}{Similarity} \\
{exposure} & {0} & {10} & {20} & {30} & {0} & {10} & {20} & {30} & {0} & {10} & {20} & {30} & {0} & {10} & {20} & {30} \\
{method} & {} & {} & {} & {} & {} & {} & {} & {} & {} & {} & {} & {} & {} & {} & {} & {} \\
MLP-CFR & {\cellcolor[HTML]{B21218}} 0.51±0.30 & {\cellcolor[HTML]{7E0610}} 0.84±0.46 & {\cellcolor[HTML]{67000D}} 1.41±0.74 & {\cellcolor[HTML]{67000D}} 2.03±1.03 & {\cellcolor[HTML]{B21218}} 0.51±0.30 & {\cellcolor[HTML]{8C0912}} 0.66±0.32 & {\cellcolor[HTML]{67000D}} 0.97±0.41 & {\cellcolor[HTML]{67000D}} 1.36±0.53 & {\cellcolor[HTML]{B21218}} 0.51±0.30 & {\cellcolor[HTML]{79040F}} 0.85±0.47 & {\cellcolor[HTML]{69000D}} 1.42±0.77 & {\cellcolor[HTML]{67000D}} 2.05±1.06 & {\cellcolor[HTML]{B21218}} 0.51±0.30 & {\cellcolor[HTML]{7C0510}} 0.91±0.45 & {\cellcolor[HTML]{69000D}} 1.62±0.78 & {\cellcolor[HTML]{67000D}} 2.36±1.13 \\
ND & {\cellcolor[HTML]{67000D}} 0.61±0.28 & {\cellcolor[HTML]{67000D}} 0.88±0.44 & {\cellcolor[HTML]{67000D}} 1.41±0.74 & {\cellcolor[HTML]{69000D}} 2.02±1.03 & {\cellcolor[HTML]{67000D}} 0.61±0.28 & {\cellcolor[HTML]{67000D}} 0.71±0.31 & {\cellcolor[HTML]{67000D}} 0.97±0.40 & {\cellcolor[HTML]{67000D}} 1.36±0.53 & {\cellcolor[HTML]{67000D}} 0.61±0.28 & {\cellcolor[HTML]{67000D}} 0.88±0.46 & {\cellcolor[HTML]{67000D}} 1.43±0.76 & {\cellcolor[HTML]{67000D}} 2.05±1.06 & {\cellcolor[HTML]{67000D}} 0.61±0.28 & {\cellcolor[HTML]{67000D}} 0.95±0.44 & {\cellcolor[HTML]{67000D}} 1.63±0.79 & {\cellcolor[HTML]{69000D}} 2.35±1.10 \\
\end{tabular}
\end{adjustbox}
\end{table}

\begin{table}[!h]
\centering
\label{TblSemiCaneFL}
\begin{adjustbox}{max width=\textwidth}
\begin{tabular}{lllllllllllllllll}
{Metric} & \multicolumn{16}{l}{$\epsilon_{ATE}$} \\
{Network} & \multicolumn{16}{l}{Flickr} \\
{Source} & \multicolumn{4}{r}{Combined} & \multicolumn{4}{r}{Mutual connections} & \multicolumn{4}{r}{Peer degree} & \multicolumn{4}{r}{Similarity} \\
{exposure} & {0} & {10} & {20} & {30} & {0} & {10} & {20} & {30} & {0} & {10} & {20} & {30} & {0} & {10} & {20} & {30} \\
{method} & {} & {} & {} & {} & {} & {} & {} & {} & {} & {} & {} & {} & {} & {} & {} & {} \\
MLP-CFR & {\cellcolor[HTML]{BE151A}} 0.33±0.27 & {\cellcolor[HTML]{77040F}} 0.66±0.37 & {\cellcolor[HTML]{67000D}} 1.17±0.56 & {\cellcolor[HTML]{67000D}} 1.72±0.75 & {\cellcolor[HTML]{BE151A}} 0.33±0.27 & {\cellcolor[HTML]{F14331}} 0.36±0.28 & {\cellcolor[HTML]{F0402F}} 0.42±0.31 & {\cellcolor[HTML]{AB1016}} 0.51±0.37 & {\cellcolor[HTML]{BE151A}} 0.33±0.27 & {\cellcolor[HTML]{75030F}} 0.67±0.38 & {\cellcolor[HTML]{67000D}} 1.20±0.57 & {\cellcolor[HTML]{67000D}} 1.76±0.77 & {\cellcolor[HTML]{BE151A}} 0.33±0.27 & {\cellcolor[HTML]{880811}} 0.63±0.37 & {\cellcolor[HTML]{8E0912}} 1.15±0.65 & {\cellcolor[HTML]{940B13}} 1.67±0.95 \\
ND & {\cellcolor[HTML]{73030F}} 0.40±0.25 & {\cellcolor[HTML]{67000D}} 0.68±0.36 & {\cellcolor[HTML]{67000D}} 1.17±0.56 & {\cellcolor[HTML]{69000D}} 1.71±0.75 & {\cellcolor[HTML]{73030F}} 0.40±0.25 & {\cellcolor[HTML]{D42121}} 0.42±0.27 & {\cellcolor[HTML]{C8171C}} 0.51±0.32 & {\cellcolor[HTML]{67000D}} 0.59±0.37 & {\cellcolor[HTML]{73030F}} 0.40±0.25 & {\cellcolor[HTML]{67000D}} 0.69±0.37 & {\cellcolor[HTML]{67000D}} 1.20±0.57 & {\cellcolor[HTML]{69000D}} 1.75±0.77 & {\cellcolor[HTML]{73030F}} 0.40±0.25 & {\cellcolor[HTML]{67000D}} 0.67±0.38 & {\cellcolor[HTML]{67000D}} 1.24±0.62 & {\cellcolor[HTML]{67000D}} 1.83±0.87 \\
\end{tabular}
\end{adjustbox}
\end{table}

\end{document}